\documentclass[11pt,letterpaper]{article}
\pdfoutput=1

\usepackage{hyperref}
\usepackage{fullpage}
\usepackage{amssymb}
\usepackage{mathtools}
\usepackage{amsthm}
\usepackage{eucal}
\usepackage{dsfont}
\usepackage[T1]{fontenc}
\usepackage{lmodern}
\usepackage[stretch=10,shrink=10]{microtype}
\usepackage[capitalise]{cleveref}
\usepackage{color,soul}
\usepackage[linesnumbered,ruled,vlined]{algorithm2e}
\usepackage{graphicx}
\usepackage{overpic}
\usepackage{bm}

\usepackage{amsmath}
\usepackage{amsfonts}

\allowdisplaybreaks[1]

\def\A{\CMcal{A}}
\def\B{\CMcal{B}}

\def\H{\CMcal{H}}

\def\N{\CMcal{N}}

\def\P{\CMcal{P}}

\theoremstyle{plain}
\newtheorem{theorem}{Theorem}[section]
\newtheorem{lemma}[theorem]{Lemma}
\newtheorem{prop}[theorem]{Proposition}
\newtheorem{cor}[theorem]{Corollary}

\theoremstyle{definition}
\newtheorem{definition}[theorem]{Definition}

\newtheorem{remark}[theorem]{Remark}

\newenvironment{proofof}[1]{\noindent{\emph{Proof of #1:}}}{\qed}

\newcommand {\minusspace} {\: \! \!}

\newcommand {\fn} [2] {\ensuremath{ #1 \minusspace \br{ #2 } }}
\newcommand {\Fn} [2] {\ensuremath{ #1 \minusspace \Br{ #2 } }}

\newcommand {\set} [1] {\ensuremath{ \left\lbrace #1 \right\rbrace }}
\newcommand {\ceil} [1] {\ensuremath{ \left\lceil #1 \right\rceil }}

\newcommand{\eps}{\epsilon}

\newcommand {\br} [1] {\ensuremath{ \left( #1 \right) }}
\newcommand {\Br} [1] {\ensuremath{ \left[ #1 \right] }}

\newcommand {\abs} [1] {\ensuremath{ \left| #1 \right| }}

\newcommand {\bra} [1] {\ensuremath{ \left\langle #1 \right| }}
\newcommand {\ket} [1] {\ensuremath{ \left| #1 \right\rangle }}
\newcommand {\ketbratwo} [2] {\ensuremath{ \left| #1 \middle\rangle \middle\langle #2 \right| }}
\newcommand {\ketbra} [1] {\ketbratwo{#1}{#1}}

\newcommand{\stringx}{\texttt{X}}
\newcommand{\stringy}{\texttt{Y}}
\newcommand{\strings}{\texttt{S}}

\newcommand {\defeq} {\ensuremath{ \stackrel{\mathrm{def}}{=} }}

\newcommand {\prob} [1] {\Fn{\Pr}{#1}}

\newcommand {\Tr} {\ensuremath{ \mathrm{Tr} }}
\newcommand {\partrace} [2] {\fn{\Tr_{#1}}{#2}}

\newcommand {\id} {\ensuremath{\mathds{1}}}

\SetKwComment{Comment}{$\rhd\;$}{}
\SetKwInput{Input}{Input}
\SetKwInput{Output}{Output}
\SetKwInput{Promise}{Promise}

\newcommand {\email} [1] {\href{mailto:#1}{\texttt{#1}}}

\newcommand {\mytitle} {Capacity Approaching Coding for Low Noise Interactive Quantum Communication \\ Part I: Large Alphabets}

\newcommand{\suppress}[1]{}

\newcommand {\authorblock} [3] {
	\begin{minipage}[t]{0.3\linewidth}
		\centering
		{#1}\\[0.8ex]
		{\footnotesize {#2}\\[-0.7ex]
		\email{#3}}
	\end{minipage}\vspace{1ex}
}

\newcommand{\rI}{{\mathrm I}}

\newcommand{\kb}[2]{| #1\rangle\!\langle #2 |}
\newcommand{\rX}{{\mathrm X}}
\newcommand{\rZ}{{\mathrm Z}}
\newcommand{\rF}{{\mathrm F}}
\newcommand{\complexi}{{\mathrm i}}

\def\sA{{\mathsf A}}
\def\sB{{\mathsf B}}

\hypersetup{
	pdfstartview={FitH},
	pdfdisplaydoctitle={true},
	breaklinks={true},
	bookmarksopen={true},
	bookmarksnumbered={false},
	pdftitle={\mytitle},
	}

\newcommand {\Debbie}  {Debbie Leung}
\newcommand {\Ashwin}  {Ashwin Nayak}
\newcommand {\Ala}  {Ala Shayeghi}
\newcommand {\Dave}   {Dave Touchette}
\newcommand {\Penghui} {Penghui Yao}
\newcommand {\Nengkun} {Nengkun Yu}

\newcommand {\CNO} {Department of Combinatorics and Optimization}
\newcommand {\IQC} {Institute for Quantum Computing}
\newcommand {\UW} {University of Waterloo}
\newcommand {\PI} {Perimeter Institute}

\newcommand{\CQSI}{Centre for Quantum Software and Information,
   Faculty of Engineering and Information Technology, University of
   Technology Sydney}
\newcommand{\NJU}{Nanjing University}
\newcommand{\SKL}{State Key Laboratory for Novel Software Technology}
\newcommand{\US}{Universit\'e de Sherbrooke}

\newcommand{\CS}{Department of Computer Science}

\begin{document}

\begin{titlepage}
\title{\textbf{\mytitle}\\[2ex]}

\author{
 \authorblock{\Debbie}{\CNO, and \IQC, \UW, and \PI}{wcleung@uwaterloo.ca}
 \authorblock{\Ashwin}{\CNO, and \IQC, \UW}{anayak@uwaterloo.ca}
 \authorblock{\Ala}{\CNO, and \IQC, \UW}{ashayeghi@uwaterloo.ca}\\
 \authorblock{\Dave}{\CS, and Institut Quantique, \US, and \IQC, \UW, and \PI}{dave.touchette@usherbrooke.ca}
 \authorblock{\Penghui}{\SKL, \NJU}{pyao@nju.edu.cn}
 \authorblock{\Nengkun}{\CQSI}{nengkunyu@gmail.com}
}

\clearpage\maketitle
\thispagestyle{empty}

\abstract{
We consider the problem of implementing two-party interactive quantum
communication over noisy channels, a necessary endeavor if we wish to
fully reap quantum advantages for communication.
For an arbitrary protocol with $n$ messages, designed for a
\emph{noiseless\/} qudit channel over a $\mathrm{poly}\br{n}$ size alphabet, 
our main result is a simulation method that fails with probability less than
$2^{-\Theta(n\epsilon)}$ and uses a qudit channel over the same alphabet $n(1 + \Theta
(\sqrt{\epsilon}))$ times, of which an $\epsilon$ fraction can be
corrupted adversarially.
The simulation is thus capacity achieving to leading order, and
we conjecture that it is optimal up to a constant factor in
the $\sqrt{\epsilon}$ term.
Furthermore, the simulation is in a model that does not require
pre-shared resources such as randomness or entanglement between the
communicating parties.
\suppress{Surprisingly, this outperforms the best known overhead of $1 +
O(\sqrt{\epsilon \log \log 1/\epsilon})$ in the corresponding
\emph{classical} model, which is also conjectured to be optimal [Haeupler, FOCS'14].
}
Our work improves over the best previously known quantum result
where the overhead is a non-explicit large constant [Brassard \emph{et
    al.}, FOCS'14] for low $\epsilon$.
}

\end{titlepage}

\thispagestyle{empty}

\tableofcontents

\newpage

\setcounter{page}{1}

\newcommand{\longpaper}[1]{#1}
\newcommand{\blurb}[1]{}


\newcommand{\NewMetaA}{\mathit{NewMetaA}}
\newcommand{\NewMetaAtilde}{\widetilde{\NewMetaA}}
\newcommand{\MA}{\mathit{MA}}
\newcommand{\MAone}{\MA_1}
\newcommand{\MAtwo}{\MA_2}
\newcommand{\mponeMA}{mp1_{\MA}}
\newcommand{\mptwoMA}{mp2_{\MA}}
\newcommand{\MAtilde}{\widetilde{\MA}}
\newcommand{\MAtildeone}{\MAtilde_1}
\newcommand{\MAtildetwo}{\MAtilde_2}
\newcommand{\mponeMAtilde}{mp1_{\MAtilde}}
\newcommand{\mptwoMAtilde}{mp2_{\MAtilde}}
\newcommand{\FullMA}{\mathit{FullMA}}
\newcommand{\ellMA}{\ell_{\MA}}
\newcommand{\ellMAtilde}{\ell_{\MAtilde}}
\newcommand{\HMA}{H_{\MA}}
\newcommand{\HMAone}{H_{\MAone}}
\newcommand{\HMAtwo}{H_{\MAtwo}}
\newcommand{\vMAone}{v_{\MAone}}
\newcommand{\vMAtwo}{v_{\MAtwo}}
\newcommand{\HMAtilde}{H_{\MAtilde}}
\newcommand{\HMAtildeone}{H_{\MAtildeone}}
\newcommand{\HMAtildetwo}{H_{\MAtildetwo}}
\newcommand{\vMAtildeone}{v_{\MAtildeone}}
\newcommand{\vMAtildetwo}{v_{\MAtildetwo}}
\newcommand{\kMA}{k_{\MA}}
\newcommand{\khatMA}{\hat{k}_{\MA}}
\newcommand{\kMAtilde}{k_{\MAtilde}}
\newcommand{\khatMAtilde}{\hat{k}_{\MAtilde}}
\newcommand{\EMA}{E_{\MA}}
\newcommand{\EMAtilde}{E_{\MAtilde}}

\newcommand{\NewMetaB}{\mathit{NewMetaB}}
\newcommand{\NewMetaBtilde}{\widetilde{\NewMetaB}}
\newcommand{\MB}{\mathit{MB}}
\newcommand{\MBone}{\MB_1}
\newcommand{\MBtwo}{\MB_2}
\newcommand{\mponeMB}{mp1_{\MB}}
\newcommand{\mptwoMB}{mp2_{\MB}}
\newcommand{\MBtilde}{\widetilde{\MB}}
\newcommand{\MBtildeone}{\MBtilde_1}
\newcommand{\MBtildetwo}{\MBtilde_2}
\newcommand{\mponeMBtilde}{mp1_{\MBtilde}}
\newcommand{\mptwoMBtilde}{mp2_{\MBtilde}}
\newcommand{\FullMB}{\mathit{FullMB}}
\newcommand{\ellMB}{\ell_{\MB}}
\newcommand{\ellMBtilde}{\ell_{\MBtilde}}
\newcommand{\HMB}{H_{\MB}}
\newcommand{\HMBone}{H_{\MBone}}
\newcommand{\HMBtwo}{H_{\MBtwo}}
\newcommand{\vMBone}{v_{\MBone}}
\newcommand{\vMBtwo}{v_{\MBtwo}}
\newcommand{\HMBtilde}{H_{\MBtilde}}
\newcommand{\HMBtildeone}{H_{\MBtildeone}}
\newcommand{\HMBtildetwo}{H_{\MBtildetwo}}
\newcommand{\vMBtildeone}{v_{\MBtildeone}}
\newcommand{\vMBtildetwo}{v_{\MBtildetwo}}
\newcommand{\kMB}{k_{\MB}}
\newcommand{\khatMB}{\hat{k}_{\MB}}
\newcommand{\kMBtilde}{k_{\MBtilde}}
\newcommand{\khatMBtilde}{\hat{k}_{\MBtilde}}
\newcommand{\EMB}{E_{\MB}}
\newcommand{\EMBtilde}{E_{\MBtilde}}

\newcommand{\mdAplus}{md_\mathrm{A}^+}
\newcommand{\mdAminus}{md_\mathrm{A}^-}
\newcommand{\mdBplus}{md_\mathrm{B}^+}
\newcommand{\mdBminus}{md_\mathrm{B}^-}

\newcommand{\NewPauliA}{\mathit{NewPauliA}}
\newcommand{\NewPauliAtilde}{\widetilde{\NewPauliA}}
\newcommand{\PA}{\mathit{PA}}
\newcommand{\PAone}{\PA_1}
\newcommand{\PAtwo}{\PA_2}
\newcommand{\mponePA}{mp1_{\PA}}
\newcommand{\mptwoPA}{mp2_{\PA}}
\newcommand{\PAtilde}{\widetilde{\PA}}
\newcommand{\PAtildeone}{\PAtilde_1}
\newcommand{\PAtildetwo}{\PAtilde_2}
\newcommand{\mponePAtilde}{mp1_{\PAtilde}}
\newcommand{\mptwoPAtilde}{mp2_{\PAtilde}}
\newcommand{\FullPA}{\mathit{FullPA}}
\newcommand{\ellPA}{\ell_{\PA}}
\newcommand{\ellPAtilde}{\ell_{\PAtilde}}
\newcommand{\HPA}{H_{\PA}}
\newcommand{\HPAone}{H_{\PAone}}
\newcommand{\HPAtwo}{H_{\PAtwo}}
\newcommand{\vPAone}{v_{\PAone}}
\newcommand{\vPAtwo}{v_{\PAtwo}}
\newcommand{\HPAtilde}{H_{\PAtilde}}
\newcommand{\HPAtildeone}{H_{\PAtildeone}}
\newcommand{\HPAtildetwo}{H_{\PAtildetwo}}
\newcommand{\vPAtildeone}{v_{\PAtildeone}}
\newcommand{\vPAtildetwo}{v_{\PAtildetwo}}
\newcommand{\kPA}{k_{\PA}}
\newcommand{\khatPA}{\hat{k}_{\PA}}
\newcommand{\kPAtilde}{k_{\PAtilde}}
\newcommand{\khatPAtilde}{\hat{k}_{\PAtilde}}
\newcommand{\EPA}{E_{\PA}}
\newcommand{\EPAtilde}{E_{\PAtilde}}

\newcommand{\NewPauliB}{\mathit{NewPauliB}}
\newcommand{\NewPauliBtilde}{\widetilde{\NewPauliB}}
\newcommand{\PB}{\mathit{PB}}
\newcommand{\PBone}{\PB_1}
\newcommand{\PBtwo}{\PB_2}
\newcommand{\mponePB}{mp1_{\PB}}
\newcommand{\mptwoPB}{mp2_{\PB}}
\newcommand{\PBtilde}{\widetilde{\PB}}
\newcommand{\PBtildeone}{\PBtilde_1}
\newcommand{\PBtildetwo}{\PBtilde_2}
\newcommand{\mponePBtilde}{mp1_{\PBtilde}}
\newcommand{\mptwoPBtilde}{mp2_{\PBtilde}}
\newcommand{\FullPB}{\mathit{FullPB}}
\newcommand{\ellPB}{\ell_{\PB}}
\newcommand{\ellPBtilde}{\ell_{\PBtilde}}
\newcommand{\HPB}{H_{\PB}}
\newcommand{\HPBone}{H_{\PBone}}
\newcommand{\HPBtwo}{H_{\PBtwo}}
\newcommand{\vPBone}{v_{\PBone}}
\newcommand{\vPBtwo}{v_{\PBtwo}}
\newcommand{\HPBtilde}{H_{\PBtilde}}
\newcommand{\HPBtildeone}{H_{\PBtildeone}}
\newcommand{\HPBtildetwo}{H_{\PBtildetwo}}
\newcommand{\vPBtildeone}{v_{\PBtildeone}}
\newcommand{\vPBtildetwo}{v_{\PBtildetwo}}
\newcommand{\kPB}{k_{\PB}}
\newcommand{\khatPB}{\hat{k}_{\PB}}
\newcommand{\kPBtilde}{k_{\PBtilde}}
\newcommand{\khatPBtilde}{\hat{k}_{\PBtilde}}
\newcommand{\EPB}{E_{\PB}}
\newcommand{\EPBtilde}{E_{\PBtilde}}

\newcommand{\pdAplus}{pd_\mathrm{A}^+}
\newcommand{\pdAminus}{pd_\mathrm{A}^-}
\newcommand{\pdBplus}{pd_\mathrm{B}^+}
\newcommand{\pdBminus}{pd_\mathrm{B}^-}

\newcommand{\RA}{\mathit{RA}}
\newcommand{\RAtilde}{\widetilde{\RA}}
\newcommand{\ellRA}{\ell_\RA}
\newcommand{\ellRAtilde}{\widetilde{\ellRA}}
\newcommand{\FullRA}{\mathit{FullRA}}
\newcommand{\RB}{\mathit{RB}}
\newcommand{\RBtilde}{\widetilde{\RB}}
\newcommand{\ellRB}{\ell_\RB}
\newcommand{\ellRBtilde}{\widetilde{\ellRB}}
\newcommand{\FullRB}{\mathit{FullRB}}


\newcommand{\Rtotal}{R_\mathrm{total}}
\newcommand{\Itertype}{\mathit{Itertype}}
\newcommand{\RewindExtend}{\mathit{RewindExtend}}	
\newcommand{\IndexA}{\mathit{IndexA}}
\newcommand{\IndexB}{\mathit{IndexB}}
\newcommand{\NextMESIndexA}{\mathit{NextMESIndexA}}
\newcommand{\NextMESIndexB}{\mathit{NextMESIndexB}}
\newcommand{\ellNextMESA}{\ell_\mathit{NextMESA}}
\newcommand{\ellNextMESB}{\ell_\mathit{NextMESB}}
\newcommand{\ellQVCA}{\ell_\mathrm{QVC}^\mathrm{A}}        
\newcommand{\ellQVCAtilde}{\widetilde{\ellQVCA}}
\newcommand{\ellQVCB}{\ell_\mathrm{QVC}^\mathrm{B}}
\newcommand{\ellQVCBtilde}{\widetilde{\ellQVCB}}
\newcommand{\QHA}{\mathit{QHA}}                            
\newcommand{\QHB}{\mathit{QHB}}
\newcommand{\LQVC}{\mathit{L}_{\mathrm{QVC}}}
\newcommand{\msg}{\mathit{msg}}
\newcommand{\PCorr}{\mathit{P}_\mathrm{Corr}}
\newcommand{\PCorrtilde}{\widetilde{\mathit{P}_\mathrm{Corr}}}
\newcommand{\JSone}{\mathit{JS}1}
\newcommand{\JStwo}{\mathit{JS}2}
\newcommand{\JSoneA}{\mathit{JS}1^\mathrm{A}}             
\newcommand{\JStwoA}{\mathit{JS}2^\mathrm{A}}             
\newcommand{\JSoneB}{\mathit{JS}1^\mathrm{B}}
\newcommand{\JStwoB}{\mathit{JS}2^\mathrm{B}}

\newcommand{\MD}{\mathrm{MD}}
\newcommand{\RD}{\mathrm{RD}}
\newcommand{\PD}{\mathrm{PD}}
\newcommand{\QH}{\mathrm{QH}}
\newcommand{\MES}{\mathrm{MES}}
\newcommand{\SIM}{\mathrm{SIM}}

\newcommand{\sR}{\mathsf{R}}                      
\newcommand{\sE}{\mathsf{E}}                      

\newcommand{\sC}{\mathsf{C}}                      
\newcommand{\OED}{\mathsf{0_{ED}}}                

\newcommand{\sS}{\mathsf{S}}					  
\newcommand{\sM}{\mathsf{M}}					  

\let\oldnl\nl 
\newcommand{\nonl}{\renewcommand{\nl}{\let\nl\oldnl}} 

\newcommand{\control}[1] {\textrm{c-}#1}      




\longpaper{\section{Introduction}}

\longpaper{\subsection{Motivation}}

\subsubsection{The main questions.}

Quantum communication offers the possibility of distributed
computation with extraordinary \emph{provable\/} savings in
communication as compared with classical communication (see, e.g.,
\cite{Regev:2011} and the references therein). Most often, if not
always, the savings are achieved by protocols that assume access to
\emph{noiseless\/} communication channels.  In practice, though,
imperfection in channels is inevitable. Is it possible to make the
protocols robust to noise while maintaining the advantages offered by
quantum communication?  If so, what is the cost of making the
protocols robust, and how much noise can be tolerated?
In this article, we address these questions in the context of quantum
communication protocols involving two parties, in the low noise
regime. Following convention, we call the two parties Alice and Bob.




\subsubsection{ Channel coding theory as a special case.}
\label{sec:channel-coding}

In the
special case when the communication is one-way (say, from Alice to
Bob), techniques for making the message noise-tolerant, via error
correcting codes, have been studied for a long time.  Coding allows us
to simulate a noiseless communication protocol using a noisy channel,
under certain assumptions about the noise process (such as having a
memoryless channel).
Typically, such simulation is possible when the \emph{error rate\/}
(the fraction of the messages corrupted)
is lower than a certain threshold.
A desirable goal is to also maximize the
\emph{communication rate\/} (also called the \emph{information rate\/}),
which is the length of the original
message, as a fraction of the length of its encoding.
In the classical setting, Shannon established
the capacity (i.e., the optimal communication rate) of
\emph{arbitrarily accurate\/} transmission,
in the limit of \emph{asymptotically large\/} number of channel uses,
through the Noisy Coding Theorem~\cite{Shannon48a}.
Since then, researchers have discovered many explicit codes with desirable
properties such as good rate, and efficient encoding and decoding
procedures (see, for example,~\cite{Stolte:2002,Arikan:2009}).
%
%
Analogous results have been developed over the
past two decades in the quantum setting.
In particular, capacity expressions for a quantum channel
transmitting classical data~\cite{Hol98,SW97} or
quantum data~\cite{Lloyd97,Shor02,Dev05} have been derived.
Even though it is not known
how we may evaluate these capacity expressions for a general quantum
channel, useful error correcting codes have been developed for many
channels of interest (see, for example,~\cite{CS96-goodcode,CRSS97,BDSW96,Bombin15}). Remarkably,
quantum effects give rise to surprising phenomena without classical
counterparts, including
\emph{superadditivity\/}~\cite{DSS98,Hastings2009},
and \emph{superactivation\/}~\cite{SY09}. All of these highlight the
non-trivial nature of coding for noisy quantum channels.


\subsubsection{Communication complexity as a special case.}

In general two-party protocols, data are transmitted in each direction
alternately, potentially over a number of rounds.  In a computation
problem, the number of rounds may grow as a function of the input
size. Such protocols are at the core of several important areas including
distributed computation, cryptography, interactive proof systems, and
communication complexity.
For example, in the case of the Disjointness function, a canonical
task in the two-party communication model, an $n$-bit input is given
to each party, who jointly compute the function with as little
communication as possible. The optimal quantum protocol for this task
consists of~$\Theta\br{\sqrt{n\,}}$ rounds of communication, each with a
constant length message~\cite{BCW98, HdW:2002, aaronson:2003}, and such a high
level of interaction has been shown to be necessary~\cite{KNTZ07, JainRS:09, BGK+15}.
Furthermore, quantum communication leads to provable advantages over
the classical setting, without any complexity-theoretic assumptions.
For example, some specially crafted problems (see, for example,~\cite{Raz99,Regev:2011})
exhibit exponential quantum advantages, and others display the power
of quantum interaction by showing that just one additional
round can sometimes lead to exponential savings~\cite{KNTZ07}.
%


\subsubsection{The problem, and motivation for the investigation.}
\label{sec:problem}

In this paper, we consider two-party interactive communication
protocols using \emph{noisy\/} communication.
%
%
%
The goal is to effectively implement an interactive communication protocol to
arbitrary accuracy despite noise in the available channels.
We want to minimize the number of uses of the noisy channel,
and the complexity of the coding operations.
%
%
The motivation is two-fold and applies to both the classical and the quantum
setting.
First, this problem is a natural generalization of channel coding from
the 1-way to the 2-way setting, with the ``capacity'' being the best ratio
of the number of channel uses in the original protocol divided by that
needed in the noisy implementation.  Here, we consider the combined
number of channel uses in both
directions.
Note that this scenario is different from ``assisted
  capacities'' where some auxiliary noiseless resources such as a
  classical side channel for quantum transmission are given to the
  parties for free.
%
%
Second, we would like to generalize interactive protocols to the noisy
communication regime.  If an interactive protocol can be implemented
using noisy channels while preserving the complexity, then the
corresponding communication complexity results become robust against
channel noise.
%
In particular, an important motivation is
to investigate whether the quantum advantage in interactive
communication protocols is robust against quantum noise.
%
Due to the ubiquitous nature of quantum noise and fragility of quantum data,
noise-resilience is of fundamental importance for the realization of
quantum communication networks.
The coding problem for interactive quantum communication was first studied in~\cite{BNTTU14}. In
Section~\ref{sec:intro-prior},
we elaborate on this work and the questions that arise from it.

\longpaper{\subsection{Fundamental difficulties in coding for quantum interactive
  communication}\label{sec:intro-diff}}


For some natural problems the optimal interactive protocols require a
lot of interaction.  For example, distributed quantum search over $n$
items~\cite{BCW98, HdW:2002, aaronson:2003} requires $\Theta
\br{\sqrt{n}}$ rounds of constant-sized messages~\cite{KNTZ07,JainRS:09,
BGK+15}.  How can we implement such highly interactive
protocols over noisy channels?  What are the major obstacles?

\subsubsection{Standard error correcting codes are inapplicable.}
\label{sec:ecc-inapp}

In both the
classical and quantum settings, standard error correcting codes are inapplicable.
%
To see this, first suppose we encode each message separately.  Then
the corruption of even a single encoded message can already derail the
rest of the protocol. Thus, for the entire protocol to be simulated
with high fidelity, we need to reduce the decoding error for each message to be
inversely proportional to the length of the protocol, say $n$.
For constant size messages,
the
overhead of coding then grows with the problem size $n$, increasing
the complexity and suppressing the rate of simulation to $0$ as $n$ increases.
The situation is even worse with adversarial errors: the adversary can
invest the entire error budget to corrupt the shortest
critical message, and it is impossible to tolerate an
error rate above $\approx 1$/number of rounds, no matter what the rate of
communication is.
To circumvent this barrier, one must employ a coding strategy acting collectively over
many messages. However, most of these are generated dynamically during the
protocol and are unknown to the sender earlier.
Furthermore, error correction or detection
may require communication between the parties, which is also corruptible.
The problem is thus reminiscent of fault-tolerant computation in that the
steps needed to implement error correction are themselves subject to errors.

\subsubsection{The no-cloning quantum problem.}
A fundamental property of quantum mechanics is that learning about an
unknown quantum state from a given specimen disturbs the
state~\cite{BBJMPSW94}. In particular, an unknown
\longpaper{quantum} state cannot be cloned~\cite{Dieks82, WZ82}. This affects our
problem in two fundamental ways. First, any logical quantum data leaked
into the environment due to the noisy channel cannot be recovered
by the communicating parties. Second, the parties hold a joint
quantum state that evolves with the protocol, but they cannot make
copies
\longpaper{of the joint state}
without corrupting it.


\longpaper{\subsection{Prior classical and quantum work }\label{sec:intro-prior}}

\blurb{\vspace*{1.5ex}{\bf \large 3. Prior classical and quantum work}\vspace*{0.5ex}}


Despite the difficulties in coding for interactive communication,
many interesting results have been discovered over the last 25 years, with a
notable extension in the quantum setting.

\subsubsection{Classical results showing positive rates.}
Schulman first raised the question of simulating noiseless interactive
communication protocols using noisy channels in the classical
setting~\cite{Sch92,Sch93,Sch96}.
He developed \emph{tree codes\/} to work with messages that are
determined one at a time, and generated dynamically during the course
of the interaction.  These codes have constant overhead, and the
capacity is thus a positive constant.
Furthermore, these codes protect data against
adversarial noise that corrupts up to a $\frac{1}{240}$ fraction of the channel
uses.
This tolerable noise rate was improved by subsequent work, culminating
to the results by Braverman and Rao~\cite{BR11}. They showed
that $<\frac{1}{4}$ adversarial errors can be tolerated provided one
can use large constant alphabet sizes and that this bound on noise rate is optimal.

\subsubsection{Classical results with efficient encoding and decoding.}
The aforementioned coding schemes are not known to be computationally
efficient, as they are built on tree codes; the computational
complexity of encoding and decoding tree codes is unknown.
Other computationally efficient encoding schemes have been
developed~\cite{BK12,BN13,BrakerskiKN:2014,GMS12,GellesMS:2014,GH14}.
%
%
The communication rates under various scenarios
have also been studied
\cite{BravermanK:2017,GhaffariHS:2014,EfremenkoGH:2015,FranklinGOS:2015}.
%
However, the rates do not approach the capacity expected of the noise rate.


\subsubsection{Classical results with optimal rates.}
Kol and Raz~\cite{KR13} first established coding with
rate approaching $1$ as the noise parameter goes to $0$, for the binary
symmetric channel.
Haeupler~\cite{Haeupler:2014} extended the above result to
adversarial binary channels corrupting at most an $\epsilon$ fraction
of the symbols, with communication rate $1 - O\,\br{\!\sqrt{\epsilon \log
  \log \br{\frac{1}{\epsilon}}\!}\,}$, which is conjectured to be optimal.
For oblivious adversaries, this increases to $1 - O(\sqrt{\epsilon})$.
Further studies of capacity have been conducted,
for example, in~\cite{HaeuplerV:2017,BenYishaiSK:2017}.
For further details about recent results on interactive coding, see the extensive survey by Gelles~\cite{Gelles17}.


\subsubsection{Quantum results showing positive rates.}
All coding for classical interactive protocols relies on ``backtracking'': if an error
is detected, the parties go back to an earlier stage of the protocol
and resume from there.
Backtracking is impossible in the quantum setting due to the no
cloning principle described in the previous subsection.  There is no generic
way to make copies of the quantum state at earlier stages
without restarting the protocol.
Brassard, Nayak, Tapp, Touchette, and Unger~\cite{BNTTU14} provided the first coding scheme with constant
overhead by using two ideas.
The first idea is to teleport each quantum message. This splits
the quantum data into a protected quantum share and an unprotected
classical share that is transmitted through the noisy channels
using tree codes.
Second, backtracking is replaced by \emph{reversing\/} of steps
to \emph{return\/} to a desirable earlier stage; i.e., the joint quantum state
is evolved back to that of an earlier stage, which circumvents the
no-cloning theorem. This is possible since local operations can be
made unitary, and communication can be reversed (up to more noise).  Together, a positive
simulation rate (or constant overhead) can be achieved.  In the noisy analogue to the
Cleve-Buhrman communication model where entanglement is free, error rate $<\frac{1}{2}$ can
be tolerated.  In the noisy analogue to the Yao (plain) model, a noisy quantum channel with
one-way quantum capacity $Q > 0$ can be used to simulate an
$n$-message protocol given $O \br{\frac{1}{Q} n}$ uses.  However, the
rate can be suboptimal and the coding complexity is unknown due to the
use of tree codes.  The rate is further reduced by a large constant in
order to match the quantum and classical data in teleportation, and in
coordinating the action of the parties (advancing or reversing the
protocol).



\longpaper{\subsection{Results in this paper, overview of techniques, and our contributions}}

\blurb{\vspace*{1.5ex}{\bf \large 4. Results in this paper, overview of techniques,
and our contributions}\vspace*{0.5ex}}

Inspired by the recent results on rate
optimal coding for the classical setting~\cite{KR13,Haeupler:2014} and the rate suboptimal coding
in the quantum setting~\cite{BNTTU14}, a fundamental question is: can
we likewise avoid the loss of communication rate for
interactive \emph{quantum\/} protocols?
In particular, is it possible to protect quantum data without pre-shared
free entanglement, and if we have to generate it at a cost, can we
still achieve rate approaching $1$ as the error rate vanishes?
Further, can erroneous steps be reversed with noisy
resources, and with negligible overhead as the error rate vanishes?
What is the complexity of rate optimal protocols, if
one exists?
Are there other new obstacles?

\suppress{Our main result
~addresses all these questions.}
To address all these questions, in this paper we start by studying a simpler setting where the input protocol $\Pi$ and the noisy communication channel operate on the same communication alphabet of polynomial size in the length of $\Pi$. This simplifies the algorithm while still capturing the main challenges we need to address. The analysis is easier to follow and shares the same outline and structure with our main result, namely simulation of noiseless interactive communication over constant-size alphabets, which we will present in an upcoming paper. The framework we develop in this work, sets the stage for a smooth transition to the small alphabet case. We focus on alternating protocols, in which Alice and Bob exchange qudits back and forth in alternation. Our main result in this paper is the following: 


\begin{theorem}
    Consider any alternating communication protocol $\Pi$ in the plain quantum model, communicating $n$ messages over a noiseless channel with an alphabet $\Sigma$ of bit-size $\Theta\br{\log n}$. We provide a simulation protocol $\Pi'$ which given $\Pi$, simulates it with probability at least $1-2^{-\Theta\br{n\epsilon}}$, over any fully adversarial error quantum channel with alphabet $\Sigma$ and error rate $\epsilon$. The simulation uses $n\br{1+\Theta\br{\sqrt{\epsilon}}}$ rounds of communication, and therefore achieves a communication rate of $1-\Theta\br{\sqrt{\epsilon}}$. 
\end{theorem}

Our rate optimal protocol requires a careful
combination of ideas to overcome various obstacles. Some of these
ideas are well-established, some are not so well known, some require significant modifications, and some are new.
A priori, it is not clear whether these previously developed tools would be useful in the context of the problem.
For the clarity of presentation, we first introduce our main ideas in a simpler communication model, where Alice and Bob have access to free entanglement and communicate over a fully adversarial error classical channel. We introduce several key ideas while developing a basic solution to approach the optimal rate in this scenario. Inspired by~\cite{BNTTU14}, we use teleportation to protect the communication and the simulation is actively rewound whenever an error is detected. We develop a framework which allows the two parties to obtain a global view of the simulation by locally maintaining a classical data structure. We adapt ideas due to Haeupler~\cite{Haeupler:2014} to efficiently update this data structure over the noisy channel and evolve the simulation. Then, we extend these ideas to the plain model of quantum communication with large alphabet size. In the plain quantum model, Alice and Bob communicate over a fully adversarial error quantum channel and do not have access to any pre-shared resources such as entanglement or shared randomness. As a result any such resources need to be established through extra communication. This in particular makes it more challenging to achieve a high communication rate in this setting. Surprisingly, an adaptation of an old technique called the Quantum Vernam Cipher (QVC)~\cite{Leung:2002} turns out to be the perfect method to protect quantum data in our application. QVC allows the two parties to recycle and reuse entanglement as needed throughout the simulation. Building on the ideas introduced in the teleportation-based protocol, one of our main contributions in this model is developing a mechanism to reliably recycle entanglement in a communication efficient way.

\newpage
\section{Preliminaries}

We assume that the reader is familiar with the quantum formalism for
finite dimensional systems; for a thorough treatment, we refer
the interested reader to good introductions in a quantum information
theory context \cite[Chapter 2]{NC00}, \cite[Chapter 2]{Wat08} \cite[
  Chapters 3, 4, 5]{Wilde11}.
  
Let $\A$ be a~$d$-dimensional Hilbert space with computational basis~$\set{\ket{0},\ldots, \ket{d-1}}$. Let $\rX$ and $\rZ$ be the operators such that~$\rX\ket{k}\defeq\ket{k+1}$ and~$\rZ\ket{k}\defeq e^{\complexi \cdot 2\pi\frac{k}{d}}\ket{k}$. The generalized Pauli operators, also known as the Heisenberg-Weyl operators, are defined as~$\set{\rX^j\rZ^k}_{0\leq j,k\leq d-1}$. Let~$\Sigma=\set{0,\ldots,d-1}$. For~$N\in \mathbb{N}$, the operators in 
\begin{equation}\label{eqn:set-of-Pauli-ops}
    \P_{d,N} \defeq \{ \rX^{j_1} \rZ^{k_1} \otimes \cdots \otimes \rX^{j_N} \rZ^{k_N} \}_{j_l k_l \in \Sigma^2,l\in \Br{N}}
\end{equation}
form a basis for the space of operators on $\A^{\otimes N}$. For~$E\in \P_{d,N}$, We denote by~$\mathrm{wt}\br{E}$ the weight of~$E$, i.e., the number of~$\A$ subsystems on which~$E$ acts non-trivially.
For $j,k\in\Sigma$, we represent the single qudit Pauli error $\rX^j\rZ^k$ by the string $jk\in\Sigma^2$. Similarly, a Pauli error on multiple qudits is represented by a string in ${\br{\Sigma^2}}^*$. The Fourier transform operator $F$ is defined to be the operator such that~$F\ket{j}\defeq\frac{1}{\sqrt{d}}\sum_{k=0}^{d-1}e^{\complexi\cdot 2\pi \frac{jk}{d}}\ket{k}$.

\begin{prop}\label{fac:paulioperatorcommute}
	Let $\set{\rX^j\rZ^k}_{0\leq j,k\leq d-1}$ be the set of generalized Pauli operators on a $d$-dimensional Hilbert space. It holds that $\rX^j\rZ^k=e^{-\complexi \cdot 2\pi\frac{jk}{d}}\rZ^k\rX^j, F\rX^jF^{\dagger}=\rZ^j$ and $F\rZ^jF^{\dagger}=\rX^j$ for every $j,k\in\set{\ket{0},\ldots, \ket{d-1}}$.
\end{prop}
\begin{definition}\label{def:Bellstates}
	Let $\A,\B$ be $d$-dimensional Hilbert spaces with computational bases~$\set{\ket{i}_A}_{0\leq i\leq d-1}$ and~$\set{\ket{i}_B}_{0\leq i\leq d-1}$, respectively. The set of Bell states in $\A\otimes\B$ is defined as
	\[\set{\ket{\phi^{j,k}}_{AB}\defeq\br{\rX_A^j\rZ_A^k\otimes\id}\ket{\phi}_{AB}:0\leq j, k\leq d-1} \enspace,\]
	where $\ket{\phi}_{AB}\defeq\frac{1}{\sqrt{d}}\sum_{i=0}^{d-1}\ket{i}_A\ket{i}_B$.
For~$j,k\in\mathbb{\rZ}$, we define $\ket{\phi^{j,k}}\defeq\ket{\phi^{j~\textsf{mod}~d, k~\textsf{mod}~d}}$.
\end{definition}

It is easy to see that $\ket{\phi^{j,k}}=\frac{1}{\sqrt{d}}\sum_{t=0}^{d-1}e^{\complexi \cdot 2\pi\frac{tk}{d}}\ket{t+j, t}$.

\begin{prop}\label{fac:bellbasis}
	The Bell states $\set{\ket{\phi^{j,k}}}_{0\leq j,k\leq d-1}$ form an orthonormal basis in $\A\otimes\B$.
\end{prop}
\begin{prop}\label{fac:uuepr}
	For any unitary operator $U$ on register $A$, it holds that $$(U\otimes\id)\ket{\phi}_{AB}=(\id\otimes U^T)\ket{\phi}_{AB}\enspace,$$
	where $U^T=\sum_{j,k}\bra{j}U\ket{k}\ketbratwo{k}{j}$. In particular, $(F\otimes\id)\ket{\phi}_{AB}=(\id\otimes F)\ket{\phi}_{AB}$.
\end{prop}

\subsection{Quantum Communication Model} 
	\label{sec:qucomm}

The definitions for the noiseless and noisy quantum communication models are copied from Ref.~\cite{BNTTU14}. We refer the reader there for a more formal definition of the noisy quantum communication model, as well as the relationship of the noiseless quantum communication model to well-studied quantum communication complexity models such a Yao's model and the Cleve-Buhrman model.

\subsubsection{Noiseless Communication Model}
	\label{sec:nslss}

In the \emph{noiseless quantum communication model\/} that we want to  
simulate, there are five quantum registers: the $A$ register held by 
Alice, the $B$ register held by Bob, the $C$ register, which is the 
communication register exchanged back-and-forth between Alice and Bob and 
initially held by Alice, the $E$ register held by a potential adversary Eve, and finally the $R$ register, a reference system which purifies the 
state of the $ABCE$ registers throughout the protocol. The initial state 
$\ket{\psi_\mathrm{init}}^{ABCER} \in \H (A \otimes B \otimes C \otimes E \otimes R)$ is chosen arbitrarily 
from the set of possible inputs, and is fixed at the outset 
of the protocol, but possibly unknown (totally or partially) to Alice and 
Bob. Note that to allow for composition of quantum protocols in an 
arbitrary environment, we consider arbitrary quantum states as input, 
which may be entangled with systems $RE$. A protocol $\Pi$ is then 
defined by the sequence of unitary operations $U_1, U_2, \cdots , U_{n + 
1}$, with $U_{i}$ for odd~$i$ known at least to Alice (or given to her 
in a black box) and acting on registers~$AC$, and $U_{i}$ for even~$i$ 
known at least to Bob (or given to him in a black box) and acting on
registers~$BC$. For simplicity, we assume that $n$ is even. We can
modify any protocol to satisfy this property, while increasing
the total cost of communication by at most one communication of the $C$ 
register. 
The unitary operators of protocol $\Pi$ can be assumed to be public information, known to Eve.
On a particular input state $\ket{\psi_\mathrm{init}}$, the 
protocol generates the final state $\ket{\psi_\mathrm{final}}^{ABCER} = 
U_{n + 1} \cdots U_1 \ket{\psi_\mathrm{init}}^{ABCER}$, for which at the 
end of the protocol the $A$ and $C$ registers are held by Alice, the $B$ 
register is held by Bob, and the $E$ register is held by Eve. The reference register $R$ is left untouched throughout the protocol.
The output of the protocol resides in systems $ABC$, i.e.,  $\Pi (\ket{\psi_\mathrm{init}}) =
\partrace{ER}{\kb{\psi_\mathrm{final}}{\psi_\mathrm{final}}^{ABCER}}$, and by  a slight
abuse of notation we also represent the induced quantum channel from $ABCE$ 
to $ABC$ simply by $\Pi$. This is depicted in Figure~\ref{fig:int_mod}.
Note that while the protocol only acts on $ABC$, we wish to maintain correlations with the reference system $R$, while we simply disregard
what happens on the $E$ system assumed to be in Eve's hand.
Since we consider local computation to be free, 
the sizes of $A$ and $B$ can be arbitrarily large, but still of finite 
size, say $m_A$ and $m_B$ qubits, respectively. 
Since we are interested in high communication rate, we do not want to restrict ourselves to
the case of a single-qubit communication register~$C$, since converting a general protocol  to one of this form can incur a factor of two overhead. We thus consider alternating protocols in which the register $C$ is of fixed size, say $d$ dimensions, and is exchanged back-and-forth. We believe that non-alternating protocols can also be simulated by adapting our techniques, but we leave this extension to future work.
Note that both the Yao and the
Cleve-Buhrman models of quantum communication complexity can be recast 
in this framework; see Ref.~\cite{BNTTU14}. 

\suppress{
		\begin{figure}
		\begin{overpic}[width=1\textwidth]{int_noiseless_prot_no_label.pdf}
		  \put(0,66.5){Reference}
		  \put(0,45){Alice}
		  \put(0,22){Bob}
		  \put(0,4.8){Eve}
		  \put(6,31){ $\ket{\psi_{\mathrm{init}}}$}
		  \put(15,66.5){\footnotesize$R$}
		  \put(15,4.8){\footnotesize$E$}
		  \put(15,54){\footnotesize$A$}
		  \put(15,48){\footnotesize$C$}
		  \put(15,13.7){\footnotesize$B$}
		  \put(22.1,49.5){\footnotesize$U_1$}
		  \put(26.2,54){\footnotesize$A$}
		  \put(26.2,48){\footnotesize$C$}
		  \put(33,15.5){\footnotesize$U_2$}
		  \put(37.2,17.4){\footnotesize$C$}
		  \put(37.2,13.7){\footnotesize$B$}
		  \put(44.2,49.5){\footnotesize$U_3$}
		  \put(48.2,54){\footnotesize$A$}
		  \put(48.2,48){\footnotesize$C$}
		  \put(55,33){\footnotesize$\cdots$}
		  \put(59.5,54){\footnotesize$A$}
		  \put(59.5,48){\footnotesize$C$}
		  \put(59.5,13.7){\footnotesize$B$}
		  \put(66.8,15.5){\footnotesize$U_{N}$}
		  \put(73.5,23){\footnotesize$C$}
		  \put(73.5,13.7){\footnotesize$B$}
		  \put(75.9,49.5){\footnotesize$U_{N+1}$}
		  \put(81.3,54){\footnotesize$A$}
		  \put(81.3,48){\footnotesize$C$}
		  \put(92,31){ $\ket{\psi_{\mathrm{final}}}$}
		\end{overpic}
		  \caption{Depiction of a quantum protocol in the noiseless communication model, adapted from the long version of~\cite[Figure 1]{Tou15}.}
		  \label{fig:int_mod}
		\end{figure}
}

\begin{figure}[ht]
	\begin{center}
\includegraphics[width=440pt]{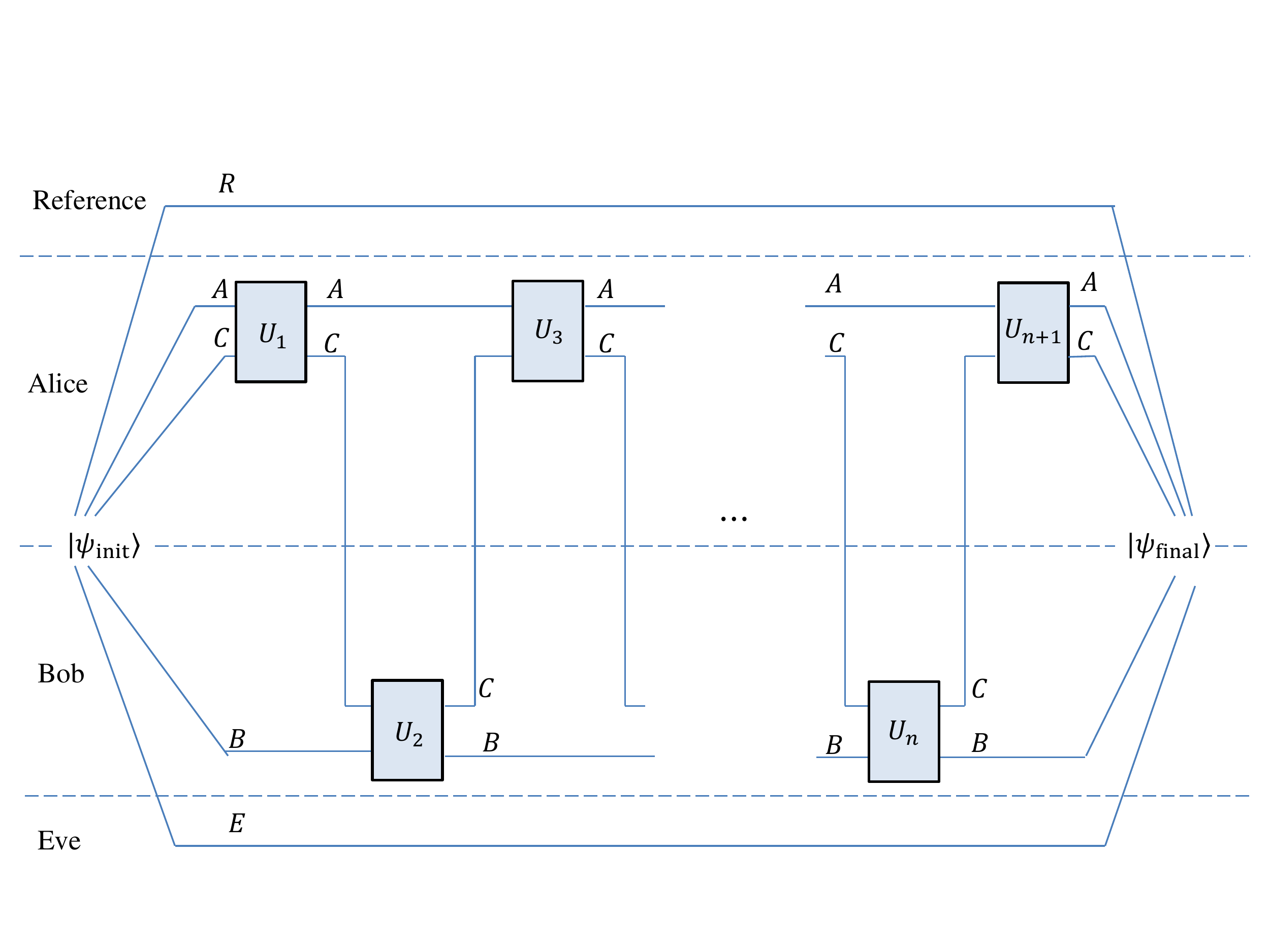}
		\vspace*{2ex}
		\caption{Depiction of a quantum protocol in the noiseless communication model.}
		\label{fig:int_mod}
	\end{center}
\end{figure}

We later embed length $n$ protocols into others of larger length 
$n^\prime > n$. To perform such \emph{noiseless protocol embedding\/}, we 
define some dummy registers $\tilde{A}$, $\tilde{B}$, $\tilde{C}$ isomorphic 
to $A$, $B$, $C$, respectively. $\tilde{A}$ and $\tilde{C}$  are part of 
Alice's scratch register and $\tilde{B}$ is part of Bob's scratch 
register. Then, for any isomorphic quantum registers $D$, $\tilde{D}$, let 
SWAP$_{D \leftrightarrow \tilde{D}}$ denote the unitary operation that 
swaps the $D, \tilde{D}$ registers. Recall that~$n$ is assumed to be
even. In a noiseless protocol embedding, 
for $i \in \{1, 2, \cdots n-1 \}$, we leave $U_i$ untouched. We replace 
$U_{n}$ by (SWAP$_{B \leftrightarrow \tilde{B}} U_{n})$ and $U_{n + 
1}$ by (SWAP$_{AC \leftrightarrow \tilde{A} \tilde{C}} U_{n + 1})$. 
Finally, for $i \in \{n+2, n+3, \cdots n^\prime + 1 \}$, we define $U_i = 
\rI$, the identity operator.
This embedding is important in the setting of interactive quantum coding for the following reasons.
First, adding these $U_i$ for $i > n$ makes the protocol well defined for $n^\prime +1$ steps.
Then,  swapping the important registers into the safe registers $\tilde{A}$, $\tilde{B}$, $\tilde{C}$
ensures that the important registers are never affected by noise arising after the first $n+1$ steps have been applied.
Hence, in our simulation, as long as we succeed in implementing the first $n+1$ steps without errors, the simulation will succeed since the
$\tilde{A}$, $\tilde{B}$, $\tilde{C}$ registers will then contain the output of the simulation, with no error acting on these registers.

\subsubsection{Noisy Communication Model} \label{sec:noisy_comm_model}

There are many possible models for noisy communication. For our main results, we focus on one in particular, 
analogous to the Yao model with no shared entanglement 
but noisy quantum communication, which we call the \emph{plain quantum model\/}.
In Section~\ref{sec:BKlarge}, we consider and define an alternative model.

For simplicity, we formally define in this 
section what we sometimes refer to as \emph{alternating\/} communication models, 
in which Alice and Bob take turns in transmitting the 
communication register to each other, and this is the model in which most of our protocols are 
defined. Our definitions easily adapt to somewhat more general models 
which we call \emph{oblivious\/} communication models, following Ref.~\cite{BR11}.
In these models, Alice and Bob do not necessarily 
transmit their messages in alternation, but nevertheless in a fixed order and
of fixed sizes
known to all (Alice, Bob and Eve) depending only on the round, and not on 
the particular input or the actions of Eve. 
Communication models with a dependence on inputs or 
actions of Eve are called \emph{adaptive\/} communication models.

\paragraph{Plain Quantum Model}


In the \emph{plain quantum model\/}, Alice has workspace $A^\prime$, Bob has workspace $B^\prime$, the adversary Eve has workspace $E^\prime$, and there is some quantum communication register $C^\prime$ of some fixed size $d^\prime$ dimensions (we will consider only $d^\prime = d$ in this work), exchanged back and forth between them $n^\prime$ times, passing through Eve's hand each time. Alice and Bob can perform arbitrary local processing between each transmission, whereas Eve's processing when the $C^\prime$ register passes through her hand is limited by the noise model as described below. The input registers $ABCE$ are shared between Alice ($AC$), Bob ($B$) and Eve ($E$) and the output registers $\tilde{A} \tilde{B} \tilde{C}$ are shared between Alice ($\tilde{A} \tilde{C}$) and Bob ($\tilde{B}$). The reference register $R$ containing the purification of the input is left untouched throughout. Alice and Bob also possess registers $C_\sA$ and $C_\sB$, respectively, acting as virtual communication register $C$ from the original protocol $\Pi$ of length $n$ to be simulated. The communication rate of the simulation is given by the ratio $\frac{n \log d}{n^\prime \log d^\prime}$.

We are interested in two models of errors, adversarial and random noise. In the \emph{adversarial\/} noise model, we are mainly interested in an adversary Eve with a bound $\delta n^\prime$ on the number of errors that she introduces on the quantum communication register $C^\prime$ that passes through her hand. The fraction $\delta$ of corrupted transmissions is called the error rate. More formally, an adversary in the quantum model with error rate bounded by~$\delta\in \Br{0,1}$ is specified by a sequence of instruments~$\N_1^{E^\prime C_1^\prime},\ldots,\N_{n^\prime}^{E^\prime C_{n^\prime}^\prime}$ acting on register $E^\prime$ of arbitrary dimension $d''$ and the communication register $C^\prime$ of dimension $d^\prime$ in protocols of length $n^\prime$. For any density operator~$\rho$ on~$\H( E^\prime \otimes C^{\prime \otimes {n^{\prime}}} )$, the action of such an adversary is 
\begin{equation}\label{eqn:noise-model-1}
	\N_1^{E^\prime C_1^\prime}\cdots\N_{n^\prime}^{E^\prime C_{n^\prime}^\prime} \br{\rho} = \sum_{i} G_i \rho G_i^\dagger \enspace,
\end{equation}
for~$i$ ranging over some finite set, subject to~$\sum_i G_i^\dagger G_i = \id^{E^\prime C^{\prime \otimes {n^{\prime}}}}$, where each $G_i$ is of the form
\begin{equation}\label{eqn:noise-model-2}
	G_i = \sum_{F\in \P_{d'',1}} \sum_{\substack{H\in \P_{d',n'} \\ \mathrm{wt(H)\leq \delta n^\prime}}} \alpha^i_{F,H} F^{E^\prime} \otimes H^{C^{\prime \otimes n^\prime}} \enspace.
\end{equation}
\suppress{and is assessed by requiring that there exists a representation of the global action of Eve on the $n^\prime$ quantum communication registers with Kraus operators of weight at most $\delta n^\prime$.}
In the random noise model, we consider $n^\prime$ independent and identically distributed uses of a noisy quantum channel acting on register $C^\prime$, half the time in each direction. Eve's workspace register $E^\prime$ (including her input register $E$) can be taken to be trivial in this noise model. Note that the adversarial noise model includes the random noise model as a special case.

For both noise models, we say that the simulation succeeds with error $\epsilon$ if for any input, the output in register $\tilde{A} \tilde{B} \tilde{C}$ corresponds to that of running protocol $\Pi$ on the same input, while also maintaining correlations with system $R$, up to error $\epsilon$ in trace distance.

Note that adversaries in the quantum model can inject fully quantum errors since the messages are quantum, in contrast to adversaries corrupting classical messages which are restricted to be modifications of classical symbols. On the other hand, for classical messages the adversary can read all the  messages without the risk of corrupting them, whereas in the quantum model, any attempt to ``read'' messages will result in an error in general on some quantum message.

\subsection{Entanglement distribution}\label{subsec:entanglement-dist}

In our algorithm in the plain quantum model, Alice and Bob need to use MESs as a resource in order to simulate the input protocol. To establish the shared MESs, one party creates the states locally and sends half of each MES to the other party using an appropriate error correcting code of distance~$4n\epsilon$, as described in Algorithm~\ref{algo:Robust Entanglement Distribution}.

\begin{algorithm}
	
	\Input{$\ell$ (desired number of MESs)}	
	
	$C\leftarrow$ Error Correcting Code with rate $1-\Theta(H(\epsilon))$ guaranteed by quantum Gilbert-Varshamov bound~\cite{FMa:2004}\;
	
	\If {$\mathrm{Alice}$}
	{
		Prepare $\ell$ MESs in registers $A,B'$ each holding half of every MES\;
		
		Transmit $C\br{B'}$ to Bob\;
	}
	\ElseIf {$\mathrm{Bob}$}
	{
		Receive $C'(B')$\;
		Decode $C'(B')$ into register $B$\;
	}
	
	\Return{\textup{\textbf{\textsf{Robust Entanglement Distribution}}}}\;
	\caption{\textbf{\textsf{Robust Entanglement Distribution}($\ell$)}}
	\label{algo:Robust Entanglement Distribution}
\end{algorithm}

\subsection{Hashing for string comparison}\label{subsec:hash}

We use randomized hashes to compare strings and catch disagreements probabilistically. The hash values can be viewed as summaries of the strings to be compared. A random bit string called the \emph{seed\/} is used to select a function from the family of hash functions. We say a \emph{hash collision\/} occurs when a hash function outputs the same value for two unequal strings. In this paper we use the following family of hash functions based on the $\eps$-biased probability spaces constructed in \cite{NaorNaor}.

\begin{lemma}[from \cite{NaorNaor}] \label{lem:hashes}
For any $l$, any alphabet $\Sigma$, and any probability $0<p<1$, there exist $s = \Theta(\log (l \log |\Sigma|) + \log \frac{1}{p})$, $o = \Theta(\log \frac{1}{p})$, and a simple function $h$, which given an $s$-bit uniformly random seed $S$ maps any string over $\Sigma$ of length at most $l$ into an $o$-bit output, such that the collision probability of any two $l$-symbol strings over $\Sigma$ is at most $p$. In short:
{%
$$\forall l,\Sigma,0<p<1: \\ 
\quad \exists s = \Theta(\log (l \log |\Sigma|) + \log \frac{1}{p}) \;,\; o = \Theta(\log \frac{1}{p}) \;,\; h: \{0,1\}^s \times \Sigma^{\leq l} \mapsto \{0,1\}^o\ \mathrm{s.t.}$$
$$\forall \stringx,\stringy \in \Sigma^{\leq l}, \stringx \neq \stringy, \strings \in \{0,1\}^s \ \mathrm{i.i.d.}\ \mathrm{Bernoulli}(1/2): \\ 
\qquad P[h_\strings(\stringx) = h_\strings(\stringy)] \leq p$$
}
\end{lemma}

In our application, the hash family of Lemma~\ref{lem:hashes} is used to compare~$\Theta\br{n}$-bit strings, where $n$ is the length of the input protocol. Therefore, in the large alphabet setting, the collision probability can be chosen to be as low as~$p=1/\mathrm{poly}\br{n}$, while still allowing the hash values to be exchanged using only a constant number of symbols. In the teleportation-based model, where Alice and Bob have access to free pre-shared entanglement, they generate the seeds by measuring the MESs they share in the computational basis. In our simulation protocol in the plain quantum model, where Alice and Bob do not start with pre-shared entanglement, they use Algorithm~\ref{algo:Robust Entanglement Distribution} at the outset of the simulation to distribute the MESs they need for the simulation. A fraction of these MESs are measured by both parties to obtain the seeds. 

One advantage of generating the seeds in this way is that the seeds are unknown to the adversary. This is in contrast to the corresponding classical model with no pre-shared randomness, were the seeds need to be communicated over the classical channel and the adversary gets to know the seeds. The knowledge of the seeds enables the adversary to introduce errors which remain undetected with certainty. As a result, Haeupler~\cite{Haeupler:2014} adds another layer of hashing to his algorithm for the oblivious noise model to protect against fully adversarial noise, dropping the simulation rate from~$1-\Theta\br{\sqrt{\epsilon}}$ to~$1-\Theta\br{\sqrt{\epsilon \log \log 1/\epsilon}}$. 

\subsection{Extending randomness to pseudo-randomness}\label{subsec:extending-randomness}

In our simulation algorithm in the plain quantum model, Alice and Bob need to share a very long random string which they need to establish through communication. A direct approach would be for them to distribute enough MESs and measure them in the computational basis to obtain a uniformly random shared string. However, this would lead to a vanishing simulation rate. Instead, they distribute a much shorter  i.i.d. random bit string~$R'$ and using Lemma~\ref{lem:stretch} below stretch it to a pseudo-random string~$R$ of the desired length which is statistically indistinguishable from being independent. Before stating Lemma~\ref{lem:stretch}, we need the following definitions and propositions.

\begin{definition}
	Let $X$ be a random variable distributed over $\{0,1\}^n$ and $J\subseteq \Br{n}$ be a non-empty set. The \emph{bias\/} of $J$ with respect to distribution $X$, denoted $\mathrm{bias}_J\br{X}$, is defined as
	\[
	\mathrm{bias}_J\br{X} \defeq \left| \Pr\br{\sum_{i\in J}X_i=1}-\Pr\br{\sum_{i\in J}X_i=0} \right|\enspace,
	\]
	where the summation is mod $2$. For $J=\emptyset$, bias is defined to be zero, i.e., $\mathrm{bias}_{\emptyset} \br{X}=0$.
\end{definition}

\begin{definition}
	Let $\delta\in\Br{0,1}$. A distribution $X$ over $\{0,1\}^n$ is called a \emph{$\delta$-biased sample space\/} if $\mathrm{bias}_J\br{X}\leq \delta$, for all non-empty subsets $J\subseteq \Br{n}$.
\end{definition}

Intuitively, a small-bias random variable is statistically close to being uniformly distributed. The following lemma quantifies this statement.

\begin{prop} \label{prop:uniform-vs-delta-biased}
	Let $X$ be an arbitrary distribution over $\{0,1\}^n$ and let $U$ denote the uniform distribution over $\{0,1\}^n$. Then we have 
	\[
	\|X-U\|_2^2 \leq 2^{-n}\sum_{J\subseteq\Br{n}} \mathrm{bias}^2_J\br{X}\enspace.
	\]
    In particular, for $\delta$-biased $x$ we have~$\|X-U\|_2 \leq \delta$.
\end{prop}	

We will make use of the following proposition providing an alternative characterization of the $L_1$-distance between two probability distributions.

\begin{prop} \label{prop:l1-distance}
	Let $p$ and $q$ be probability distributions over some (countable) set $\mathcal{Z}$, then 
	\[
	\|p-q\|_1 = 2 \sup_{A\subseteq \mathcal{Z}} \left|p(A)-q(A)\right|\enspace. 
	\]
\end{prop}

We use the following lemma in our algorithms.

\begin{lemma}[\cite{NaorNaor}]\label{lem:stretch}
	For every $\delta\in(0,1)$, there exists a deterministic algorithm which given $O\br{\log n+\log\frac{1}{\delta}}$ uniformly random bits outputs a $\delta$-biased pseudo-random string of $n$ bits. Any such $\delta$-biased string is also $\epsilon$-statistically close to being $k$-wise independent for $\epsilon=\delta^{\Theta\br{1}}$ and $k=\Theta\br{\log \frac{1}{\delta}}$.
\end{lemma}

\subsection{Protocols over qudits}\label{subsec:protocoloverqudits}

In this section, we revisit two quantum communication protocols, both of which are essential to our simulation algorithms and analyze the effect of noise on these protocols.

\subsubsection{Quantum teleportation over noisy channels}

The protocol given here is an extension of quantum teleportation to qudits. Readers may refer to Chapter 6 in~\cite{Wilde11} for more details.

\begin{definition}\label{def:teleportation}\textbf{Quantum teleportation protocol}
	
	Alice possesses an arbitrary $d$-dimensional qudit in state $\ket{\psi}_{A}$, which she wishes to communicate to Bob. They share an MES in the state $\ket{\phi}_{A_1B_1}$.
\begin{enumerate}
	\item  Alice performs a measurement on registers $AA_1$ with respect to the Bell basis $\set{\ket{\phi^{j,k}}}_{j,k}$.
	\item She transmits the measurement outcome $\br{j,k}$ to Bob.
	\item Bob applies the unitary transformation $\rZ_{B_1}^k\rX_{B_1}^j$ on his state to recover $\ket{\psi}$.
	
\end{enumerate}
\end{definition}

\vspace{3mm}

In the rest of the paper the measurements implemented in Definition~\ref{def:teleportation} are referred to as the {\em teleportation measurements} and the receiver's unitary transformation to recover the target state is referred to as {\em teleportation decoding operation}.

If Bob receives $\br{j',k'}$ due to a corruption on Alice's message, the state he gets after decryption will be the following:
\begin{equation}\label{eqn:corruptedteleportation}
	\rZ_B^{k'}\rX_B^{j'}\rX_B^{d-j}\rZ_B^{d-k}\ket{\psi}=e^{\complexi \cdot \frac{2\pi}{d}(j'-j)k'} \rX^{j'-j}\rZ^{k'-k}\ket{\psi}\enspace.
\end{equation}

\subsubsection{Quantum Vernam cipher over noisy qudit channels}\label{sec:qvc}
In this section, we revisit {\em quantum Vernam cipher} (QVC) introduced by Leung~\cite{Leung:2002}, which is a quantum analog of Vernam cipher (one-time-pad). For a unitary operation $U$, the controlled gate $\control{U}$ is defined as
\[\br{\control{U}}_{AB}\ket{j}_A\ket{k}_B\defeq\ket{j}U^j\ket{k} \enspace.\]

The extension of quantum Vernam cipher to qudit systems goes as follows.
\begin{definition}\label{def:Quantum Vernam Cipher}\textbf{Quantum Vernam cipher}
		
		Alice possesses an arbitrary $d$-dimensional qudit in state $\ket{\psi}_{A}$, which she wishes to communicate to Bob. They share an MES pair in the state  $\ket{\phi}_{A_1B_1}\ket{\phi}_{A_2B_2}$, with Alice and Bob holding registers $A_1A_2$ and $B_1B_2$, respectively.
	\begin{enumerate}
		\item Alice applies the unitary transformation $\br{\control{\rZ}}_{A_2A}\br{\control{\rX}}_{A_1A}$.
		\item She transmits the register $A$ to Bob.
		\item Bob applies the unitary transformation $\br{\control{\rX^{-1}}}_{B_1B}\br{\control{\rZ^{-1}}}_{B_2B}$.		
	\end{enumerate}
\end{definition}

\vspace{3mm}

Quantum Vernam cipher uses entanglement as the key to encrypt quantum information sent through an insecure quantum channel. In sharp contrast with the classical Vernam cipher, the quantum key can be recycled securely. Note that if no error occurs on Alice's message, then Bob recovers the state $\ket{\psi}$ perfectly, and at the end of the protocol the MES pair remain intact. The scheme detects and corrects for arbitrary transmission errors, and it only requires local operations and classical communication between the sender and the receiver.

\begin{figure}[ht]
	\begin{center}
\includegraphics[width=250pt]{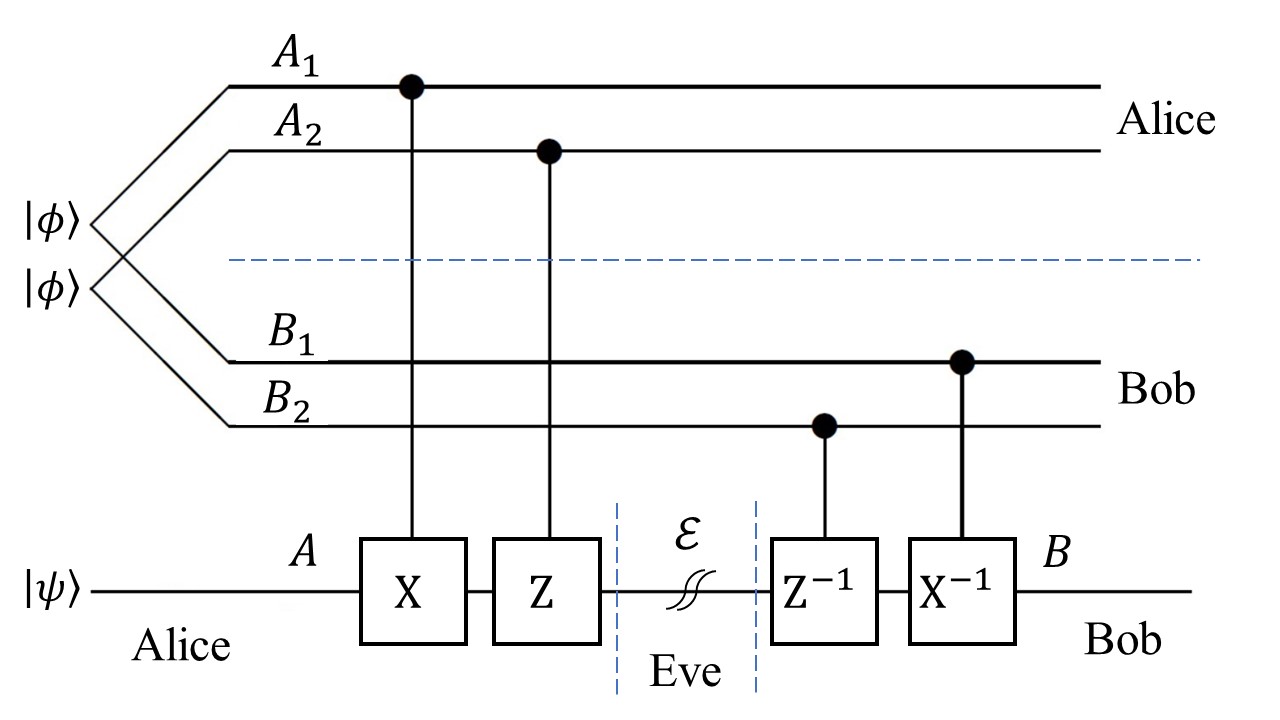}
		\vspace*{2ex}
		\caption{sending one qudit through quantum channel $\mathcal{E}$ using quantum Vernam cipher.}
		\label{fig:scheme1}
	\end{center}
\end{figure}

In particular, if Alice's message is corrupted by the Pauli error $\rX^j\rZ^k$, the joint state after Bob's decryption is
	\begin{eqnarray} \label{QVC with error}
	&&\br{\control{\rX^{-1}}}_{B_1A}\br{\control{\rZ^{-1}}}_{B_2A}\rX^j\rZ^k\br{\control{\rZ}}_{A_2A}\br{\control{\rX}}_{A_1A}\ket{\phi}_{A_1B_1}\ket{\phi}_{A_2B_2}\ket{\psi}_{A}\nonumber\\
	&=&\frac{1}{d}\sum_{t,t'=0}^{d-1}\br{\control{\rX^{-1}}}_{B_1A}\br{\control{\rZ^{-1}}}_{B_2A}\rX^j\rZ^k\br{\control{\rZ}}_{A_2A}\br{\control{\rX}}_{A_1A} \ket{t}_{A_1}\ket{t}_{B_1}\ket{t'}_{A_2}\ket{t'}_{B_2}\ket{\psi}_{A}  \nonumber\\
	&=&\frac{1}{d}\sum_{t,t'=0}^{d-1}\ket{t}_{A_1}\ket{t}_{B_1}\ket{t'}_{A_2}\ket{t'}_{B_2}\rX^{-t}\rZ^{-t'}\rX^j\rZ^k\rZ^{t'}\rX^t\ket{\psi}_{A}\nonumber\\
	&=&\frac{1}{d}\sum_{t,t'=0}^{d-1}e^{\complexi \cdot \frac{2\pi}{d} (kt-jt')}\ket{t}_{A_1}\ket{t}_{B_1}\ket{t'}_{A_2}\ket{t'}_{B_2}\rX^j\rZ^k\ket{\psi}\nonumber\\
	&=&\ket{\phi^{0,k}}_{A_1B_1}\ket{\phi^{0,-j}}_{A_2B_2}\otimes \rX^j\rZ^k\ket{\psi}.\label{eq:messagecorrupt}
	\end{eqnarray}

Note that by Equation \eqref{eq:messagecorrupt}, there is a one-to-one correspondence between the Pauli errors and the state of the maximally entangled pair. An $\rX^{j}$ error on the cipher-text is reflected in the state of the second MES as a $\rZ^{-j}$ error and a $\rZ^{k}$ error on the cipher-text is reflected in the state of the first MES as a $\rZ^{k}$ error. Note that for every integer $s$ we have
$$\br{F\otimes F^{\dagger}} \ket{\phi^{0,s}} \quad=\quad \br{F\rZ^{s}\otimes F^{\dagger}}\ket{\phi} \quad=\quad \br{F\rZ^{s}F^{\dagger}\otimes\id}\ket{\phi} \quad=\quad \br{\rX^{s}\otimes \id}\ket{\phi}.$$
Therefore, in order to extract the error syndrome, it suffices for Alice and Bob to apply $F$ and $F^{\dagger}$, respectively, on their marginals of the MESs and measure them in the computational basis. By comparing their measurement outcomes they can determine the Pauli error.

When quantum Vernam cipher is used for communication of multiple messages, it is possible to detect errors without disturbing the state of the MES pairs at the cost of an additional fresh MES. This error detection procedure allows for recycling of MESs which is crucial in order to achieve a high communication rate, as explained in Section~\ref{sec:Qhashing}. Here we describe a simplified version of the detection procedure. First we need the following two lemma.

\begin{prop}\label{lem:cnotbell}
	It holds that
	\[\br{\control{\rX}}_{A_1A_2} \cdot \br{\control{\rX} }_{B_1B_2}\ket{\phi^{j_1,k_1}}_{A_1B_1}\ket{\phi^{j_2,k_2}}_{A_2B_2}=\ket{\phi^{j_1,k_1-k_2}}_{A_1B_1}\ket{\phi^{j_1+j_2,k_2}}_{A_2B_2}.\]
	In particular,
	\[\br{\control{\rX}}_{A_1A_2} \cdot \br{ \control{\rX} }_{B_1B_2}\ket{\phi^{0,k_1}}_{A_1B_1}\ket{\phi^{0,k_2}}_{A_2B_2}=\ket{\phi^{0,k_1-k_2}}_{A_1B_1}\ket{\phi^{0,k_2}}_{A_2B_2}.\]
\end{prop}
\begin{proof}
	\begin{eqnarray*}
	&&\br{\control{\rX}}_{A_1A_2}\cdot \br{\control{\rX}}_{B_1B_2}\ket{\phi^{j_1,k_1}}_{A_1B_1}\ket{\phi^{j_2,k_2}}_{A_2B_2}\\&=&\frac{1}{d}\sum_{t_1,t_2=0}^{d-1}\br{\control{\rX}}_{A_1A_2}\cdot \br{\control{\rX}}_{B_1B_2}e^{\complexi \cdot \frac{2\pi}{d} (t_1k_1+t_2k_2) }\ket{t_1+j_1}_{A_1}\ket{t_1}_{B_1}\ket{t_2+j_2}_{A_2}\ket{t_2}_{B_2}\\
	&=&\frac{1}{d}\sum_{t_1,t_2=0}^{d-1}e^{\complexi \cdot \frac{2\pi}{d}(t_1k_1+t_2k_2)}\ket{t_1+j_1}_{A_1}\ket{t_1}_{B_1}\ket{t_2+j_2+t_1+j_1}_{A_2}\ket{t_2+t_1}_{B_2}\\
	&=&\ket{\phi^{j_1,k_1-k_2}}_{A_1B_1}\ket{\phi^{j_1+j_2,k_2}}_{A_2B_2}.
	\end{eqnarray*}
\end{proof}

\suppress{
\begin{lemma}\label{lem:hashorder}
  Given registers $A_0,A_1,A_2,B_0,B_1,B_2$, the operators $$\br{\control{\rX}}_{A_0A_1}, \br{\control{\rX}}_{B_0B_1}, \br{\control{\rX}}_{A_0A_2},
  \br{\control{\rX}}_{B_0B_2}$$
	commute with each other.
\end{lemma}
\begin{proof}
	$\br{\control{\rX}}_{A_0A_1}$ commutes with $\br{\control{\rX}}_{B_0B_1}$ and $\br{\control{\rX}}_{B_0B_2}$ as they act on disjoint registers. For the same reason, $\br{\control{\rX}}_{A_0A_2}$ commutes with $\br{\control{\rX}}_{B_0B_1}$ and $\br{\control{\rX}}_{B_0B_2}$. Thus, it suffices to show that $\br{\control{\rX}}_{A_0A_1}$ commute with $\br{\control{\rX}}_{A_0A_2}$.
	
	Let $\ket{a_0},\ket{a_1},\ket{a_2}$ be an element of the computational basis in the registers $A_0,A_1,A_2$, respectively.
	Then
	\begin{eqnarray*}
		&&\br{\control{\rX}}_{A_0A_1}\cdot\br{\control{\rX}}_{A_0A_2}\ket{a_0,a_1,a_2}\\
		&=&\br{\control{\rX}}_{A_0A_2}\cdot\br{\control{\rX}}_{A_0A_1}\ket{a_0,a_1,a_2}\\
		&=&\ket{a_0, \, a_0{+}a_1 \bmod d, \, a_0{+}a_2 \bmod d}.
	\end{eqnarray*}
\end{proof}
}

Suppose that Alice and Bob start with $m$ copies of the MES $\ket{\phi}$ and use them in pairs to communicate messages using QVC over a noisy channel. By Equation~\eqref{eq:messagecorrupt} all the MESs remain in $\textsf{span}\set{\ket{\phi^{0,k}}:0\leq k\leq d-1}$. This invariance is crucial to the correctness of our simulation. Let $\ket{\phi^{0,k_i}}_{A_iB_i}$ be the state of the $i$-th MES after the communication is done. In order to detect errors, Alice and Bob use an additional MES $\ket{\phi}_{A_0B_0}$. For $i=1,...,m$, Alice and Bob apply $\br{\control{\rX}}_{A_0A_i}$ and $\br{\control{\rX}}_{B_0B_i}$, respectively. By Proposition~\ref{lem:cnotbell}, the joint state of the register $A_0B_0$ will be $\ket{\phi^{0,-\!\sum_{i=1}^m \! k_i}}_{A_0B_0}$. Now, all Alice and Bob need to do is to apply $F$ and $F^{\dagger}$ on registers $A_0$ and $B_0$, respectively, and measure their marginal states in the computational basis. By comparing their measurement outcomes they can decide whether any error has occurred. In this procedure the MESs used as the keys in QVC are not measured. Note that if the corruptions are chosen so that~$\sum_{i=1}^m \! k_i=0 \bmod d$ then this procedure fails to detect the errors. We will analyze a modified version of this error detection procedure in detail in Section~\ref{sec:Qhashing} which allows error detection with high probability independent of the error syndrome.

\newpage
\section{Teleportation-based protocols via classical channel with large alphabet}\label{sec:BKlarge}


\subsection{Overview}

We adapt from ~\cite{BNTTU14} the ideas to teleport each quantum
message and to rewind the protocol instead of backtracking.

We also adapt Haeupler's template~\cite{Haeupler:2014} to make a
conversation robust to noise: Both parties conduct their original
conversation as if there were no noise, except for the following:
\begin{itemize}
\item At regular intervals they exchange concise summaries (a $\Theta \br{1}$\suppress{ or
  $\Theta \br{\log \log n}$}-bit hash value) of the conversation up to the
  point of the exchange.
\item If the summary is consistent, they continue the conversation.
\item If the summary is inconsistent, an error is detected. The parties backtrack to an earlier
  stage of the conversation and resume from there.
\end{itemize}
This template can be interpreted as an error correcting
code over many messages, with trivial (and most importantly \emph{message-wise\/})
encoding.  The 2-way summaries measure the error syndromes over a large number
of messages, thereby preserving the rate.  It
works (in the classical setting) by limiting the maximum amount of
communication wasted by a single error to $O_\epsilon \br{1}$.  The worst case
error disrupts the consistency checks, but Alice and Bob agree to
backtrack a constant amount when an inconsistency is detected.
As the error fraction vanishes, the communication rate goes to~$1$.
In addition, these consistency tests are efficient, consisting of
evaluation of hash functions.

\subsubsection{Insufficiency of simply combining
\cite{BNTTU14} and~\cite{Haeupler:2014}.}
Suppose we have to simulate an interactive protocol $\Pi$ that uses
noiseless classical channels in the teleportation-based model.
When implementing $\Pi$ with noisy classical channels, it is \emph{not sufficient\/}
to apply Haeupler's template to the classical messages used in teleportation,
and rewind as in~\cite{BNTTU14} when an error is detected.
%
%
%
The reason is that, in~\cite{BNTTU14}, each message is expanded to
convey different types of actions in one step (simulating the protocol
forward or reversing it). This also maintains the matching between
classical data with the corresponding MES, and the matching between
systems containing MESs.  However, this method incurs a large constant factor
overhead which we cannot afford to incur.

\subsubsection{New difficulties in rate-optimal simulations.}
Due to errors in
communication, the parties need to actively rewind the
simulation to correct errors on their joint quantum state. This itself
can lead to a situation where the parties may not agree on how they
proceed with the simulation (to rewind simulation or to proceed forward).
In order to move on, both parties first need to know what the other
party has done so far in the simulation. This allows them to obtain a
global view of the current joint state and decide on their next
action.
In Ref. \cite{BNTTU14}, this reconciliation step was facilitated
by the extra information sent by each party and the use of tree codes.
This mechanism is not available to us.

\subsubsection{Framework.}
Our first new idea is to introduce sufficient yet
concise data structures so that the parties can detect inconsistencies
in (1)  the stage in which they are in the protocol, (2) what type of action
they should be taking, (3) histories leading to the above, (4)
histories of measurement outcomes generated by one party versus the
potentially different (corrupted) received instruction for teleportation decoding,
(5) which system contains the next MES to be used, (6) a classical
description of the joint quantum state, which is only partially known
to each party.  Each of Alice and Bob maintain her/his data (we
collectively call these $D_\sA, D_\sB$ respectively, here), and also an
estimate of the other party's data ($\widetilde{D_\sB}, \widetilde{D_\sA}$
respectively).  Without channel noise, these data are equal to their
estimates.

\subsubsection{A major new obstacle: out-of-sync teleportation.}

At
every step in the simulation protocol $\Pi'$, Alice and Bob may engage in
one of three actions: a forward step in $\Pi$, step in reverse, or
the exchange of classical summaries.
However, the summaries can also be corrupted.  This leads to a new
difficulty: errors in the summaries can trigger Alice and Bob
to engage in different actions.  In particular, it is possible that
one party tries to teleport while the other expects classical
communication, with only one party consuming his/her half of an MES.
They then become out-of-sync over which MESs to use.
This kind of problem, to the best of our knowledge,
has not been encountered before, and it is not
clear if quantum data can be protected from such error.  (For example, Alice may
try to teleport a message into an MES that Bob already ``used''
earlier.)
One of our main technical contributions is to show that the quantum
data can always be located and recovered when Alice and Bob resolve
the inconsistencies in their data $(D_\sA, \widetilde{D_\sB})$ and
$(\widetilde{D_\sA},D_\sB)$ in the low noise regime.
%
This is particularly surprising since quantum data
can potentially leak irreversibly to the environment (or the adversary):
Alice and Bob potentially operate in an open system due to channel noise,
and out-of-sync teleportation a priori does not protect the messages so sent.

\subsubsection{Tight rope between robustness and rate.}
The simulation maintains sufficient data
structures to store information about each party's view
 so that Alice and Bob can
overcome all the obstacles described above.
The simulation makes progress so long as Alice's and Bob's views are consistent.
The robustness of the simulation requires that the consistency checks be frequent
and sensitive enough so that errors are caught quickly. On the other hand,
to optimize interactive channel capacity, the checks have to remain communication efficient and not too frequent neither.
This calls for delicate analysis in which we balance the two. We also put in some redundancy
in the data structures to simplify the analysis.

\subsection{Result}

In this section, we focus on the teleportation-based quantum communication model with polynomial-size alphabet. In more detail, Alice and Bob share an unlimited number of copies of an MES before the protocol begins. The parties effectively send each other a qudit using an MES and communicating two classical symbols from the communication alphabet. The complexity of the protocol is the number of classical symbols exchanged, while the MESs used are available for free.
We call this model noiseless if the classical channel is noiseless.

The following is our main result in this model for simulation of an $n$-round noiseless communication protocol over an adversarial channel that corrupts any $\epsilon$ fraction of the transmitted symbols.
\suppress{First, we state the result for large alphabets.}
\begin{theorem}
    Consider any $n$-round alternating communication protocol $\Pi$ in the teleportation-based model, communicating messages over a noiseless channel with an alphabet $\Sigma$ of bit-size $\Theta\br{\log n}$. Algorithm \ref{algo:Mainalgorithm} is a computationally efficient coding scheme which given $\Pi$, simulates it with probability at least $1-2^{-\Theta\br{n\epsilon}}$, over any fully adversarial error channel with alphabet $\Sigma$ and error rate $\epsilon$. The simulation uses $n\br{1+\Theta\br{\sqrt{\epsilon}}}$ rounds of communication, and therefore achieves a communication rate of $1-\Theta\br{\sqrt{\epsilon}}$. Furthermore. the computational complexity of the coding operations is~$O\br{n^2}$.
\end{theorem}

\suppress{Consider teleportation-based noiseless communication protocols of
length~$n$ defined over a channel with a $\Theta(\log n)$-bit alphabet,
and the problem of simulating them with a noisy version of the channel 
over the same alphabet.

There is a protocol that with probability at least~$1 -
2^{-\Omega(\epsilon n)}$, simulates any $n$-symbol teleportation-based 
noiseless communication protocol using $n(1+\Theta(\sqrt{\epsilon}))$ 
symbols over any fully adversarial error channel with error rate at 
most~$\epsilon$. In other words, the simulation achieves information 
rate $1-\Theta(\sqrt{\epsilon})$.

The simulation of channels over constant-size alphabets is more
challenging. Nonetheless, we show that a similar simulation is possible
in this case as well. 
\begin{theorem}
Consider teleportation-based noiseless communication
protocols of length~$n$ defined over a channel with a constant-size alphabet,
and the problem of simulating them with a noisy version of the channel 
over the same alphabet.

There is a protocol that with probability at least~$1 -
2^{-\Omega(\epsilon n)}$, simulates any $n$-symbol teleportation-based 
noiseless communication protocol using $n(1+\Theta(\sqrt{\epsilon}))$ 
symbols over any fully adversarial error channel with error rate at 
most~$\epsilon$. In other words, the simulation achieves information 
rate $1-\Theta(\sqrt{\epsilon})$.
\end{theorem}
We prove this in Section~\ref{sec:small-alphabet-CB}.
}

\subsection{Description of Protocol}\label{sec:general-descrpition-largeclasscial}


We follow the notation associated with quantum
communication protocols introduced in Section~\ref{sec:qucomm} in the
description below.

Recall that in the teleportation-based quantum communication model, Alice 
and Bob implement a protocol~$\Pi_0$ with prior shared entanglement and 
quantum communication by substituting teleportation for quantum communication.
For simplicity, we assume that~$\Pi_0$ is alternating, and begins with
Alice.  In the implementation~$\Pi$ of~$\Pi_0$, the message 
register~$C$ from~$\Pi_0$ has two
counterparts, $C_\sA$ and~$C_\sB$, held by Alice and Bob, respectively. The
unitary operations on~$AC$ in~$\Pi_0$ are applied by Alice on~$A C_\sA$
in~$\Pi$. When Alice
sends the qudit in~$C$ to Bob in~$\Pi_0$, she applies the teleportation
measurement to~$C_\sA$ and her share of the next available MES, and sends
the measurement outcome to Bob in~$\Pi$. Then Bob applies a decoding 
operation on his share of the MES, based on the message received, and 
swaps the MES register with~$C_\sB$.
Bob and Alice's actions in~$\Pi$ when Bob wishes to do a local
operation and send a qudit to Alice in~$\Pi_0$ are analogously defined.
For ease of comparison with the joint state in~$\Pi_0$, we describe the 
joint state of the registers in~$\Pi$ (or its simulation over a noisy
channel) in terms of registers~$ABC$. There, $C$ stands for~$C_\sA$ 
if Alice is to send the next message or all messages have been sent, and
for~$C_\sB$ if Bob is to send the next message.

Starting with such a protocol~$\Pi$ in the teleportation-based model,
we design a simulation protocol~$\Pi'$ which uses a noisy classical 
channel. 
The simulation works with \emph{blocks\/} of even number of
messages.  By a \emph{block\/} of size~$r$ (for even~$r$) of~$\Pi$, we mean a 
sequence of~$r$ local operations and messages alternately sent in~$\Pi$  
by Alice and Bob, starting with Alice.

Roughly speaking, 
Alice and Bob run the steps of the original protocol~$\Pi$ as is, in
blocks of size $r \defeq \Theta (\frac{1}{\sqrt{\epsilon}})$, with $r$ even.
They exchange summary information between these blocks, in order to
check whether they agree on the operations that have been applying to the quantum registers $A B C$ in the simulation. 
The MESs used for teleportations are correspondingly divided into 
blocks of $r$ MESs, implicitly numbered from $1$ to $r$: the odd numbered
ones are used to simulate quantum communication from Alice to Bob,
and the even numbered ones from Bob to Alice. 
If either party detects an error in transmission, they may run a block 
of~$\Pi$ in reverse, or simply communicate classically to help
recover from the error. The classical communication is also conducted in
sequences equal in length to the ones involving a block of~$\Pi$. A block
of~$\Pi'$ refers to any of these types of sequences.

\subsubsection{Metadata}

In more detail, Alice uses an iteration in~$\Pi'$ for one out of four different types of
operations: evolving the simulation by running a block of~$\Pi$ in the
forward direction (denoted a ``$+1$'' block); reversing the simulation by applying inverses
of unitary operations of~$\Pi$ (denoted a ``$-1$'' block); synchronizing with Bob on the number of MESs
used so far by applying identity operators between rounds of teleportation or reversing such an iteration (denoted a ``$0$'' block, with $0$
standing for the application of unitary operations~$U_i^0$ which are~$\id^{AC}$); catching up on the
description of the protocol so far by exchanging classical data with Bob
(denoted a ``$\sC$'' block,
with $\sC$ standing for ``classical''). Alice records the sequence of types of iterations as
her ``metadata'' in the string $\FullMA \in \{\pm1, 0, \sC \}^*$.
$\FullMA$ gets extended by one symbol for each new iteration of the
simulation protocol~$\Pi'$.
The number of blocks of $r$ MESs Alice has used is denoted $q_\MA$ which 
corresponds to the number of non-$\sC$ symbols in $\FullMA$.
Similarly, Bob maintains data $\FullMB$ and  $q_\MB$.

$\FullMA$ and $\FullMB$ may not agree due to the 
transmission errors. To counter this, the two players exchange
information about their metadata at the end of each block.
Hence, Alice also holds $\MBtilde$
and $q_{\MBtilde}$ as her best estimation of Bob's
metadata and the number of MESs he has used, respectively. Similarly,
Bob holds $\MAtilde$ and $q_{\MAtilde}$.
We use these data to control the simulation; before taking any action
in~$\Pi'$, Alice checks if her guess $\MBtilde$
equals~$\FullMB$. Bob does the analogous check for his data.

\subsubsection{Number of MESs used}

Once both parties reconcile their view of each other's metadata with
the actual data, they might detect a discrepancy in the number of
MESs they have used.
The three drawings in Figure~\ref{fig:EPR_registers} represent the $\lceil \frac{n}{2r}(1 + O(r \epsilon))\rceil$ blocks
of $r=O(\sqrt{1/\epsilon})$ MESs at different points in the protocol:
first, before the protocol begins; second, when Alice and Bob have used
the same number of MESs; and third, when they are not synchronized, say, 
Alice has used more blocks of MESs than Bob.
A difference in~$q_{\MA}$ and~$q_{\MB}$ indicates that the joint state 
of the protocol~$\Pi$ can no longer be recovered from 
registers~$A C_\sA C_\sB B$ alone. Since one party did not correctly 
complete the teleportation operations, the (possibly erroneous) joint state 
may be thought of as having ``leaked'' into the partially measured MESs
which were used by only one party. We will elaborate on this scenario in Section~\ref{sec:Out-of-sync teleportation}.

\begin{figure}[!t]
\centering
\includegraphics[width=480pt]{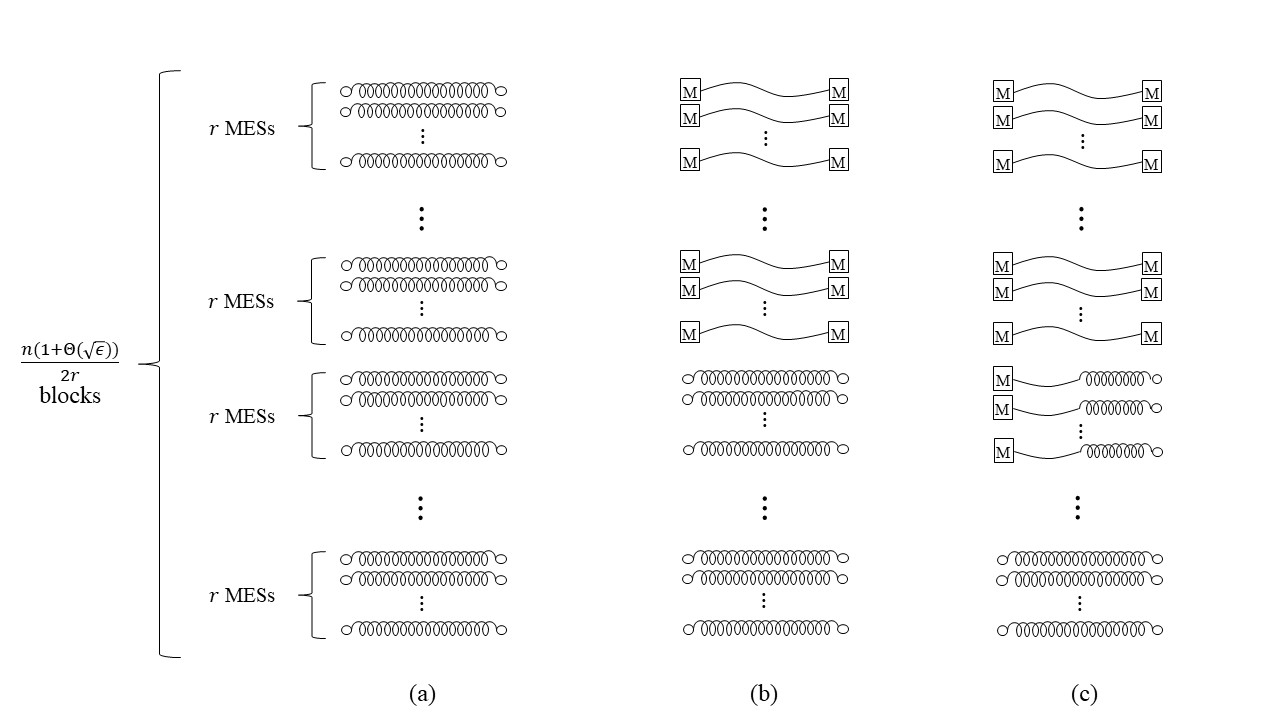}
\caption{These figures represent the MES blocks at different
stages of the protocol. The systems depicted by circles have not been used yet
for teleportation, those depicted by squares have been used already. (either ``Measured'' or teleportation-decoded.)
Figure (a) represents the MES blocks at the beginning of the protocol, when none
have been used. Figure (b) represents them when Alice and Bob have
used the same number of them; this is the desired situation. Figure
(c) represents a situation when Alice and Bob are out of sync; e.g., Alice has used more MES blocks than Bob. They then work to get back in sync before resuming the simulation.}
\label{fig:EPR_registers}
\end{figure}

\subsubsection{Pauli data}

The last piece of information required to complete the description of what has happened so far on the quantum registers $A B C$
is about the Pauli operators corresponding to teleportation, which we
call the ``Pauli data''.
These Pauli data contain information about the teleportation measurement outcomes as well as about
the teleportation decoding operations. Since incorrect teleportation decoding may arise due
to the transmission errors, we must allow the parties to apply Pauli corrections at some point.
We choose to concentrate such Pauli corrections on the receiver's side at the end of each teleportation.
These Pauli corrections are computed from the history of all classical
data available, before the evolution or reversal of~$\Pi$ in a block starts.\suppress{ 
whereas the
measurement and decoding Pauli data are exchanged online during the computation.} The
measurement data are directly transmitted over the noisy classical communication
channel and the decoding data are directly taken to be the data received over the
noisy channel. If there is no transmission error, the decoding Pauli operation
should correspond to the inverse of the effective measurement Pauli operation and cancel out to yield a noiseless quantum channel. Figure~\ref{fig:teleportation-representation} depicts the
different types of Pauli data in a block corresponding to type $+1$ for Alice and $-1$
for Bob. 
The Pauli operations applied on Alice's side are in the following order: 
\begin{quote}
teleportation measurement for the first qudit she sends, \\
decoding operation for the first qudit she receives, \\
correction operation for the same qudit (the first qudit she receives); 

teleportation measurement for the second qudit she sends, \\
decoding operation for the second qudit she receives, \\
correction operation for the same qudit (the second qudit she receives); 

and so on.
\end{quote} 
The Pauli operations applied on Bob's side are in a different order:
\begin{quote}
decoding operation for the first qudit he receives, \\ 
correction operation for the same qudit (the first qudit he receives),
\\
teleportation measurement for the first qudit he teleports; 

decoding operation for the second qudit he receives, \\
correction operation for the same qudit (the second qudit he receives),
\\
teleportation measurement for the second qudit he sends;

and so on.
\end{quote}

\noindent Alice records as her Pauli data in the string $\FullPA \in (\Sigma^{3r})^*$, the sequence of Pauli operators that are applied on the quantum register on her side. Each block of $\FullPA$ is divided into 3 parts of r symbols from the alphabet set $\Sigma$. The first part corresponds to the $\frac{r}{2}$ teleportation measurement outcomes with two symbols for each measurement outcome. Each of the $\frac{r}{2}$ teleportation decoding operations are represented by two symbols in the second part. Finally, the third part contains two symbols for each of the $\frac{r}{2}$ Pauli corrections. Similarly, Bob records the sequence of Pauli operators applied on his side in $\FullPB$.
As described above, the measurement outcome and the decoding Pauli
operations are available to the sender and the receiver, respectively.
Based on the message transcript in~$\Pi'$ so far,
Alice maintains her best guess $\PBtilde$ 
for Bob's Pauli data and Bob maintains his best guess $\PAtilde$ for Alice's Pauli data. 
These data also play an important role in the simulation. Before taking any
action in~$\Pi'$, Alice checks if her guess~$\PBtilde$ 
equals $\FullPB$. Bob does the analogous check for his data.

Alice and Bob check and synchronize their classical data, i.e., the metadata 
and Pauli data, by employing the ideas underlying the Haeupler
algorithm~\cite{Haeupler:2014}.
Once they agree on each other's metadata and
Pauli data, they both possess enough information to
compute the content of the quantum register (to the best of their knowledge).

\begin{figure}[!t]
\centering
\includegraphics[width=350pt]{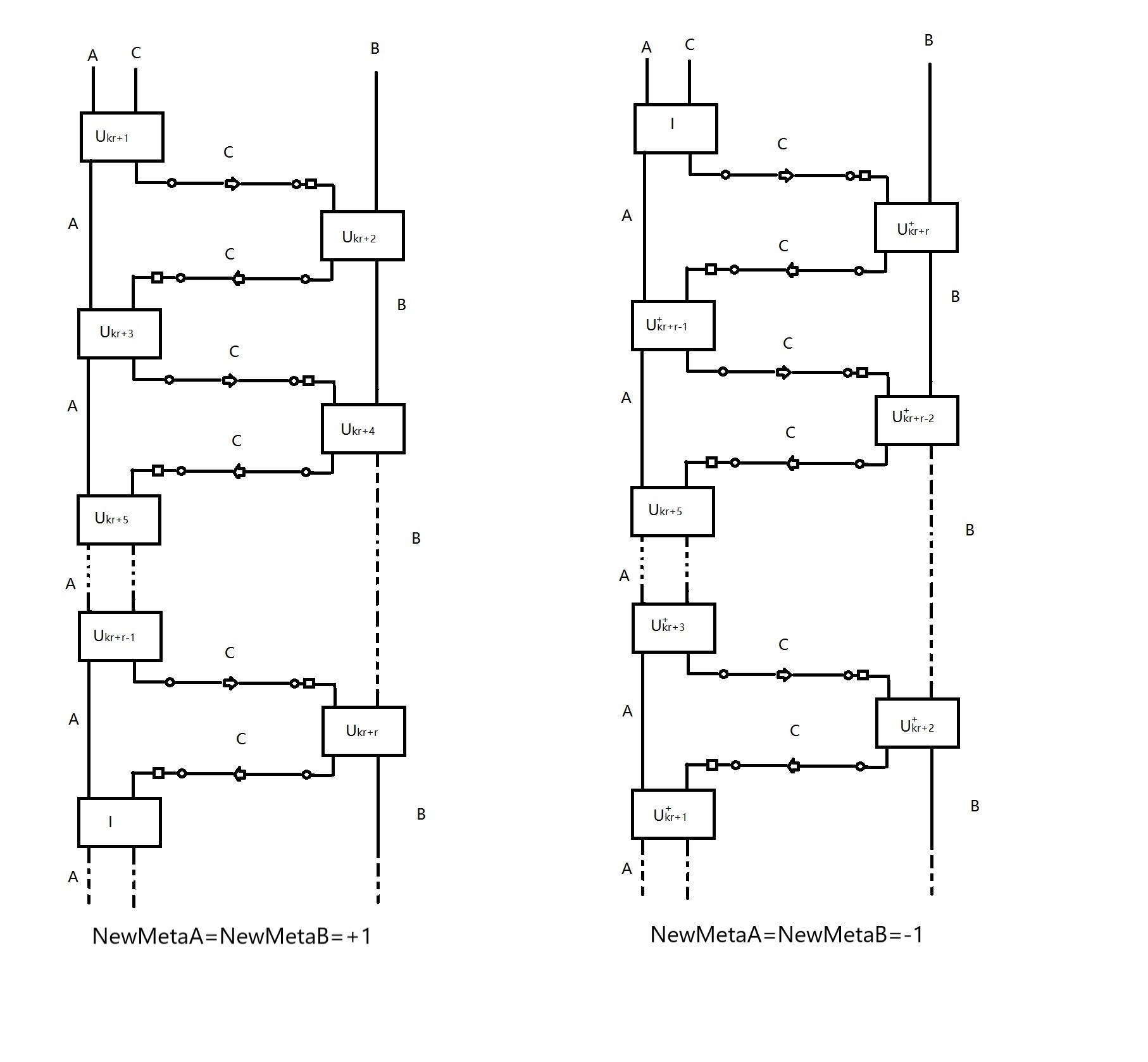}
\caption{Representation of the teleportation scheme for a size $r$
block. The figure on the left corresponds to Alice and Bob having blocks
of type $+1$, the most common block type, and the one on the right to a
block of type $-1$ for both. The large rectangles correspond to unitary
operations of the original
protocol or their inverses, or even an identity operator, being applied by Alice or by Bob to $A C$ or
$B C$, respectively. Bob has $r/2$ rectangles and applies
a unitary operation or an inverse in each of them whenever he has a block of type
$\pm 1$. Alice has $r/2 + 1$ rectangles and uses the first $r/2$ to
apply unitary operations in a block of type $+1$ and apply an identity
on the last one, while she applies an identity in the first one and
inverses of unitary operations in the $r/2$ last ones in a block of type
$-1$. This is so that a $-1$ block for Alice can be the inverse of a
$+1$ block for Alice, and vice-versa. The small circles correspond to
the Pauli operations due to teleportation measurement and teleportation
decoding, with the teleportation being from Alice to Bob on odd numbered MESs
and from Bob to Alice on even numbered MESs. The small squares on the
receiver side right after the teleportation decoding circle corresponds
to the Pauli corrections made in order to try to correct errors in previous blocks.}
\label{fig:teleportation-representation}
\end{figure}

\subsubsection{Out-of-Sync Teleportation}
\label{sec:Out-of-sync teleportation}

\suppress{
Consider the situation where the classical data on both sides are of
full length but do not match at the beginning of an iteration of the
simulation. Suppose that the adversary corrupts the classical data that
the parties communicate to each other in this iteration. If no error or
hash collision occurs in the next round, the parties will realize there
is an error by checking the hash values and will try to fix their
incorrect information. Now if a hash collision or corruption of the
hashes by the adversary makes only Alice believe that her current guess
of Bob's local data is correct, she will continue the simulation of the
input protocol for another round, consuming the next block of MESs.
On the other hand, Bob will try to resolve the inconsistency in their
classical data, without accessing the quantum registers.}

\subsubsection*{Basic out-of-sync scenario}

Consider an iteration in which Alice believes she should implement
a~$+1$ block, while Bob believes he has to resolve an inconsistency in
their classical data.
Alice will simulate one block of the input protocol~$\Pi$,
consuming the next block of MESs.
On the other hand, Bob will try to resolve the inconsistency through
classical communication alone, and not access the quantum registers.
Thus Alice will
treat Bob's messages as the outcomes of his teleportation measurements,
and she performs the teleportation decoding operations according to
these messages. The situation is even worse, since Alice sends quantum
information to Bob through teleportation of which Bob is unaware, and
Bob views the teleportation measurement outcomes sent by Alice 
as classical information about Alice's local Pauli data and
metadata corresponding to previous iterations. Note that at this point
the quantum state in registers~$ABC$ may potentially be lost.
This scenario could continue for
several  iterations and derail the simulation completely. To recover
from such a situation, especially to retrieve the quantum information in
the unused MESs at his end, it would seem that Alice and Bob would have
to rewind the simulation steps in~$\Pi'$ (and not only the steps of the
original protocol~$\Pi'$) to an appropriate point in the past.
This rewinding itself
would be subject to error, and the situation seems hopeless.
Nonetheless, we provide a simple solution to address this kind of error,
which translates out-of-sync teleportation to errors in implementing the
forward simulation or rewinding of the original protocol~$\Pi$.

As explained in the previous subsection, Alice and Bob first reconcile
their view of the history of the simulation stored in their metadata.
Through this, suppose they both discover the discrepancy in the number
of MESs used. (There are other scenarios as well; for example, they may
both think that~$q_\MA = q_\MB$. These scenarios lead to further
errors, but the simulation protocol~$\Pi'$  eventually discovers the
difference in MESs used.) In the scenario in which Alice and Bob both discover
that~$q_\MA \neq q_\MB$, they try to ``gather'' the quantum data hidden in the partially used MESs back into the registers~$A B C$. In more detail, suppose Bob has used fewer MESs than Alice, and he discovers this at the beginning of the~$i$-th iteration. Let~$E_1 E_2 \dotsb E_r$ be registers with Bob that hold the halves of the \emph{first\/} block of MESs that Alice has used but Bob has not. Note that~$E_1, E_3, \dotsc, E_{r-1}$ contain quantum information teleported by Alice, and~$E_2, E_4, \dotsc, E_r$ are MES-halves intended for teleportation by Bob. The MES-halves corresponding to~$E_2, E_4, \dotsc, E_r$ have already been used by Alice to ``complete'' the teleportations she assumed Bob has performed. Say Alice used this block of MESs in the~$i'$-th iteration. In the~$i$-th iteration, Bob teleports the qudit~$E_1$ using the MES-half~$E_2$, $E_3$ with~$E_4$, and so on. That is, Bob teleports qudit~$E_j$ using the MES-half~$E_{j+1}$ in increasing order of~$j$, for all odd~$j \in [r]$, as if the even numbered MESs had not been used by Alice. The effect of this teleportation is the same as if Alice and Bob had \emph{both\/} tried to simulate the local operations and communication from the original protocol in the~$i'$-th iteration (in the forward direction or to correct the joint state), \emph{except that the following also happened independently of channel error\/}:
\begin{enumerate}
\item the Pauli operations used by Bob to decode~$E_1, E_3, \dotsc, E_{r-1}$ were all the identity,
\item the unitary operations used by Bob on the registers~$B C$ were all the identity, and
\item the Pauli operations applied by Alice for decoding Bob's teleportation were unrelated to the outcome of Bob's teleportation measurements.
\end{enumerate} 
This does not guarantee correctness of the joint state in~$A B C$, but has the advantage that quantum information in the MES-halves~$E_1, E_3, \dotsc, E_{r-1}$ that is required to restore correctness is redirected back into the registers~$A B C$. In particular, the difference in the number of MESs used by the two parties is reduced, while the errors in the joint quantum state in~$A B C$ potentially increase. The errors in the joint state are eventually corrected by reversing the incorrect unitary operations, as in the case when the teleportations are all synchronized.

To understand the phenomenon described above, consider a simpler
scenario where Bob wishes to teleport a qudit~$\ket{\xi}$ in
register~$B_1$ to Alice using an MES in registers~$E_1' E_1$, after
which Alice applies the unitary operation~$V$ to register~$E_1'$. If
they follow the corresponding sequence of operations, the final state
would be~$V\ket{\xi}$, stored in register~$E_1'$. Instead suppose they
do the following. First, Alice applies~$V$ to register~$E_1'$,
\emph{then\/} Bob measures registers~$B_1 E_1$ in the generalized Bell
basis and gets measurement outcome~$(j,k)$. He sends this outcome to
Alice. We may verify the state of register~$E_1'$ conditioned on the
outcome is~$V(\rX^j \rZ^k)\ket{\xi}$. Thus, the quantum information in~$\xi$
is redirected to the correct register, albeit with a Pauli error (that
is known to Alice because of his message). In particular, Alice may
later reverse~$V$ to correctly decode the teleported state. The chain of
teleportation steps described in the previous paragraph has a similar effect.

\subsubsection{First representation of the quantum registers}
\label{sec:large-classical-first-rep}

A first representation for the content of the quantum registers~$ABC$
in~$\Pi'$ can be obtained directly and explicitly from the
metadata and the Pauli data, and is denoted $\JSone$, as in
Eq.~\eqref{eqn:JS1} below, with $\mathit{JS}$ standing for ``joint state''.
We emphasize that this is the state conditioned on the outcomes of the
teleportation measurements as well as the transcript of
classical messages received by the two parties.
However, the form $\JSone$ is essentially useless for
deciding the next action that the simulation protocol~$\Pi'$ should take,
but it can be simplified into a more useful representation.
This latter form, denoted $\JStwo$, as in Eq.~\eqref{eqn:JS2_1}
below,  directly corresponds to the further actions we may take in order to evolve the simulation of the original protocol or to actively reverse
previous errors.
We first consider $\JSone$
and $\JStwo$ in the case when $q_{\MA} = q_{\MB}$. 


We sketch how to obtain $\JSone$ from $\FullMA$,
$\FullMB$, $\FullPA$ and $\FullPB$ (when~$q_\MA
= q_\MB$). Each block of $r$ MESs which
have been used by both Alice and Bob corresponds to a bracketed
expression~$[*j]$ for some content ``$*j$'' corresponding
to the $j$-th block that we describe below. The content of the quantum
registers is then the $A B C$ part of
\begin{align}\label{eqn:JS1}
\JSone = [*q_{\MA}] \cdots [*2] [*1]
\ket{\psi_{\mathrm{init}}}^{A B C  E R},
\end{align}
with  $\ket{\psi_{\mathrm{init}}}^{ABCER}$ being the initial state of
the original protocol. (To be accurate, the representation corresponds to the sequence of operations that have been applied to $\ket{\psi_{\mathrm{init}}}$, and knowledge of $\ket{\psi_{\mathrm{init}}}$ is not required to compute the representation.) It remains to describe the content $*j$ of the $j$-th bracket.
It contains from right to left $\frac{r}{2}$ iterations of the following: 
\begin{quote}
Alice's unitary operation - Alice's teleportation measurement outcome - \\
Bob's teleportation decoding - Bob's Pauli correction -
Bob's unitary operation - Bob's teleportation measurement outcome - \\
Alice's teleportation decoding - Alice's Pauli correction.
\end{quote}
It also allows  for an additional unitary operation of Alice on the far
left when she is implementing a block of type~$-1$; we elaborate on this later.
If Alice's block type is $+1$, all her unitary operations are
consecutive unitary operations from the original protocol (with the
index of the unitary operations depending on the number of $\pm 1$ in
$\FullMA$), while if it is $-1$, they are inverses of such
unitary operations. If Alice's block type is $0$, all unitary operations
are equal to the identity on registers $A C_\sA$. Similar properties hold for
Bob's unitary operations on registers
$BC$. Alice's block type corresponds to the content of the $j$-th
non-$\sC$ element in $\FullMA$, and Bob's to the content of the
$j$-th non-$\sC$ element in $\FullMB$. Alice's Pauli data corresponds to the content of the $j$-th block in $\FullPA$, and Bob's to the content of the $j$-th block in $\FullPB$. The precise rules by which Alice and Bob determine their respective types for a block in~$\Pi'$, and which blocks of~$\Pi$ (if any) are involved, are deferred to the next section. Note that when~$q_\MA = q_\MB$, the first $q_\MA$ MES blocks have been used by both parties but not necessarily in the same iterations. Nevertheless, the remedial actions the parties have taken to recover from out-of-sync teleportation have reduced the error on the joint state to transmission errors as if all the teleportations were synchronized and the adversary had introduced those additional errors; see Section\ref{sec:Out-of-sync teleportation}.

To give a concrete example, suppose from her classical data, Alice
determines that in \emph{her\/} $j$-th non-$\sC$ block of~$\Pi'$, she
should actively reverse the unitary operations of block~$k$ of~$\Pi$ to correct
some error in the joint state. So her~$j$-th non-$\sC$ block of~$\Pi'$ is
of type~$-1$. Suppose Alice's Pauli data in the~$j$-th block of $\FullPA$ correspond to Pauli operators~$p_{\sA,1} p_{\sA,2} \cdots p_{\sA,3r/2}$ in the order affecting the joint state. I.e., the Pauli operators~$p_{\sA,1} \,,\, p_{\sA,4} \,,\, \ldots \,,\, p_{\sA, 3 (r/2 -1) + 1}$ correspond to the sequence of Alice's teleportation measurement outcomes, the Pauli operators~$p_{\sA,2} \,,\, p_{\sA,5} \,,\, \ldots \,,\, p_{\sA, 3 (r/2 -1) + 2}$ are her teleportation decoding operations and~$p_{\sA,3} \,,\, p_{\sA,6} \,,\, \ldots \,,\, p_{\sA,3r/2}$ are her Pauli corrections, respectively.
Consider Bob's~$j$-th non-$\sC$ block of~$\Pi'$. Note that this may be a
different block of~$\Pi'$ than Alice's $j$-th non-$\sC$ block.
Suppose from \emph{his\/} classical data, Bob
determines that in his~$j$-th non-$\sC$ block of~$\Pi'$, he
should apply the unitary operations of block~$l$ of~$\Pi$ to evolve
the joint state further. So his~$j$-th non-$\sC$ block of~$\Pi'$ is
of type~$+1$. Suppose Bob's Pauli data in the~$j$-th block of $\FullPB$ correspond to Pauli operators~$p_{\sB,1} p_{\sB,2} \cdots p_{\sB,3r/2} $, in the order affecting the joint state. I.e., the Pauli operators~$p_{\sB,1} \,,\, p_{\sB,4} \,,\, \ldots \,,\, p_{\sB, 3 (r/2 -1) + 1}$ are Bob's decoding operations and~$p_{\sB,2} \,,\, p_{\sB,5} \,,\, \ldots \,,\, p_{\sB, 3 (r/2 -1) + 2}$ are his Pauli corrections and~$p_{\sB,3} \,,\, p_{\sB,6} \,,\, \ldots \,,\, p_{\sB,3r/2}$ correspond to his teleportation measurement outcomes, respectively.
Then from
$\FullMA,\FullMB,\FullPA,\FullPB$,
we can compute a description of the joint state as in
Eq.~\eqref{eqn:JS1}, with~$*j$ equal to
\suppress{
\begin{align}
\sigma_{B,en (i+1)r}U^{m_{b,i}}_{B,(i+1)r}\tilde{\sigma}_{ir+r-1}\ldots\hat{\sigma}_{ir+s}U^{m_{b,i}}_{B,ir+s}\tilde{\sigma}_{ir+s}U^{m_{a,i}}_{A,ir+s}\hat{\sigma}_{ir+s-1}\nonumber\\
\cdots\hat{\sigma}_{ir+1}U^{m_{b,i}}_{B,ir+1}\tilde{\sigma}_{ir+1}U^{m_{a,i}}_{A,ir+1}\sigma_{A,pc,ir+1}\sigma_{A,de,ir+1},
\end{align}
where
\[\tilde{\sigma}_{ir+s}=\sigma_{B,pc,ir+s}\sigma_{B,de,ir+s}\sigma_{A,en,ir+s},\]
and
\[\hat{\sigma}_{ir+s}=\sigma_{A,pc,ir+s+1}\sigma_{A,de,ir+s+1}\sigma_{B,en,ir+s}.\]
}
\begin{align*}
       & U_{kr + 1}^{-1} \\
\times & \left( p_{\sA, 3 (r/2 -1) + 3}\;\;  p_{\sA, 3(r/2 -1) + 2}
         \right)
         \left( p_{\sB, 3 (r/2 -1) + 3}\;\;  U_{lr + r} \;\;  
             p_{\sB, 3(r/2 -1) + 2} \;\;  p_{\sB, 3(r/2 -1) + 1}
         \right) \\
       & \qquad \times  \left( p_{\sA, 3 (r/2 -1) + 1} \;\;  U_{kr + 3}^{-1} 
         \right) \\
\times & \dotsb \\
\times & \left( p_{\sA, 3 (s-1) + 3}\;\;  p_{\sA, 3(s - 1) + 2}
         \right)
         \left( p_{\sB, 3 (s -1) + 3}\;\;  U_{lr + 2s} \;\; 
             p_{\sB, 3(s -1) + 2} \;\;  p_{\sB, 3(s -1) + 1}
         \right) \\
       & \qquad \times
         \left( p_{\sA, 3 (s -1) + 1} \;\;  U_{kr + (r - 2s + 3)}^{-1}
         \right) \\
\times & \dotsb \\
\times & \left( p_{\sA, 6}\;\;  p_{\sA, 5}
         \right)
         \left( p_{\sB, 6}\;\;  U_{lr + 4} \;\; 
             p_{\sB, 5} \;\;  p_{\sB, 4}
         \right)
         \left( p_{\sA, 4} \;\;  U_{kr + (r - 1)}^{-1}
         \right) \\
\times & \left( p_{\sA, 3}\;\;  p_{\sA, 2}
         \right)
         \left( p_{\sB, 3}\;\;  U_{lr + 2} \;\; 
             p_{\sB, 2} \;\;  p_{\sB, 1}
         \right)
         \left( p_{\sA, 1} \;\;  \id
         \right) \enspace.
\end{align*}
Note that Alice and Bob are not necessarily able to compute the state
$\JSone$. Instead, they use their best guess for the other party's metadata and Pauli data in the procedure described in this section to compute their estimates $\JSoneA$ and $\JSoneB$ of $\JSone$, respectively. Note that Alice and Bob will not compute their estimates of $\JSone$ unless they believe that they both know each other's metadata and Pauli data and have used the same number of MES blocks.   


\subsubsection{Second representation of the quantum registers}
\label{sec:large-classical-second-rep}

To obtain $\JStwo$ from $\JSone$, we first look inside each
bracket and recursively cancel consecutive Pauli operators inside the bracket.
In case a bracket evaluates to the identity operator on registers $A B C$, we remove it.
Once each bracket has been cleaned up in this way, we recursively try to
cancel consecutive brackets if their contents
correspond to the inverse of one another (assuming that no two $U_i$ of the original protocol are the
same or inverses of one another). Once no such cancellation works out anymore, what we are left with is
representation $\JStwo$, which is of the following form (when $q_{\MA} = q_{\MB}$):
\begin{align}\label{eqn:JS2_1}
\JStwo = [\#b] \cdots [\#1] [U_{gr} \cdots U_{(g-1)r + 2}
U_{(g-1)r + 1} ] \cdots [U_r \cdots U_2 U_1]
\ket{\psi_{\mathrm{init}}}^{A B C  E R}.
\end{align}
Here, the first $g$ brackets starting from the right correspond to the
``good'' part of the simulation, while the last $b$ brackets correspond
to the ``bad'' part of the simulation, the part that Alice and Bob have
to actively rewind later. The integer~$g$ is determined by the
left-most bracket such that along with its contents, those of the
brackets to the right equal the sequence of unitary 
operations~$U_1, U_2, \dotsc, U_{gr}$ from the original protocol~$\Pi$
in reverse. The brackets to the left of the last~$g$ brackets are all
considered bad blocks.
Thus, the content of $[\#1]$ is not $[U_{(g+1)r} \cdots U_{gr + 1}]$,
while the contents of $[\#2]$ to $[\#b]$ are arbitrary and have to be actively rewound before Alice and Bob can reverse the content of $[\# 1]$.

Once the two parties synchronize their metadata, the number of MESs they have used and their Pauli data, they compute their estimates of $\JSone$. Alice uses $\JSoneA$ in the above procedure to compute her estimate $\JStwoA$ of $\JStwo$. Similarly, Bob computes $\JStwoB$ from $\JSoneB$. These in turn determine their course of action in the simulation as described next. If $b > 0$, they actively reverse the incorrect unitary operators in the last bad block, while assuming the other party does the same. They start by applying the inverse of $[\#b]$, choosing appropriately whether to have a type $\pm 1$ or $0$ block, and also choosing appropriate Pauli corrections. Else, if $b=0$, they continue implementing unitary operations $U_{gr+1}$ to $U_{(g+1)r}$ of the original input protocol~$\Pi$ to evolve the simulation. Note that each player has their independent view of the joint state, and takes actions assuming that their view is correct. In this process, Alice and Bob use their view of the joint state to predict each other's next action in the simulation and extend their estimates of each other's metadata and Pauli data accordingly. 

We describe a few additional subtleties on how the parties access the
quantum register in a given block, as represented in Figure~\ref{fig:teleportation-representation}. First, each block begins and ends with
Alice holding register $C$ and being able to perform a unitary
operation. In $+1$ blocks, she applies a unitary operation at the beginning and
not at the end, whereas in $-1$ blocks she applies the inverse of a
unitary operation at the end and not at the beginning.
This is in order to allow a $-1$ block to be the inverse of a $+1$ block, and vice-versa. Second, whenever Alice
and Bob are not synchronized in the number of MESs they have used so far, as explained in Section~\ref{sec:Out-of-sync teleportation}, the party who has used more will wait for the other to catch up by creating a new type $\sC$ block while the party who has used less will try to catch up by creating a type $0$ block, sequentially feeding the $C$ register at the output of a teleportation decoding to the input of the next teleportation measurement. 
Notice that due to errors in communication, it might happen that $+1$ blocks are used to correct previous erroneous $-1$ blocks and $0$ blocks are used to correct previous erroneous $0$ blocks. As illustrated in Figure~\ref{fig:teleportation-representation}, the block on the right is the inverse of the one on the left if the corresponding Pauli operators are inverses of each other.

\subsubsection{Representations of quantum registers while out-of-sync}

We now define the $\JSone$ and $\JStwo$ representations of the joint state in the case when $q_{\MA} \neq q_{\MB}$. Note that in this case, conditioned on the classical data with the two parties, $\JSone$ and $\JStwo$ represent a pure state. However, in addition to the $ABCER$ registers, we must also include the half-used MES registers in the representation. Let $u\defeq |q_{\MA}-q_{\MB}|$. For concreteness, suppose that $q_{\MA} > q_{\MB}$. Then the $\JSone$ representation is of the following form:
\begin{align}\label{eqn:JS1-OoS}
\JSone = [*q_\MA] \cdots [*q_\MB] \cdots [*2] [*1]
\ket{\psi_{\mathrm{init}}}^{A B C  E R}\enspace.
\end{align}
The content of the first $q_\MB$ brackets from the right, corresponding to the MES blocks which have been used by both parties are obtained as described in Subsection~\ref{sec:large-classical-first-rep}. The leftmost $u$ brackets correspond to the MES blocks which have been used only by Alice. We refer to these blocks as the \emph{ugly\/} blocks. These brackets contain Alice's unitary operations from the input protocol, her teleportation decoding operations and Pauli correction operations in her last~$u$ non-classical iterations of the simulation. Additionally, they contain the $u$ blocks of MES registers used only by Alice. In each of these blocks, the registers indexed by an odd number have been measured on Alice's side and the state of the MES register has collapsed to a state which is obtained from Alice's Pauli data.

The representation $\JStwo$ is obtained from $\JSone$ as follows: We denote by $[@u]\cdots [@1]$ the leftmost $u$ brackets corresponding to the ugly blocks. We use the procedure described in Subsection~\ref{sec:large-classical-second-rep} on the rightmost $q_{\MB}$ brackets in $\JSone$ to obtain $\JStwo$ of the following form:
\begin{align} \label{eqn:JS2_OoS}
\JStwo = [@u] \cdots [@1] [\#b] \cdots [\#1] [U_{gr} \cdots U_{(g-1)r + 2} U_{(g-1)r + 1} ] 
\cdots [U_r \cdots U_2 U_1]
\ket{\psi_{\mathrm{init}}}^{A B C E R}\enspace,
\end{align}
with~$g$ \emph{good\/} blocks, and~$b$ \emph{bad\/} blocks, for some non-negative integers~$g,b$.

\suppress{
The representations are obtained in the manner described in Sections~\ref{sec:large-classical-first-rep} and~\ref{sec:large-classical-second-rep}, except that we account for the additional blocks of unitary operations along with blocks of MES halves $F_1, F_2, \ldots, F_u$ which arise from MES blocks that have been used by only one party. We refer to these blocks of unitary operators as the \emph{ugly\/} blocks. More formally,
we denote the final (leftmost)~$u$ blocks in the $\JSone$ representation as~$[@ u] \cdots [@ 1]$. These are the blocks of unitary operations performed in the last~$u$ non-classical iterations for Alice in the simulation and they contain untouched MES registers $F_1, F_2, \ldots, F_u$ on Bob's side.
Let $\JSone'$ denote the state obtained from $\JSone$ with these~$u$ blocks removed. Let $\JStwo'$ denote the state obtained from $\JSone'$ by following the procedure described in Section~\ref{sec:large-classical-second-rep}. It is of the form
\begin{align}\label{eqn:JS2goodbadform1} 
\JStwo' = [\#b] \cdots [\#1] [U_{gr} \cdots U_{(g-1)r+ 1 }] \cdots [U_r \cdots U_1] \ket{\psi_{init}}^{ABCER} \enspace,
\end{align}
with~$g$ ``good'' blocks, and~$b$ ``bad'' blocks, for some non-negative integers~$g,b$.
Then, $\JStwo$ is defined as $\JStwo\defeq [@ u] \cdots [@ 1] \JStwo'$.
}

Thus, in the rest of this section, we assume that $\JStwo$ is of the form of Equation~\eqref{eqn:JS2_OoS} at the end of each iteration
for some non-negative integers~$g,b,u$ which are given by
\begin{align}
&g \defeq
\text{the number of good unitary blocks in $\JStwo$,}\label{eqn:g}\\
&b \defeq
\text{the number of bad unitary blocks in $\JStwo$, and}\label{eqn:b}\\
&u\defeq |q_{\MA}-q_{\MB}|.\label{eqn:u}
\end{align}
%
%

We point out that Alice and Bob compute their estimates of $\JSone$ and $\JStwo$ only if, based on their view of the simulation so far, they believe that they have used the same number of MES blocks. Therefore, whenever computed, $\JSoneA,\JSoneB$ and $\JStwoA,\JStwoB$ are always of the forms described in Subsections~\ref{sec:large-classical-first-rep} and~\ref{sec:large-classical-second-rep}, respectively.

Notice that if there are no transmission errors or hash collisions and Alice and Bob do as described earlier in this section after realizing that $q_{\MA} > q_{\MB}$, then the ugly blocks $[@ u] \cdots [@ 2]$ remain as they were while block $[@ 1]$ becomes a standard block of unitary operations acting on registers $ABC$ only, quite probably being a new bad block, call it $[\#b+1]$. More generally, if there is either a transmission error or a hash collision, Bob might not realize that $q_{\MA} > q_{\MB}$. Then he might either have a $\sC$ type of iteration in which case block $[@ 1]$ also remain as is, or else it is a $+1$, $-1$ or $0$ (non-$\sC$) type of iteration and then he may apply non-identity Pauli operations and unitary operations on registers $BC$, which still results in block $[@ 1]$ becoming a standard block of unitary operations acting on registers $ABC$ only. Similarly if there is either a transmission error or a hash collision,  Alice might  not realize that $q_{\MA} > q_{\MB}$. Then she might have a non-$\sC$ type of iteration in which case a new ugly block, call it $[@ u+1]$, would be added to the left of $[@ u]$. 

\subsubsection{Summary of main steps}

The different steps that Alice and Bob follow in the simulation protocol~$\Pi'$ are summarized in Algorithm~\ref{summary-TP-based}. Recall that each party runs the simulation algorithm based on their view of the simulation so far. 

\suppress{
In one iteration of the simulation, only one step involving communication is conducted (and this constitutes one block of operations).}

\RestyleAlgo{boxruled}
\begin{algorithm}\label{summary-TP-based}
	
	 Agree on the history of the simulation contained in the metadata, i.e., ensure $\FullMA = \MAtilde$ and $\FullMB = \MBtilde$. This involves Algorithm \ref{algo:rewindMD}---\textbf{\textsf{rewindMD}}, and Algorithm \ref{algo:extendMD}---\textbf{\textsf{extendMD}}.\\
	 
	Synchronize the number of MESs used, in particular, ensure $q_{\MA} = q_{\MBtilde}$ and $q_{\MB} = q_{\MAtilde}$. This involves Algorithm \ref{algo:syncMES}---\textbf{\textsf{syncMES}}. \\
	
	Agree on Pauli data for all the teleportation steps and additional Pauli corrections for addressing channel errors, i.e., ensure $\FullPA = \PAtilde$ and $\FullPB = \PBtilde$. This is done via Algorithm \ref{algo:rewindPD}---\textbf{\textsf{rewindPD}} and Algorithm \ref{algo:extendPD}---\textbf{\textsf{extendPD}}. \\
	
	Compute the best guess for $\JSone$ and $\JStwo$. If there are any ``bad'' blocks in the guess for $\JStwo$, reverse the last bad block of unitary operations. I.e., implement quantum rewinding so that~$b = 0$ in $\JStwo$. This is done in Algorithm \ref{algo:simulate}---\textbf{\textsf{simulate}}. \\
	
	If no ``bad'' blocks remain, implement the next block of the original protocol. This results in an increase in $g$ in $\JStwo$, and is also done through Algorithm \ref{algo:simulate}---\textbf{\textsf{simulate}}. \\

	\caption{Main steps in one iteration of the simulation
for the large alphabet teleportation-based model}
	\label{algo:Main steps-large-alphabet-cleve-burhman}
\end{algorithm}
\RestyleAlgo{ruled}
The algorithms mentioned in each step are presented in the next section. Figure~\ref{fig:flow-telep} summarizes the main steps in flowchart form.

In every iteration exactly one of the steps listed in Algorithm~\ref{summary-TP-based} is conducted. Alice and Bob skip one step to the next only if the goal of the step has been achieved through the \emph{previous\/} iterations. The simulation protocol is designed so that unless there is a transmission error or a hash collision in comparing a given type of data, Alice and Bob will go down these steps in tandem, while never returning to a previous step. For instance, once Alice and Bob achieve the goal of step $1$, as long as no transmission error or hash collision occurs, their metadata will remain synchronized while they are conducting any of the next steps. This is in fact a crucial property which we utilize in the analysis of the algorithm. In particular, to ensure this property, Alice and Bob need to synchronize the number of MESs they have used \emph{before\/} synchronizing their Pauli data.

\begin{figure}[!t]
\centering
\includegraphics[width=475pt]{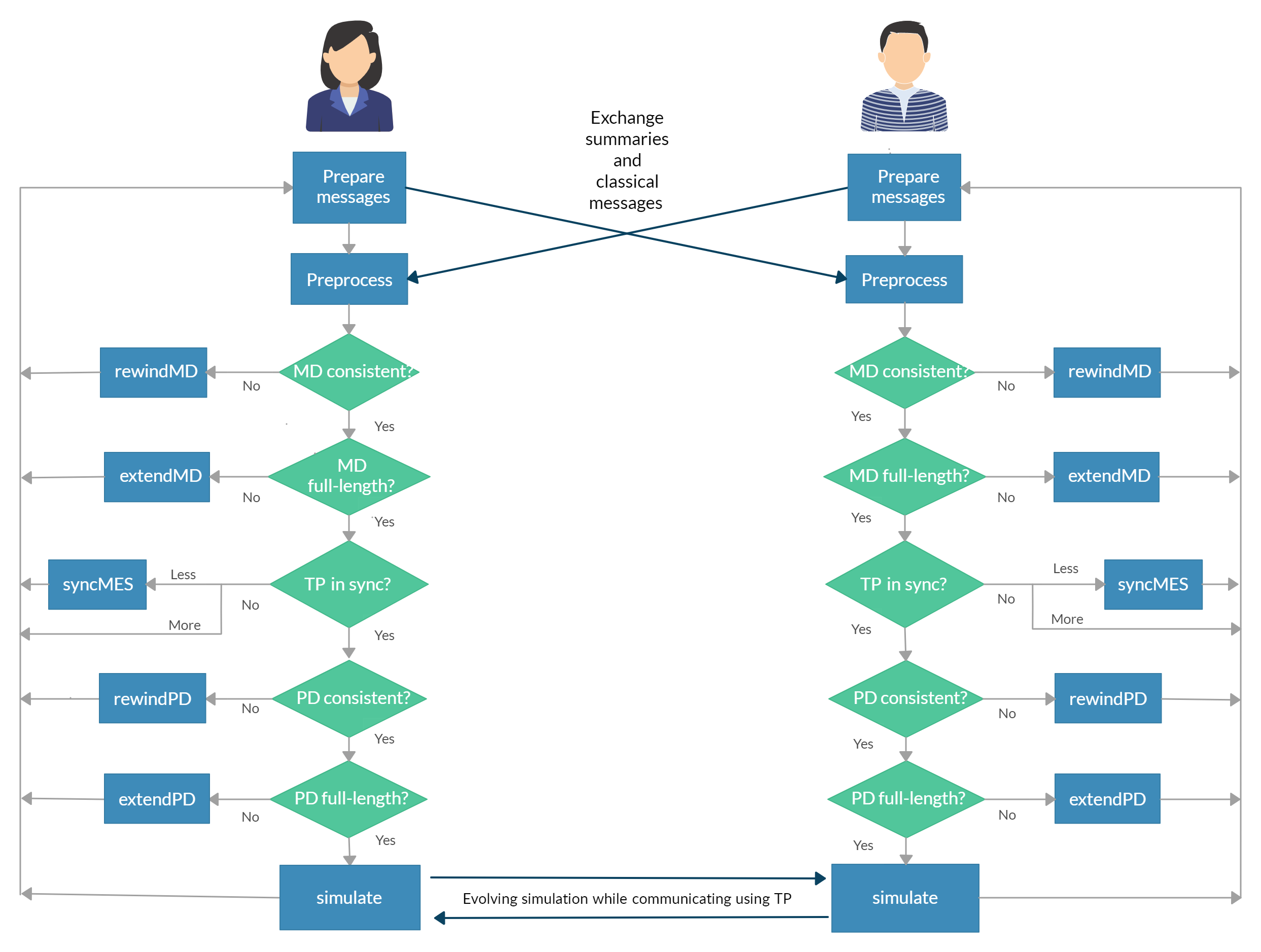}
\caption{Flowchart of the teleportation-based scheme for high rate noisy
interactive quantum communication. Most of the communication is spent
actually trying to simulate the protocol, in the \textsf{simulate} subroutine.}
\label{fig:flow-telep}
\end{figure}

\suppress{
Notice that unless there is a transmission error or a hash collision in
comparing a given type of data (as in Ref.~\cite{Haeupler:2014}),
Alice and Bob cycle through these steps in tandem.}

\subsection{Algorithm}\label{subsec:algforlargealphabet}

In this section, we present our simulation protocol $\Pi'$ in the teleportation-based model when the communication alphabet is polynomial-size. We first introduce the data structure used in our algorithm in this model, which summarizes the definition of the variables appearing in the pseudocodes.

\subsubsection{Data structure}\label{sec:datastructurelargealphabetcleveburhman}

\begin{itemize}
	\item \textbf{Metadata:} In every iteration $\NewMetaA
\in {\{\pm 1,0,\sC\}}$ corresponds to Alice's block type which determines how the simulation of the input
protocol proceeds locally on Alice's side. $\NewMetaA=\sC$ corresponds to a classical iteration, in which Alice does not access the quantum registers. $\NewMetaA \in {\{\pm 1,0\}}$ determines the exponent of the unitary operators from the input protocol $\Pi$ applied by Alice in the current iteration of the simulation. Alice records her metadata in $\FullMA$ which is concatenated with $\NewMetaA$ in every iteration and has length $i$ after $i$ iterations. Her best guess of Bob's block type in the current iteration is denoted by $\NewMetaBtilde$. Alice maintains a guess for Bob's metadata in $\MBtilde$ which gets modified or corrected as she gains more information through interaction with Bob. Note that $\MBtilde$ is not necessarily full-length in every iteration and its length may decrease. $\ellMBtilde$ denotes the length of $\MBtilde$. Bob's local data, $\NewMetaB$, $\FullMB$, $\NewMetaAtilde$, $\MAtilde$ and $\ellMAtilde$ are defined similarly.

Alice maintains a guess~$\ellMA$ for the length of~$\MAtilde$, which is with Bob. We define~$\MA$ to be the prefix of $\FullMA$ of length~$\ellMA$, i.e., ~$\MA \defeq \FullMA\Br{1:\ellMA}$. When $\MA$ appears in any of the algorithms in this section, it is implicitly computed by Alice from $\FullMA$ and $\ellMA$. The number of MES blocks used by Alice for teleportation is denoted by $q_{\MA}$. We use $q_{\MBtilde}$ to denote Alice's guess of the number of MES blocks used by Bob. 
Note that $q_{\MA}$ and $q_{\MBtilde}$ are the number of $0$, $1$ and $-1$ symbols in $\MA$ and $\MBtilde$, respectively. Bob's $\MB$, $\ellMB$, $q_{\MB}$ and $q_{\MAtilde}$ are defined similarly.
	
\item \textbf{Pauli data:} In every iteration $\NewPauliA \in {\br{\Sigma^{r}}^3}$ consists of three parts: The first part corresponds to the outcomes of Alice's teleportation measurements in the current iteration; the second part corresponds to the received transmissions which determine the teleportation decoding operation and the last part which corresponds to Pauli corrections. 

The Pauli data are recorded locally by Alice in $\FullPA$. Starting from the empty string, $\FullPA$ is concatenated with $\NewPauliA$ whenever Alice implements a non-$\sC$ iteration. Alice's best guess for Bob's $\NewMetaB$ in each iteration is denoted by $\NewMetaBtilde$. She maintains a string $\PBtilde$ as an estimate of Bob's Pauli data. The length of $\PBtilde$ is denoted by $\ellPBtilde$. Alice also maintains $\ellPA$, her estimate for the length of $\PAtilde$, which is with Bob. $\PA$ denotes the prefix of $\FullPA$ of length $\ellPA$, i.e., $\PA\defeq \FullPA\Br{1:\ellPA}$. When $\PA$ appears in any of the algorithms in this section, it is implicitly computed by Alice from $\FullPA$ and $\ellPA$. Bob's local Pauli data $\NewPauliB\;,\;\FullPB\;,\;\NewPauliAtilde\;,\;\PAtilde\;,\;\ellPAtilde\;,\;\ellPB\;,\;\PB$ are defined similarly.

A critical difference between the metadata and the Pauli data is that the metadata assigns one symbol for each block while the Pauli data assigns $3r$ symbols for each block.
	
	\item We use $H$ with the corresponding data as subscript to denote the hashed data, e.g., $\HMA$ denotes the hash value of the string $\MA$.
	
	\item  The data with ~$'$~ denote the received data after transmission over the noisy channel, e.g., $\ellMB'$ denotes what Alice receives when Bob sends $\ellMB$.
		
	\item The variable $\Itertype \in \{\MD,\PD,\MES,\SIM\}$ determines the iteration type for the party: $\MD$ and $\PD$ correspond to iterations where metadata and Pauli data are processed or modified, $\MES$ is used for iterations where the party is trying to catch up on the number of used MESs, and $\SIM$ corresponds to iterations where the party proceeds with evolving the simulation of $\Pi$ by applying a block of unitary operators from $\Pi$ or the inverse of such a block of unitary operators in order to fix an earlier error.
	
	\item The variable $\RewindExtend \in \{\sR,\sE\}$ determines in classical iterations if a string of the local metadata or Pauli data is extended or rewound in the current iteration.

\end{itemize}

\subsubsection{Pseudo-codes}
This section contains the pseudo-codes for the main algorithm and the subroutines that each party runs locally in the simulation protocol. The subroutines are the following: \textsf{Preprocess}, which determines what will happen locally to the classical and quantum data in the current iteration of the simulation; \textsf{rewindMD} and \textsf{extendMD}, which process the local metadata; \textsf{syncMES} which handles the case when the two parties do not agree on the number of MES blocks they have used; \textsf{rewindPD} and \textsf{extendPD}, process the local Pauli data; and finally, \textsf{simulate}, in which the player moves on with the simulation of the input protocol according to the information from subroutine \textsf{Computejointstate} of \textsf{Preprocess}.   
 When the party believes that the classical data are fully synchronized, he or she uses the subroutine \textsf{Computejointstate} to extract the necessary information to decide how to evolve the joint quantum state next. This information includes $\JSone$ and $\JStwo$ defined in~\eqref{eqn:JS1} and~\eqref{eqn:JS2_1}, respectively, $\NewMetaA$, $\RewindExtend$, $\NewMetaBtilde$, $\mathit{Block}$ which represents the index of the block of unitary operations from the input protocol $\Pi$ the party will perform, $\PCorr$ representing Alice's Pauli corrections and $\PCorrtilde$ representing Alice's guess of Bob's Pauli corrections.
 \suppress{
 To get $\JStwo$ from $\JSone$, Alice and Bob follow the procedure explained in Section \ref{sec:general-discription-large-alphabet-cleveburhman}. The computed joint state will be of the following form
 \begin{equation}\label{eqn:JS2}
 \JStwo=\prod_i\br{\sigma_{\br{2r+1}i}U^{b_i}_{B,\br{i+1}r}\sigma_{\br{2r+1}i-1}U^{a_i}_{A,\br{i+1}r}\sigma_{\br{2r+1}i-1}\ldots\sigma_{2ri+3}U^{b_i}_{B,ir+1}\sigma_{2ri+2}U^{a_i}_{A,ir+1}\sigma_{2ri+1}}\ket{\Psi_{\mathrm{init}}}.
 \end{equation}
 }
 
 For the subroutines used in the simulation protocol, we list all the global variables accessed by the subroutine as the \textbf{Input} at the beginning of the subroutine. Whenever applicable, the relation between the variables when the subroutine is called is stated as the \textbf{Promise} and the global variables which are modified by the subroutine are listed as the \textbf{Output}. 
 
 \begin{remark}
 	The amount of communication in each iteration of Algorithm \ref{algo:Mainalgorithm} is independent of the iteration type.
 \end{remark}
 \begin{remark}
 	Since in every iteration of Algorithm \ref{algo:Mainalgorithm} the lengths of $\FullMA$ and $\FullMB$ increase by $1$, in order to be able to catch up on the metadata, Alice and Bob need to communicate two symbols at a time when extending the metadata. This is done by encoding the two symbols into strings of length $r$ of the channel alphabet $\Sigma$ using the mapping encodeMD in Algorithm \ref{algo:Preprocess} and decoding it using the mapping decodeMD in Algorithm \ref{algo:extendMD}.
 \end{remark}

\begin{algorithm}
    \Input{$n$ round protocol $\Pi$ in teleportation-based model over polynomial-size alphabet $\Sigma$}
	\BlankLine
	
	Initialize \\
	\nonl  $\qquad r \leftarrow \Theta\br{1/\sqrt{\epsilon}}$ \; 
	\nonl  $\qquad \Rtotal \leftarrow \lceil \frac{n}{2r}+\Theta(n\eps)\rceil$ \;
	\nonl  $\qquad q_{\MA},\ellMA,\ellMBtilde,\ellPA,\ellPBtilde\leftarrow 0$ \; \nonl  $\qquad \MA,\MBtilde,\PA,\PBtilde\leftarrow \emptyset$ \;	
	\BlankLine
	
	$h\leftarrow$ hash function of Lemma~\ref{lem:hashes} with $p=1/n^5$ and $o=s=\Theta(\log n)$ \;
	Measure $\Theta\br{\Rtotal}$ MESs in the computational basis and record the binary representation of the outcomes in $S_1,\ldots,S_{4\Rtotal}$ \;
	\tcp*[f]{\textbf{$4\Rtotal$ seeds of length $s$ for the hash function $h$}}\\
	
	\SetKwProg{ForLoop}{For}{}{}
	\SetAlgoLined
	\ForLoop{$i = 1 \to \Rtotal$}	
	    {   \SetAlgoVlined
	        \Comment*[f] {\textbf{Preprocessing phase}}\\
	
		    $\HMA \leftarrow h_{S_{4i-3}}\br{\MA}$ \; 
		    $\HMBtilde \leftarrow h_{S_{4i-2}}\br{\MBtilde}$ \; 
		    $\HPA \leftarrow h_{S_{4i-1}}\br{\PA}$ \; 
		    $\HPBtilde \leftarrow h_{S_{4i}}\br{\PBtilde}$ \;
		
	    }	

		\caption{\textbf{\textsf{Main algorithm }}(Alice's side)}\label{algo:Mainalgorithm}
\end{algorithm}

\setcounter{algocf}{2}

\begin{algorithm}
\setcounter{AlgoLine}{8}
    \SetKwBlock{Begin}{}{}
	\Begin{
        Send $$\br{\HMA,\ellMA,\HMBtilde,\ellMBtilde,\HPA,\ellPA,\HPBtilde,\ellPBtilde};$$\\
		
		Receive $$\br{\HMAtilde',\ellMAtilde',\HMB',\ellMB',\HPAtilde',\ellPAtilde',\HPB',\ellPB'};$$\\
		
		\BlankLine
		\textbf{\textsf{Preprocess}}\;
		
		\If   {$\Itertype \neq \SIM$}
		    {   Send $\msg$\;
			    Receive $\msg'$\;
		    }
		\tcp*[f]{\textbf{messages are communicated alternately}}\\  
		
	\BlankLine
	\Comment*[f] {\textbf{Case i.A}}\\
		\If   {$\Itertype = \MD \;\mathrm{and}\; \RewindExtend = \sR$}
		    {   \textbf{\textsf{rewindMD}}\;
		    }	
	\Comment*[f] {\textbf{Case i.B}}\\
	    \ElseIf {$\Itertype = \MD \;\mathrm{and}\; \RewindExtend = \sE$}
	        {   \textbf{\textsf{extendMD}}\;
	        }
    \Comment*[f] {\textbf{Case ii.A}}\\
        \ElseIf{$\Itertype = \MES \;\mathrm{and}\; \NewMetaA=\sC$}
            {   return\;
            }
    \Comment*[f] {\textbf{Case ii.B}}\\
        \ElseIf{$\Itertype = \MES \;\mathrm{and}\; \NewMetaA=\mathsf{0}$}
        	{   \textbf{syncMES}\;
        	}
    \Comment*[f] {\textbf{Case iii.A}}\\
        \ElseIf   {$\Itertype = \PD \;\mathrm{and}\; \RewindExtend = \sR$}
            {   \textbf{\textsf{rewindPD}}\;
            }
    \Comment*[f] {\textbf{Case iii.B}}\\
        \ElseIf{$\Itertype = \PD \;\mathrm{and}\; \RewindExtend = \sE$}
            {   \textbf{\textsf{extendPD}}\;
            }

	    \tcp*[f] {\textbf{Classical data are synchronized}}\\
	\Comment*[f] {\textbf{Case iv}}\\	
	    \Else{\textbf{\textsf{simulate}}.}
    } 
	\Return{\textup{\textbf{\textsf{Main algorithm}}}}\;
	\caption{\textbf{\textsf{Main algorithm }}(Alice's side, cont. from previous page)}
\end{algorithm}

\begin{algorithm}
	
	\Input{ 
		    $$\br{ 
		    	  \begin{array}{c}
			        \HMA,\ellMA,\HMBtilde,\ellMBtilde,\HPA,\ellPA,\HPBtilde,\ellPBtilde \\
			        \HMAtilde',\ellMAtilde',\HMB',\ellMB',\HPAtilde',\ellPAtilde',\HPB',\ellPB' \\
		            \FullMA,\MBtilde,\FullPA,\PBtilde,q_{\MA}
		          \end{array}
		         }$$
	      }

	\Output{ $\br{\Itertype,\RewindExtend,\NewMetaA,\FullMA,\ellMA, \NewMetaBtilde,\ellMBtilde,\msg}$ }
	
	\BlankLine
	\If   
	    {  
	        $\br{\HMA,\HMBtilde}=\br{\HMAtilde',\HMB'}  \;\mathrm{and}\; \ellMA=\ellMAtilde'=\ellMBtilde=\ellMB'=i-1$
	    }
	    {   Compute $q_{\MBtilde}$\;
	    }

	\Comment*[f] {\textbf{Processing metadata}}\\
	\Comment*[f] {\textbf{Case i.A}}\\
	\If
	   {
	   	   $\br{\HMA,\HMBtilde,\ellMA,\ellMBtilde}\neq \br{\HMAtilde',H'_{\mathrm{MB}},\ellMAtilde',\ellMB'}$
       }
	   {   $\Itertype \leftarrow \MD$\;
	   	   $\RewindExtend \leftarrow \sR$\;
	   	   $\NewMetaA \leftarrow \sC$\;
	   	   $\FullMA \leftarrow \left(\FullMA,\NewMetaA\right)$\;	
	   	   $\msg \leftarrow \text{dummy message of length } r$\;
	   }
   \Comment*[f] {\textbf{Case i.B}}\\
	\ElseIf
	   {   $\br{\ellMA < i-1} \;\mathrm{or}\; \br{\ellMBtilde < i-1}$
	   }
	   {   $\Itertype \leftarrow \MD$\;
	       $\RewindExtend \leftarrow \sE$\;
	   	   $\NewMetaA \leftarrow \sC$\;
	   	   $\FullMA \leftarrow \left(\FullMA,\NewMetaA\right)$\;
	   	   \If   {$\ellMA < i-1$}
	             {   $\msg \leftarrow \mathrm{encodeMD}\br{\FullMA\Br{\ellMA+1,\ellMA+2}}\!;$ 
	             	 \tcp*[f]{\textbf{Encode MD in $\Sigma^r$}}\\
	             }
	        \Else
	             {   $\msg \leftarrow \text{dummy message of length } r$\;
	             }      	
	   }	
	\Comment*[f] {\textbf{Comparing number of used MES blocks}}\\
	\Comment*[f] {\textbf{Case ii.A}}\\
	\ElseIf   {$q_{\MA} > q_{\MBtilde}$}
	   {   $\Itertype \leftarrow \MES$\;
	   	   $\NewMetaA \leftarrow \sC$\;
	   	   $\FullMA \leftarrow \left(\FullMA,\NewMetaA\right)$\;
           $\ellMA \leftarrow \ellMA+1$\;
	   	   $\NewMetaBtilde \leftarrow 0$\;
           $\MBtilde \leftarrow \br{\MBtilde,\NewMetaBtilde}$\;
           $\ellMBtilde \leftarrow \ellMBtilde+1$\;
		   $\msg \leftarrow \text{dummy message of length } r$\;
	   }
	\caption{\textbf{\textsf{Preprocess }} (Alice's side)}
\label{algo:Preprocess}
\end{algorithm}

\setcounter{algocf}{3}

\begin{algorithm}
\setcounter{AlgoLine}{25}
   \Comment*[f] {\textbf{Case ii.B}}\\
	\ElseIf   {$q_{\MA} < q_{\MBtilde}$}	
	   {   $\Itertype \leftarrow \MES$\;
	   	   $\NewMetaA \leftarrow \mathsf{0}$\;
	   	   $\FullMA \leftarrow \left(\FullMA,\NewMetaA\right)$\;
           $\ellMA \leftarrow \ellMA+1$\;
		   $\NewMetaBtilde \leftarrow \sC$\;
           $\MBtilde \leftarrow \br{\MBtilde,\NewMetaBtilde}$\;
           $\ellMBtilde \leftarrow \ellMBtilde+1$\;
		   $\msg \leftarrow \text{dummy message of length }r$\;
	   }	

	\Comment*[f] {\textbf{Processing Pauli data}}\\
	\Comment*[f] {\textbf{Case iii.A}}\\
	\ElseIf   {$\br{\HPA,\HPBtilde,\ellPA,\ellPBtilde}\neq \br{\HPAtilde',\HPB',\ellPAtilde',\ellPB'}$
	      }
	   {   $\Itertype \leftarrow \PD$\;
	       $\RewindExtend \leftarrow \sR$\;
	   	   $\NewMetaA \leftarrow \sC$\;
	       $\FullMA \leftarrow \left(\FullMA,\NewMetaA\right)$\;
           $\ellMA \leftarrow \ellMA+1$\;
           $\NewMetaBtilde \leftarrow \sC$\;
           $\MBtilde \leftarrow \br{\MBtilde,\NewMetaBtilde}$\;
           $\ellMBtilde \leftarrow \ellMBtilde+1$\;
		   $\msg \leftarrow \text{dummy message of length }r$\;
	   }
	 \Comment*[f] {\textbf{Case iii.B}}\\
	\ElseIf   {$\br{\ellPA < 3q_{\MA} \cdot r} \;\mathrm{or}\; \br{\ellPBtilde <
			3q_{\MBtilde} \cdot r}$}
	{   $\Itertype \leftarrow \PD$\;
		$\RewindExtend \leftarrow \sE$\;
		$\NewMetaA \leftarrow \sC$\;
		$\FullMA \leftarrow \left(\FullMA,\NewMetaA\right)$\;
        $\ellMA \leftarrow \ellMA+1$\;
        $\NewMetaBtilde \leftarrow \sC$\;
        $\MBtilde \leftarrow \br{\MBtilde,\NewMetaBtilde}$\;
        $\ellMBtilde \leftarrow \ellMBtilde+1$\;
		\If   {$\ellPA < 3q_{\MA} \cdot r$}
		{   $\msg \leftarrow {\FullPA}\Br{\ellPA+1,\ellPA+r}$
		}	
	}
	\caption{\textbf{\textsf{Preprocess}} (Alice's side, cont. from previous page)}
\end{algorithm}

\setcounter{algocf}{3}

\begin{algorithm}
\setcounter{AlgoLine}{55}
	\Comment*[f] {\textbf{Processing joint quantum state}}\\
	\Comment*[f] {\textbf{Case iv}}\\
	\Else
	{   $\Itertype \leftarrow \SIM$\;
		\textsf{\textbf{computejointstate}}\;
		$\FullMA=\left(\FullMA,\NewMetaA\right)$\;
        $\ellMA \leftarrow \ellMA+1$\;
        $\MBtilde \leftarrow \br{\MBtilde,\NewMetaBtilde}$\;
        $\ellMBtilde \leftarrow \ellMBtilde+1$\;
	}	
	
	\Return{\textup{\textbf{\textsf{Preprocess}}}}\;
	\caption{\textbf{\textsf{Preprocess}} (Alice's side, cont. from previous page)}
\end{algorithm}

\begin{algorithm}
	
	\Input{$\br{\HMA,\ellMA,\HMBtilde,\ellMBtilde,\HMAtilde',\ellMAtilde',\HMB',\ellMB'}$}
	
	\Promise{$\br{\HMA,\HMBtilde,\ellMA,\ellMBtilde}\neq \br{\HMAtilde',\HMB',\ellMAtilde',\ellMB'}$.}
	
	\Output{$\br{\ellMA,\ellMB'}$}
		
	\If   {$\ellMA \neq \ellMAtilde' \;\mathrm{or}\; \ellMBtilde \neq \ellMB'$}
	{   \If   {$\ellMA > \ellMAtilde'$}
		{$\ellMA \leftarrow \ellMA-1$\;}		
		\If   {$\ellMBtilde > \ellMB'$}
		{$\ellMBtilde \leftarrow \ellMBtilde-1$\;}		
	}
	\Else   {
		\If   {$\HMA \neq \HMAtilde'$}
		{$\ellMA\leftarrow\ellMA-1$\;}
		\If   {$\HMBtilde \neq \HMB'$}
		{$\ellMBtilde\leftarrow\ellMBtilde-1$\;}		
	}
	
	\Return{\textup{\textbf{\textsf{rewindMD}}}}\;
	\caption{\textbf{\textsf{rewindMD}} (Alice's side)}
	\label{algo:rewindMD}
\end{algorithm}

\begin{algorithm}
	\Input{$\br{\ellMA,\ellMBtilde,\MBtilde,\msg',i}$}
	
	\Promise{
$\br{\HMA,\HMBtilde,\ellMA,\ellMBtilde}=\br{\HMAtilde',\HMB',\ellMAtilde',\ellMB'}$,
	$\ellMA < i-1 \quad \mathrm{or} \quad \ellMBtilde < i-1$.
}
	
	\Output{$\br{\ellMA,\MBtilde,\ellMBtilde}$}
		
	\If   {$\ellMA < i-1$}
	   {   $\ellMA \leftarrow \ellMA+2$\;          	  	
	   }
	\ElseIf    {$\ellMA = i-1$}
	   {    $\ellMA \leftarrow \ellMA+1$\;
	   }
	\If   {$\ellMBtilde < i-1$}
	   {   $\MBtilde\Br{\ellMBtilde+1,\ellMBtilde+2} \leftarrow \mathrm{decodeMD}\br{\msg'}\!;$        	
	   \tcp*[f]{\textbf{decode MD from $\Sigma^r$}}\\
	   	   $\ellMBtilde \leftarrow \ellMBtilde+2$\;          	  	
	   }	
	\ElseIf   {$\ellMBtilde=i-1$}
	    {   $\MBtilde \leftarrow \br{\MBtilde,\sC}$\;        	
	   		$\ellMBtilde \leftarrow \ellMBtilde+1$\;          	  	
	   	}
	
	\Return{\textup{\textbf{\textsf{extendMD}}}}\;
	\caption{\textbf{\textsf{extendMD}} (Alice's side) }
	\label{algo:extendMD}
\end{algorithm}

\begin{algorithm}
	
	\Input{$\br{\FullPA,q_{\MA}}$}
	
    \Promise{ $\br{\HMA,\HMBtilde,\ellMA,\ellMBtilde}=\br{\HMAtilde',\HMB',\ellMAtilde'+1,\ellMB'+1}$,
    $\ellMA=\ellMBtilde=i \;,\; q_{\MA} < q_{\MBtilde}$.}
	
	\Output{$q_{\MA},\NewPauliA,\FullPA$}
	
	Recall that~$A' B' C'$ are the registers that are used to generate the joint quantum state of the protocol being simulated, and~$C'$ is the communication register\;

	Let~$E_1 E_2 \dotsb E_r$ be the~$r$ registers with Alice containing halves of the block of~$r$ MESs with indices in the interval~$(q_{\text{MA}} \cdot r, \: (q_{\text{MA}} + 1)\cdot r]$ \;

	Teleport~$C'$ using~$E_1$; then teleport~$E_2$ using~$E_3$, $E_4$ using~$E_5$, and so on (i.e., teleport~$E_j$ using $E_{j+1}$ for even~$j \in [r-2]$), and then store~$E_r$ in register~$C'$ \;
	\tcp*[f]{\textbf{See Section~\ref{sec:Out-of-sync teleportation} for the rationale, and Bob's analogue of this step}} \\


	Store the teleportation measurement outcomes in $m\in{\Sigma}^r$\;

	$\NewPauliA \leftarrow \br{m,0^r,0^r}$\;
    $\FullPA \leftarrow \br{\FullPA,\NewPauliA}$\;
    $q_{\MA} \leftarrow q_{\MA}+1$\;
		
	\Return{\textup{\textbf{\textsf{syncMES}}}}\;
	\caption{\textbf{\textsf{syncMES}} (Alice's side)}
	\label{algo:syncMES}
\end{algorithm}

\begin{algorithm}
		
	\Input{$\br{\HPA,\ellPA,\HPBtilde,\ellPBtilde,\HPAtilde',\ellPAtilde',\HPB',\ellPB'}$}
	
	\Promise{$\br{\HMA,\HMBtilde,\ellMA,\ellMBtilde}=\br{\HMAtilde',\HMB',\ellMAtilde'+1,\ellMB'+1}$ , 
    $\ellMA=\ellMBtilde=i \;,\; q_{\MA}=q_{\MBtilde}$ , 
    $\br{\HPA,\HPBtilde,\ellPA,\ellPBtilde}\neq \br{\HPAtilde',\HPB',\ellPAtilde',\ellPB'}$.}
	
	\Output{$\br{\ellPA,\ellPBtilde}$}
	
	\If   {$\ellPA \neq \ellPAtilde' \;\mathrm{or}\; \ellPBtilde \neq \ellPB'$}
	{   \If   {$\ellPA > \ellPAtilde'$}
		{$\ellPA \leftarrow \ellPA-r$\;}		
		\If   {$\ellPBtilde > \ellPB'$}
		{$\ellPBtilde \leftarrow \ellPBtilde-r$\;}		
	}
	\Else   {
		\If   {$\HPA \neq \HPAtilde'$}
		{$\ellPA \leftarrow \ellPA-r$\;}
		\If   {$\HPBtilde \neq \HPB'$}
		{$\ellPBtilde \leftarrow \ellPBtilde-r$\;}		
	}	
	
	\Return{\textup{\textbf{\textsf{rewindPD}}}}\;
	\caption{\textbf{\textsf{rewindPD}} (Alice's side)}
	\label{algo:rewindPD}
\end{algorithm}

\begin{algorithm}
	
	\Input{$\br{\ellPA,\ellPBtilde,\PBtilde,q_{\MA},q_{\MBtilde},\msg'}$}
	
	\Promise{$\br{\HMA,\HMBtilde,\ellMA,\ellMBtilde}=\br{\HMAtilde',\HMB',\ellMAtilde'+1,\ellMB'+1}$ , 
    $\ellMA=\ellMBtilde=i \;,\; q_{\MA}=q_{\MBtilde}$ ,  $\br{\HPA,\HPBtilde,\ellPA,\ellPBtilde}= \br{\HPAtilde',\HPB',\ellPAtilde',\ellPB'}$ , 
		$\ellPA < 3q_{\MA} \cdot r \quad \mathrm{or} \quad \ellPBtilde < 3q_{\MBtilde} \cdot r$.}
	
	\Output{$\br{\ellPA,\PBtilde,\ellPBtilde}$}
	
	\If   {$\ellPA < 3q_{\MA} \cdot r$}
  	    {$\ellPA \leftarrow \ellPA+r$\;}
	\If   {$\ellPBtilde < 3q_{\MBtilde} \cdot r$}
      	{   $\PBtilde\Br{\ellPBtilde+1:\ellPBtilde+r} \leftarrow \msg'$\;
		$\ellPBtilde \leftarrow \ellPBtilde+r$\;
	    }
	
	\Return{\textup{\textbf{\textsf{extendPD}}}}\;
	\caption{\textbf{\textsf{extendPD}} (Alice's side) }
	\label{algo:extendPD}
\end{algorithm}

\begin{algorithm}
	
	\Input{$\br{\FullMA,\MBtilde,\FullPA,\PBtilde}$}
	
	\Promise{$\br{\HMA,\HMBtilde,\ellMA,\ellMBtilde}=\br{\HMAtilde',\HMB',\ellMAtilde',\ellMB'}$,
    $\ellMA=\ellMBtilde=i-1 \;,\; q_{\MA}=q_{\MBtilde}$, $\br{\HPA,\HPBtilde,\ellPA,\ellPBtilde}= \br{\HPAtilde',\HPB',\ellPAtilde',\ellPB'}$,
		$\ellPA= \ellPAtilde'=3q_{\MA} \cdot r\;,\; \ellPBtilde =\ellPB'=3q_{\MBtilde} \cdot r$.}
	
	\Output{$\br{\JSoneA, \JStwoA,\NewMetaA,\NewMetaBtilde, \mathit{Block},\RewindExtend,\PCorr,\PCorrtilde}$}
	
	Compute $\JSoneA$\;
	Compute $\JStwoA$\;
	Compute $\NewMetaA$\;
    Compute	$\RewindExtend$\;
	Compute $\NewMetaBtilde$\;
	Compute $\mathit{Block}$\;
	Compute $\PCorr$\;
	Compute $\PCorrtilde$\;
	\tcp*[f]{\textbf{Refer to Sections~\ref{sec:large-classical-first-rep},~\ref{sec:large-classical-second-rep} to see how these variables are computed}}
	
	\Return{\textup{\textbf{\textsf{Computejointstate}}}}\;
	\caption{\textbf{\textsf{Computejointstate}} (Alice's side)}
	\label{algo:Computejointstate}
\end{algorithm}

\begin{algorithm}
	
	\Input{$\br{q_{\MA},\FullPA,\ellPA,\PBtilde,\ellPBtilde,\RewindExtend,\NewMetaA,\mathit{Block},\PCorr,\PCorrtilde}$}
	
	\Promise{$\br{\HMA,\HMBtilde,\ellMA,\ellMBtilde,q_{\MA}}=\br{\HMAtilde',\HMB',\ellMAtilde'+1,\ellMB'+1,q_{\MBtilde}}$,
    $\ellMA=\ellMBtilde=i$, $\br{\HPA,\HPBtilde,\ellPA,\ellPBtilde}= \br{\HPAtilde',\HPB',\ellPAtilde',\ellPB'}$, 
	$\ellPA=\ellPBtilde=3q_{\MA} \cdot r$}
	
	\Output{$\br{\FullPA,\ellPA,\PBtilde,\ellPBtilde}$}
	
	Continue the simulation of the input protocol according to $\mathit{Block}$, $\NewMetaA$ and $\PCorr$\;
	Record all teleportation measurement outcomes in $\alpha$\;
	Record all received Bob's teleportation measurement outcomes in $\beta$\;
	$\NewPauliA \leftarrow \br{\alpha,\beta,\PCorr}$\;
	$\FullPA \leftarrow \br{\FullPA,\NewPauliA}$\;
	$\ellPA \leftarrow \ellPA+3r$\;
	$\NewPauliBtilde\leftarrow \br{\beta,\alpha,\PCorrtilde}$\;
	$\PBtilde \leftarrow \br{\PBtilde,\NewPauliBtilde}$\;
	$\ellPBtilde \leftarrow \ellPBtilde+3r$\;
	$q_{\MA} \leftarrow q_{\MA}+1$\;

	\Return{\textup{\textbf{\textsf{simulate}}}}\;
	\caption{\textbf{\textsf{simulate}} (Alice's side)}
	\label{algo:simulate}
\end{algorithm}

\newpage
\subsection{Analysis}\label{subsec:polysizeclassicalanalysis}

In order to show the correctness of the above algorithm, we condition on some view of the metadata and Pauli data, i.e., $\FullMA$, $\MA$, $\MAtilde$, 
$\FullMB$, $\MB$, $\MBtilde$, 
$\FullPA$, $\PA$, $\PAtilde$, 
$\FullPB$, $\PB$ and $\PBtilde$. 
We define a potential function $\Phi$ as
\begin{align*}
\Phi\defeq \Phi_{\mathrm{Q}}+\Phi_{\MD}+\Phi_{\PD}\enspace,
\end{align*}
where $\Phi_{\MD}$ and $\Phi_{\PD}$ measure the correctness of the two parties' current estimate of each other's metadata and Pauli data, respectively, and $\Phi_{\mathrm{Q}}$ measures the progress in reproducing the joint state of the input protocol. We define
\begin{align}
	&\mdAplus \defeq~\text{the length of the longest prefix where $\MA$ and $\MAtilde$ agree;}\label{eqn:mda+}\\
	&\mdBplus \defeq~\text{the length of the longest prefix where $\MB$ and $\MBtilde$ agree;}\label{eqn:mdb+}\\
	&\mdAminus \defeq \max\{\ellMA,\ellMAtilde\}-\mdAplus;\label{eqn:mda-}\\
	&\mdBminus \defeq \max\{\ellMB,\ellMBtilde\}-\mdBplus;\label{eqn:mdb-}\\
	&\pdAplus \defeq \lfloor\frac{1}{r} \times~\text{the length of the longest prefix where $\PA$ and $\PAtilde$ agree}\rfloor;\label{eqn:pda+}\\
	&\pdBplus \defeq \lfloor\frac{1}{r} \times~\text{the length of the longest prefix where $\PB$ and $\PBtilde$ agree}\rfloor;\label{eqn:pdb+}\\
	&\pdAminus \defeq \frac{1}{r} \max\{\ellPA,\ellPAtilde\}-\pdAplus;\label{eqn:pda-}\\
	&\pdBminus \defeq \frac{1}{r} \max\{\ellPB,\ellPBtilde\}-\pdBplus.\label{eqn:pdb-}
\end{align}

Also, recall that

\begin{align}
&g \defeq
\text{the number of good unitary blocks in $\JStwo$,}\label{eqn:g-analysis}\\
&b \defeq
\text{the number of bad unitary blocks in $\JStwo$, and}\label{eqn:b-analysis}\\
&u\defeq |q_{\MA}-q_{\MB}|,\label{eqn:u-analysis}
\end{align}

with $q_{\MA}$ and $q_{\MB}$ the number of non-$\sC$ iterations for Alice and Bob, respectively.

\suppress{
Note that by the following lemma, the $\JStwo$ representation defined above is well-defined.

\begin{lemma}\label{lem:invariance}
	In the end of any iteration, the joint state is of the following form.
	\[\frac{1}{2^{\abs{(q_{\MA}-q_{\MB})\cdot r}}} \sum_{l\in\set{0,\ldots,d-1}^{\abs{(q_{\MA}-q_{\MB})\cdot r}}}\ketbra{l}\otimes\ketbra{\psi_l},\]
	where $\ket{\psi_l}$ is of the form JS1, and $d=|\Sigma|$.
	
	Moreover, The first $\min\br{q_{\MA},q_{\MB}}$ blocks (from right to left) of $\ket{\psi_l}$ are same for all $l$. If $q_{\MA}>q_{\MB}$ (resp. $q_{\MA}<q_{\MB}$), then in the last $q_{\MA}-q_{\MB}$ (resp. $q_{\MB}-q_{\MA}$) blocks, all the powers of $U_{B,i}$ (resp. $U_{A,i}$) are $0$.
\end{lemma}
\begin{proof}
	The above form is closed under all quantum operations occurred in the algorithms.
\end{proof}
}

Now we are ready to define the components of the potential function. At the end of the $i$-th iteration, we let
\begin{align}
	&\Phi_{\mathrm{Q}}\defeq g-b-5u, \label{eqn:phiQ}\\
	&\Phi_{\MD}\defeq \mdAplus-3\mdAminus+\mdBplus-3\mdBminus-2i,\label{eqn:phimd}\\
	&\Phi_{\PD}\defeq  \pdAplus-\pdAminus+\pdBplus-\pdBminus-3q_{\MA}-3q_{\MB},\label{eqn:phipd}\\
	&\Phi\defeq\Phi_{\mathrm{Q}}+\Phi_{\MD}+\Phi_{\PD}.\label{eqn:phi}
\end{align}
where $g$, $b$ and and $u$ are defined in Eqs.~\eqref{eqn:g-analysis}, \eqref{eqn:b-analysis}, and \eqref{eqn:u-analysis}.

\begin{lemma}\label{lem:phimdpdnegativelargeclassical}
	Throughout the algorithm, it holds that
	\begin{itemize}
		\item $\Phi_{\MD} \leq 0$ with equality if and only if Alice and Bob have full knowledge of each other's metadata, i.e.,   $\mdAplus = \mdBplus = i$ and $\mdAminus = \mdBminus = 0$.
		\item $\Phi_{\PD} \leq 0$ with equality if and only if Alice and Bob have full knowledge of each other's Pauli data, i.e., $\pdAplus  = 3q_{\MA}$, $ \pdBplus = 3q_{\MB}$ and $\pdAminus = \pdBminus = 0$.
	\end{itemize}
\end{lemma}
\begin{proof}
	The first statement follows from the property that $\mdAplus,\mdBplus \leq i$, and the second statement holds since $\pdAplus \leq 3q_{\MA}$ and $\pdBplus \leq 3q_{\MB}$.
\end{proof}

Note that if $g-b-u\geq n/{2r}$, the noiseless protocol embedding described in Section \ref{sec:nslss}, guarantees that not only is the correct final state of the original protocol produced and swapped into the safe registers $\tilde{A}$, $\tilde{B}$ and $\tilde{C}$, but also they remain untouched by the bad and ugly blocks of the simulation. Therefore, by Lemma \ref{lem:phimdpdnegativelargeclassical}, for successful simulation of an $n$-round protocol it suffices to have $\Phi \geq n/{2r}$, at the end of the simulation.

The main result of this section is the following:

\setcounter{theorem}{0}
\begin{theorem}[\textbf{Restated}]\label{theorem:simplealg}
	Consider any $n$-round alternating communication protocol $\Pi$ in the teleportation-based model, communicating messages over a noiseless channel with an alphabet $\Sigma$ of bit-size $\Theta\br{\log n}$. Algorithm \ref{algo:Mainalgorithm} is a computationally efficient coding scheme which given $\Pi$, simulates it with probability at least $1-2^{-\Theta\br{n\epsilon}}$, over any fully adversarial error channel with alphabet $\Sigma$ and error rate $\epsilon$. The simulation uses $n\br{1+\Theta\br{\sqrt{\epsilon}}}$ rounds of communication, and therefore achieves a communication rate of $1-\Theta\br{\sqrt{\epsilon}}$. Furthermore. the computational complexity of the coding operations is~$O\br{n^2}$.
\end{theorem}

{\bf Proof Outline.} We prove that any iteration without an error or hash collision increases the potential by at least one while any iteration with error or hash collision reduces the potential by at most some fixed constant. As in Ref.~\cite{Haeupler:2014}, with very high probability the number of hash collisions is at most $O(n\epsilon)$, the same order of magnitude as the number of errors, therefore negligible. Finally, our choice of the total number of iterations, $R_{total} \defeq \lceil n/2r + \kappa n \epsilon \rceil$ (for a sufficiently large constant $\kappa$), guarantees an overall potential increase of at least~$n/2r$. As explained above, this suffices to prove successful simulation of the input protocol.

\setcounter{theorem}{4}
\begin{lemma}\label{lem:potential increase}
	Each iteration of the Main Algorithm (Algorithm~\ref{algo:Mainalgorithm}) without a hash collision or error increases the potential $\Phi$ by at least $1$.
\end{lemma}

\begin{proof} Note that in an iteration with no error or hash collision, Alice and Bob agree on the iteration type. Moreover, if $\Itertype=\;\MD$ or $\PD$ (Case i or iii), they also agree on whether they extend or rewind the data (the subcase A or B), and if $\Itertype\;=\;\MES$ (Case ii), then exactly one of them is in \textsf{Case A} and the other one is in \textsf{Case B}. We analyze the potential function in each of the cases, keeping in mind that we only encounter Case~ii or later cases once the metadata of the two parties are consistent and of full length, and similarly, that we encounter Case~iv once the parties have used the same number of MESs and the Pauli data with the two parties are consistent and of full length. Lemma~\ref{lem:phimdpdnegativelargeclassical} guarantees that~$\Phi_{\MD}$ becomes~0 on entering Case~ii, and that~$\Phi_{\MD} = \Phi_{\PD} = 0$ on entering Case~iv.
	
	\begin{itemize}
		\item Alice and Bob are in \textsf{Case i.A}:
		\begin{itemize}
			\item $\Phi_{\PD}$ and $\Phi_{\mathrm{Q}}$ stay the same.
			\item $i$ increases by $1$.
			\item $\mdAplus$ and $\mdBplus$ stay the same.
			\item None of $\mdAminus$ and $\mdBminus$ increases, and at least one decreases by $1$.
		\end{itemize}
		Therefore, $\Phi_{\MD}$ increases at least by $3-2=1$, and so does $\Phi$.
		
		\item Alice and Bob are in \textsf{Case i.B}:
		\begin{itemize}
			\item $\Phi_{\PD}$ and $\Phi_{\mathrm{Q}}$ stay the same.
			\item $i$ increases by $1$.
			\item $\mdAminus$ and $\mdBminus$ stay at $0$.
			\item  At least one of $\ellMA$ or $\ell_{\MB}$ is smaller than $i-1$; If only $\ellMA < i-1$, then $\mdAplus$ increases by $2$, and $\mdBplus$ by $1$. The case where only $\ell_{\MB} < i-1$ is similar. If both are smaller than $i-1$, then $\mdAplus$ and $\mdBplus$ both increase by $2$.
		\end{itemize}
		Therefore, $\Phi_{\MD}$ increases by at least $3-2=1$, and so does $\Phi$.
		
		\item Alice is in \textsf{Case ii.A}, Bob is in \textsf{Case ii.B}:
		\begin{itemize}
			\item $\Phi_{\MD}$ stays at $0$.
			\item $q_{\MB}$ increases by $1$.
			\item $q_{\MA}$, $\pdAplus$, $\pdAminus$, $\pdBplus$, $\pdBminus$ all stay the same.
			\item $g$ remains the same, $b$ increases by at most $1$, and $u$ decreases by 1.
		\end{itemize}
		Therefore,  $\Phi_{\mathrm{Q}}$ increases by at least $5-1=4$, and $\Phi_{\PD}$ decreases by $3$. So $\Phi$ increases by at least $1$.
		
		\item Alice is in \textsf{Case ii.B}, Bob is in \textsf{Case ii.A}: This case is similar to the above one.

		\item Alice and Bob are in \textsf{Case iii.A}
		\begin{itemize}
			\item $\Phi_{\MD}$ stays at $0$, and $\Phi_{\mathrm{Q}}$ stays the same
			\item $\pdAplus$, $\pdBplus$, $q_{\MA}$ and $q_{\MB}$  stay the same.
			\item  None of $\pdAminus$ and $\pdBminus$ increases, and at least one decreases by $1$.
		\end{itemize}
		Therefore,  $\Phi_{\PD}$ increases by at least $1$, and so does $\Phi$.
		
		\item Alice and Bob are in \textsf{Case iii.B}
		\begin{itemize}
			\item $\Phi_{\MD}$ stays at $0$, and $\Phi_{\mathrm{Q}}$ stays the same.
			\item $\pdAminus$, $\pdBminus$ stay at $0$, and $q_{\MA}$, $q_{\MB}$ stay the same.
			\item  At least one of the following holds:
$\ellPA < 3q_{\MA}\cdot r$, in which case $\pdAplus$ increases
by $1$ (otherwise it remains unchanged), or $\ellPB < 3q_{\MB}\cdot r$, and then $\pdBplus$ increases by $1$ (otherwise it remains unchanged).
		\end{itemize}
		Therefore,  $\Phi_{\PD}$ increases by at least $1$, and so does $\Phi$.
		
		\item Alice and Bob are in \textsf{Case iv}
		\begin{itemize}
			\item $\Phi_{\MD}$ and $\Phi_{\PD}$ stay at $0$.
			\item $u$ stays at $0$
			\item Either $g$ stays the same and $b$ decreases by 1 (when $b\neq0$) or $b$ stays at 0 and $g$ increases by 1.
		\end{itemize}
		Therefore,  $\Phi_{\mathrm{Q}}$ increases by $1$, and so does $\Phi$.

	\end{itemize}
	
	Hence $\Phi$ increases at least by $1$ for each iteration of the algorithm without a hash collision or error.
	
\end{proof}

\begin{lemma}\label{lem:simpleerrorpotential}
	Each iterations of Algorithm \ref{algo:Mainalgorithm}, regardless of the number of hash collisions and errors, decreases the potential $\Phi$ by at most $45$.
\end{lemma}
\begin{proof}
	At each step, $i$ increases by $1$ while, in the worst case, $g$, $\mdAplus$,$\mdBplus$, $\pdAplus$ and $\pdBplus$ decrease by at most $1$,  $b$, $u$, $q_{\MA}$ and $q_{\MB}$ increase by at most $1$, $\mdAminus$ and $\mdBminus$ increase by at most $3$ and $\pdAminus$ and $\pdBminus$ increase by at most $4$.
Hence, $\Phi_{\mathrm{Q}}$, $\Phi_{\MD}$ and $\Phi_{\PD}$ decrease at most by $7$, $22$, and $16$, respectively. So in total, $\Phi$ decreases by at most $45$.
\end{proof}
The following lemma is from \cite{Haeupler:2014}.
\begin{lemma}\label{lem:simplehashcollisions}
	The number of iterations of Algorithm \ref{algo:Mainalgorithm} suffering from a hash collision is at most $6n\eps$ with probability at least $1 - 2^{-\Theta(\eps n)}$.
\end{lemma}

\begin{proofof} {Theorem \ref{theorem:simplealg}}
Let $\Rtotal=\lceil\frac{n}{2r}\rceil+368n\epsilon$. The total number of iterations is less than $2n$, so the total number of iterations with an error is at most $2n\epsilon$. By Lemma \ref{lem:simplehashcollisions}, with probability at least $1 - 2^{-\Theta(\eps n)}$, the number of iterations with a hash collision is at most $6n\epsilon$. Therefore, by Lemma \ref{lem:potential increase}, in the remaining $\Rtotal-8n\epsilon=\lceil\frac{n}{2r}\rceil+360n\epsilon$ iterations, the potential $\Phi$ increases by at least one. The potential decreases only when there is an error or hash collision and it decreases by at most $45$. So at the end of the simulation, we have 
\[g-b-u\geq\Phi_{\mathrm{Q}}\geq\Phi\geq \Rtotal-8n\epsilon-45\times 8n\epsilon\geq\frac{n}{2r}.\]
Hence the simulation is successful. Furthermore, note that the of amount communication in each iteration is independent of the iteration type and is always $2r+\Theta(1)$ symbols: in every iteration each party sends $\Theta(1)$ symbols to communicate the hash values and the lengths of  the metadata and Pauli data in line 9 of Algorithm \ref{algo:Mainalgorithm}; each party sends another $r$ symbols, either in line 13 of Algorithm \ref{algo:Mainalgorithm}, if $\Itertype \neq \SIM$ or in Algorithm \ref{algo:simulate} to communicate the teleportation measurement outcomes. 
So the total number of communicated symbols is
\begin{equation}
\Rtotal\cdot(2r+\Theta(1))=\left(\lceil\frac{n}{2r}\rceil+\Theta(n\epsilon)\right)\left(2r+\Theta(1)\right)=n(1+\Theta(\sqrt{\epsilon})),
\end{equation}
as claimed.
\end{proofof}

\newpage
\section{Recycling-based protocol via quantum channel with large alphabet}

\label{sec:BriefQuLarge}

\subsection{Overview}

\subsubsection{Teleportation is inapplicable}

Switching from the teleportation-based
model to the plain quantum model, suppose we are given a protocol $\Pi$ using
noiseless quantum communication, and we are asked to provide a protocol $\Pi'$
using noisy quantum channels under the strongly adversarial model
described earlier.  In the absence of free entanglement, how can we
protect quantum data from leaking to the environment without
incurring a non-negligible overhead?  First, note that some form of
protection is necessary, as discussed in
\longpaper{Section~\ref{sec:intro-diff}.}
\blurb{Section~2 of this extended abstract.}
  Second,
teleportation would be too expensive to use, since it incurs an overhead of
at least $3$: we have to pay for the MES as well as the
classical communication required.

Surprisingly, an old and relatively unknown idea called
the Quantum Vernam Cipher (QVC)~\cite{Leung:2002} turns out to
be a perfect alternative method to protect quantum data with
negligible overhead as the noise rate approaches 0.

\subsubsection{Quantum Vernam Cipher (QVC)}

Suppose
Alice and Bob share two copies of MESs, each over two $d$-dimensional
systems.  For Alice to send a message to Bob, she applies a controlled-$\rX$ operation with her half of the first MES as control, and
the message as the target.  She applies a controlled-$\rZ$
operation from her half of the second MES to the message.  When Bob
receives the message, he reverses the controlled operations using his
halves of the MESs.  The operations are similar for the opposite direction of
communication. A detailed description is provided in Section~\ref{sec:qvc}.
\suppress{the full paper.
\blurb{Section~2.5 of the full manuscript; Figure~7 at the end of this extended abstract depicts the protocol.}}

QVC is designed so that given access to an
authenticated classical channel from Alice to Bob, Bob can determine and correct any error in the transmission of the quantum message.  This can simply be
done by measuring $\rZ^l$ type changes to one half of the two MES.
They can also run QVC many times to send multiple messages and determine the errors in a large
block using a method called ``random hashing'', and recycle the MESs
if the error rate (as defined in our adversarial model) is low.  This is
a crucial property of QVC and leads to one of the earliest (quantum) key
recycling results known.
What makes QVC particularly suitable for our problem is that encoding and decoding
are performed message-wise, while error detection can be done in large blocks, and
entanglement can be recycled if no error is detected. It may thus be viewed as
a natural quantum generalization to Haeupler's consistency
checks.
%

As an aside, in Appendix E of Ref.~\cite{Leung:2002}, the relative merits of
teleportation and QVC were compared (those are the only two ciphers
with complete quantum reliability), and it was determined that
entanglement generation over an insecure noisy quantum channel
followed by teleportation is more entanglement efficient than QVC
with entanglement recycling in some test settings.  However, this
difference vanishes for low noise.  Furthermore, the comparison
assumes authenticated noiseless classical communication to be free.
QVC requires an amount of classical communication for the
consistency checks which vanishes with the noise parameter (but this
cost was not a concern in that study).  Furthermore, QVC was also
proposed as an authentication scheme, but the requirement for
interaction to authenticate and to recycle the key or entanglement was
considered a disadvantage, compared to non-interactive schemes.
(Those are only efficient for large block length, and cannot identify
the error when one is detected. So, these authentication schemes are
inapplicable).  We thus provide renewed insight into QVC when
interaction is natural (while it is considered expensive in many other
settings).


\subsubsection{Entanglement recycling and adaptations of QVC  for the current problem}

In the current scenario, we have neither free MESs nor an authenticated
classical channel.
Instead, Alice and Bob start the protocol by distributing the MESs they need, using a high rate quantum error correcting code over the low-noise channel.
Then, they run the input protocol $\Pi$ as is over the noisy channel, while frequently checking for errors by performing \emph{quantum hashing\/}~\cite{BDSW96,Leung:2002}, using the same noisy quantum channel instead
of an authenticated classical channel. If they detect an inconsistency, assuming that the errors are most likely recent, they measure a small block of MESs in the recent past to determine the errors. They continue this process until they get matching quantum hash values indicating (with constant probability) that they have located and identified all the errors and the remaining MESs can be recycled and reused to encrypt the messages. Frequent quantum hashing allows Alice and Bob to boost their confidence about recyclability of the earlier MESs and reuse MESs in a cyclic way. Note that for successful simulation it is crucial to ensure that the recycled MESs are indeed not corrupted and that Alice and Bob recycle the same sequence of MESs. One of our main contributions in this section is developing a framework for recycling entanglement in a communication efficient way. 
We show that entanglement generation of $O\br{n\sqrt{\epsilon}}$ MESs, where $n$ is the length of the input protocol $\Pi$ and $\epsilon$ is the noise parameter, is sufficient to last through the whole simulation.

\subsubsection{Framework}

As in the case of the teleportation-based protocols, due to transmission errors and collisions, Alice and Bob do not necessarily always agree on their actions in the simulation. Therefore, in every iteration both parties need to obtain a global view of the history of the simulation so far to correctly decide their next actions. They achieve this goal by maintaining a similar data structure as in the teleportation-based case. The data structure now contains additional information to keep track of their measurements, which Alice and Bob use in the recycling process. 

\subsubsection{Additional out-of-sync problems}

Due to transmission errors introduced by the adversary, Alice and Bob may get out of sync in QVC. In such a scenario, the QVC operations are performed by only one party and the quantum data intended to be sent to the other party leaks into the MES registers used to encrypt the messages. Furthermore, the parties may not agree on the subset of MESs they have already measured when they perform quantum hashing. As we will explain in Section~\ref{sec:out-of-sync QH}, in the worst case, this can further lead to the leakage of the quantum data into all the MES registers involved in the quantum hashing procedure.   

\suppress{
In particular, corruptions can
lead only one party to recycle an MES can cause a significant discrepancy in
how many MESs the two parties are holding.  It is much more involved
to analyse the joint quantum state.}
%
\suppress{To tackle these problems, we develop further data structures and adapt
the ``quantum hashing'' procedure of Ref.~\cite{BDSW96,Leung:2002} to
our setting.}
%

We show that, surprisingly, once again the quantum data can be recovered once Alice and Bob reconcile the differences in the data structure developed for the task. This is in spite of the fact that there is no reason to expect out-of-sync QVC to be sufficient to protect the quantum data from leaking to the environment when encoding and decoding operations are performed incorrectly and quantum data is sent via the noisy quantum channel.





\subsection{Result}

Our main result in the plain quantum model with polynomial-size communication alphabet is the simulation of an $n$-round noiseless communication protocol over a fully adversarial channel of error-rate $\epsilon$ defined in Section~\ref{sec:noisy_comm_model}. \suppress{First, we state the result in the large alphabet case.}

\begin{theorem} \label{thm:Qmessagelargealphabet}
	Consider any $n$-round alternating communication protocol $\Pi$ in the plain quantum model, communicating messages over a noiseless channel with an alphabet $\Sigma$ of bit-size $\Theta\br{\log n}$. Algorithm \ref{algo:MainalgorithmQMessage} is a quantum coding scheme which given $\Pi$, simulates it with probability at least $1-2^{-\Theta\br{n\epsilon}}$, over any fully adversarial error channel with alphabet $\Sigma$ and error rate $\epsilon$. The simulation uses $n\br{1+\Theta\br{\sqrt{\epsilon}}}$ rounds of communication, and therefore achieves a communication rate of $1-\Theta\br{\sqrt{\epsilon}}$.
\end{theorem}

\subsection{Description of Protocol}\label{sec:description-large-quantum}
\subsubsection{General Description}\label{subsec:description-large-quantum}

Our simulation of noiseless protocols in the plain quantum model of communication proceeds using the same idea of running $O_\epsilon\br{1}$ rounds of the input protocol as is, while checking if the adversary has corrupted the communication during the previous iterations and if necessary, actively rewinding the simulation to correct errors. The quantum messages are protected using  QVC against corruptions by the adversary. In order to detect potential transmission errors, the MES pairs used as the key in QVC may be measured after each communication round. The measurement outcomes may be stored and later on compared to obtain the error syndrome. Therefore, using a data structure similar to the one introduced in the previous section, one can obtain a coding scheme for simulating any protocol in the plain quantum model.  However, this approach is not efficient in using the entanglement. Recall that in the plain quantum model, the parties do not pre-share any entanglement, hence they need to establish the shared MESs through extra communication. Rather than measuring the MES pairs immediately after each round of communication, we use the quantum hashing procedure described in Section ~\ref{sec:Qhashing} to check whether any transmission error has occurred so far. Note that if Alice and Bob detect an error, they need to determine the error and eventually actively rewind the simulation to correct it and resume the simulation from there. However, similar to the teleportation-based protocol, due to transmission errors Alice and Bob may not always agree on how they proceed with the simulation in every iteration. Thus, in every iteration before taking further actions, each party needs to know the actions of the other party so far. More accurately, they first need to obtain a global view of their joint quantum state. Alice and Bob locally maintain a similar data structure as in the teleportation-based protocol containing metadata and Pauli data, which needs to be synchronized in the algorithm. They first need to ensure they have full knowledge of each other's metadata. Then similar to the teleportation-based case, in order to avoid out-of-sync scenarios in communication using QVC, it is crucial for them to synchronize the number of MESs they have used (see Section \ref{subsec:out-of-sync QVC}). We denote by $\ellQVCA$ and $\ellQVCB$, the number of blocks of MES pairs used by Alice and Bob, respectively. After ensuring $\ellQVCA=\ellQVCB$ (to the best of their knowledge), they compare their quantum hash values to check for errors. Note that quantum hashing does not indicate in which round of communication the error has occurred. If the hash values do not match, they measure the last block of MES pairs which has not gained as much trust as the older blocks through quantum hashing. This effectively collapses the adversary's action on this block to Pauli errors. The effective errors can be determined jointly from the measurement outcomes, which are recorded by Alice and Bob as part of their local Pauli data. Otherwise, if the hashes do match, then Alice and Bob synchronize their Pauli data. Note that similar to the teleportation-based protocol, the Pauli corrections performed by Alice and Bob are also recorded as part of the Pauli data. Together with the metadata, this gives Alice and Bob all the information they need to compute their estimate of the current joint state. Hence, they can determine how to proceed with the simulation next.

\paragraph{Recycling entanglement and recycling data.}

An important complication arising in simulation of protocols in the plain quantum model is that in order to achieve the simulation rate of $1-\Theta(\sqrt{\epsilon})$, we cannot afford to access a new MES pair in every round of communication using QVC. This is where we use a crucial property of QVC, namely the key recycling property. Note that in communication using QVC if no error occurs on the message then the pair of MESs used will remain intact, hence they can be recycled and used in later rounds. Otherwise, at some point Alice and Bob need to measure the MES pair to get the error syndrome, in which case the pair cannot be recycled. By performing quantum hashing regularly and carefully keeping track of the measured MES blocks, Alice and Bob can recycle MES pairs as needed and run the simulation by establishing a smaller number of MESs at the beginning of the protocol. 

In order to correctly simulate the input protocol we need to ensure that the recycling is successful in every iteration, namely that the same sequences of MES blocks are recycled by the two parties and that they are indeed not corrupted when being recycled. Note that if the two parties reuse an MES pair which has been corrupted due to an earlier transmission error in QVC, then even if Alice and Bob detect the error and measure the MES pair, they have no way of knowing whether the error has occurred the last time the MES pair were used or in an earlier round. Moreover, if a block of MES pairs has been locally measured by only one party, say Alice, then the other party, Bob, needs to measure the block and avoid this block when encrypting future messages using QVC.

In order to achieve successful recycling, we modify the metadata so that it contains additional information to keep track of each party's measurements. Alice maintains a string $\RA\in \{\sS,\sM\}^*$, where each $\sM$ symbol corresponds to a measured MES block and $\sS$ is used for an MES block still in superposition. In each iteration, Alice finds the next reusable MES block, records the index of the block in a string, $\IndexA$, of recycled MESs and concatenates $\RA$ with a new $\sS$ symbol. The string $\IndexA$ serves as Alice's queue of reusable MES blocks. Moreover, if she measures an MES block in the current iteration, she changes the corresponding $\sS$ symbol in $\RA$ to $\sM$. Similarly, Bob maintains the strings $\RB$ and $\IndexB$. Note that the recycled MES blocks are not necessarily reused immediately or even in the same iteration by the two parties. 

Alice and Bob need to ensure that the strings $\RA$ and $\RB$ are the same. Using her full-length estimate $\MBtilde$ of Bob's metadata $\FullMB$, Alice computes an estimate $\RBtilde$ of $\RB$. Note that if $\MBtilde=\FullMB$ then $\RBtilde=\RB$. Similarly, Bob computes his estimate $\RAtilde$ of $\RA$ from his full-length estimate $\MAtilde$ of Alice's metadata $\FullMA$. After synchronizing their metadata and ensuring that they have used the same number of MES blocks, they synchronize their recycling data. Alice recycles MES blocks using $\RA$ and $\IndexA$ and Bob using $\RB$ and $\IndexB$. Frequent hashing of the metadata allows them to be highly confident that their recycling data agree in a sufficiently long prefix. Furthermore, quantum hashing ensures that with high probability all the recycled MES blocks are indeed reusable. 

Note that the synchronization of the recycling data is slightly different from the synchronization of the metadata and the Pauli data. For the latter, Alice and Bob just need to know each other's local data, i.e., Alice needs to know $\FullMB$ and $\FullPB$ and Bob needs to know $\FullMA$ and $\FullPA$ and the corresponding data need not be the same. Whereas for the recycling data, Alice and Bob need to learn each other's data and match them, i.e., they need to ensure $\RA=\RB$.

\paragraph{Entanglement distribution.}

At the outset of the simulation, Alice and Bob use the Robust Entanglement Distribution protocol (Algorithm~\ref{algo:Robust Entanglement Distribution}) introduced in Section~\ref{subsec:hash} to share $\Theta(n\sqrt{\epsilon})$ copies of the MES $\ket{\phi^{0,0}}$, defined in Definition ~\ref{def:Bellstates}. The shared MESs are used as follows: 
\begin{itemize}
	\item $\Theta\br{n\sqrt{\epsilon}}$ MESs are used in pairs to serve as the key for encryption of messages using QVC. They are divided into $\LQVC=\Theta\br{n\epsilon}$ blocks of $2r$ MES pairs, where $r=\Theta(\frac{1}{\sqrt{\epsilon}})$. In each block the MES pairs are implicitly numbered from $1$ to $2r$. The odd-numbered pairs are used in QVC to send messages from Alice to Bob and the even-numbered pairs are used to send messages from Bob to Alice.
	\item $\Theta(n\sqrt{\epsilon})$ MESs are reserved to be used in quantum hashing. 
	\item The remaining $\Theta(n\sqrt{\epsilon})$ MESs are measured in the computational basis by both parties to obtain a common random string to be used as the seed for classical and quantum hashing.
\end{itemize} 
We show that with the limited error budget of the adversary, the $\Theta(n\sqrt{\epsilon})$ MESs established at the beginning of simulation are sufficient to successfully simulate the input protocol. This allows us to achieve a simulation rate of $1-\Theta(\sqrt{\epsilon})$.

\begin{figure}[!t]
\centering
\includegraphics[width=480pt]{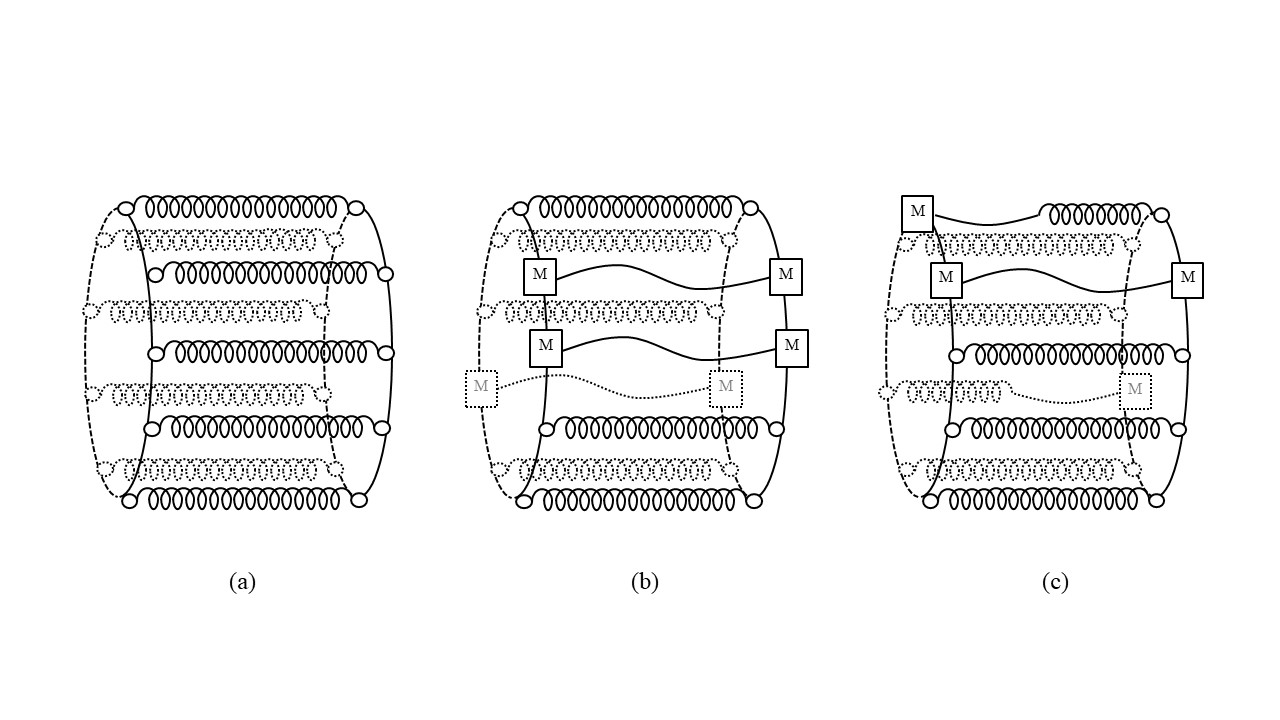}
\caption{These figures represent the blocks of MES pairs at different stages of the protocol. To simplify the figure, we represent each block by a single MES. Note that these are used in a circular pattern, corresponding to recycling some of the previously used blocks of MES pairs. Those depicted as circles are assumed to be good and usable for QVC, those depicted by squares have been measured already in order to extract the error syndrome. Figure (a) represent the MES blocks at the beginning of the protocol, when none have been measured. Figure (b) represents them when Alice and Bob agree on which ones have been measured and have used the same amount of them for QVC, which is the desired state. Figure (c) represents a situation when Alice and Bob have gotten out-of-sync, e.g., Alice has measured some blocks that Bob has not (and maybe used QVC more often than Bob). They then work to get back in sync before resuming the simulation.
}
\label{fig:EPR_circle}
\end{figure}

\subsubsection{Quantum Hashing}\label{sec:Qhashing}

By performing local measurements on the MES pairs used as the keys for encrypting the messages in QVC and comparing the measurement outcomes on both sides, one can extract the error syndrome corresponding to the corruptions introduced over the noisy communication channel.  As explained in the previous subsection, although this allows the two parties to detect errors immediately, it is not efficient for our application. In Subsection~\ref{sec:qvc}, we introduced an error detection procedure which allows the parties to check for corruptions when QVC is used over several rounds of communication to send multiple messages at the cost of losing only one MES which is measured at the end. However, this error detection procedure is not directly useful in our application since the adversary can always choose the corruptions in a way that makes it impossible for Alice and Bob to detect errors; see subsection~\ref{sec:qvc}. Instead, Alice and Bob use the \emph{quantum hashing\/} procedure described below to check whether there is an undetected error. 
To avoid the adversary from hiding her corruptions from the detection procedure above, Alice and Bob choose a random subset of the MESs and try to detect errors in this subset rather than all MESs used in QVC. More precisely, quantum hashing involves the following steps in our algorithm. At the beginning of the $i$-th iteration, Alice and Bob pick a fresh MES serving as the control system used in the error detection procedure. Recall that Alice and Bob locally maintain the strings $\IndexA$ and $\IndexB$, respectively, corresponding to their recycled MES blocks. Alice uses the shared randomness established at the outset of the protocol to choose a random subset of the MES registers contained in the blocks specified by $\IndexA\Br{i-t+1:\ellQVCA}$, for $t\in \Theta \br{n\epsilon}$. Using her recycling data $\RA$, she locally determines the MESs in this random subset which she has not measured already. She performs her operations in the detection procedure described in Subsection~\ref{sec:qvc} only on these MESs. Bob does the same locally, based on $\IndexB$ and $\RB$. Alice and Bob store their measurement outcomes in $\QHA$ and $\QHB$, respectively, and exchange the values. We prove that except with exponentially small probability, recycling is successful throughout the execution of the algorithm.  Therefore we can assume $\IndexA=\IndexB$ in every iteration. But $\RA$ and $\RB$ do not necessarily match in every iteration. Therefore, in some iterations, Alice and Bob may perform the quantum hashing on different subsets of the MESs. Moreover, if the two parties are not synchronized in the number of MES blocks they have used , they might perform quantum hashing on an MES that is only used by one party but not the other. We discuss these out-of-sync quantum hashing scenarios in Subsection~\ref{sec:out-of-sync QH}.

Alice and Bob compare their quantum hash values only when they believe they have full knowledge of each other's metadata, have used the same number of MES blocks and agree on their recycling data. If they get the same hash values, they assume no error has occurred. Otherwise, they believe they have detected an error. Note that similar to all other types of data that are communicated over the noisy channel, the quantum hash values may get corrupted by the adversary. We say a \emph{quantum hash collision\/} has occurred in an iteration only when recycling has been successful so far, Alice and Bob are indeed synchronized in their metadata, the number of MES blocks they have used and their recycling data and $\QHA=\QHB$ despite the fact that there are non-measured MES blocks in the last $t$ blocks of $\IndexA=\IndexB$, which are not in the $\ket{\phi^{0,0}}^{\otimes 4r}$ state.

The following lemma shows that in every iteration of the algorithm, assuming successful recycling up to that point, the probability of a quantum hash collision is at most $1/2$. 
\suppress{
Note that by Equations ~\eqref{eq:messagecorrupt},\eqref{eqn:quantumvernamcipher} and \eqref{eqn:reverseorderquantumvernamcipher}, even if Alice and Bob reuse a block of MES pairs corresponding to an iteration in which communication is corrupted, all the MESs always stay in $\textsf{span}\set{\ket{\phi^{0,k}}:0\leq k\leq d-1}$.
}

\begin{lemma}\label{lem:quantumhash}
	Let $m\in \mathbb{N}$ and $k=k_1 \ldots k_m \in \{0,1,\ldots,d-1\}^m$. Suppose that Alice and Bob share the states $\ket{\phi^{0,0}}_{A_0B_0},\ket{\phi^{0,k_1}}_{A_1B_1},\ldots,\ket{\phi^{0,k_m}}_{A_mB_m}$and a random variable $S$ distributed over $\set{0,1}^m$, interpreted as a subset of $[m]$. Alice and Bob apply $\control{\rX}_{A_0A_i}$ and $\control{\rX}_{B_0B_i}$, for all $i\in S$. Then they apply the quantum Fourier transform operator $\rF$ and its inverse $\rF^{\dagger}$ on $A_0$ and $B_0$, respectively. They measure the registers in the computational basis with outcomes $\QHA$ and $\QHB$, respectively. Then for $k=0^m$, independent of the random variable $S$, we have
	\[
	\prob{\QHA=\QHB}=1\enspace.
	\] 
	Moreover, for uniformly random $S$, for all $k\neq0^m$, we have 
	\[
	\prob{\QHA=\QHB}\leq \frac{1}{2}\enspace.
	\]
\suppress{
	\[
	\prob{\QHA=\QHB}\begin{cases}
	=1 & \text{if $k_1=\ldots=k_m=0$}\\
	\leq \frac{1}{2} &\text{otherwise}\enspace,
	\end{cases}
	\]
	and if $S$ is $\delta$-biased, then
	\[
	\prob{\QHA=\QHB}
	\begin{cases}
	=1 & \text{if $k_1=\ldots=k_m=0$} \\
	\leq \frac{1}{2}+\delta & \text{otherwise}.
	\end{cases}
	\]
	Furthermore, the states in $A_1B_1,\ldots, A_mB_m$ remain unchanged.
}
\end{lemma}

\begin{proof}
	Let $S=S_1S_2\ldots S_m$. By Lemma~\ref{lem:cnotbell}, the state in register $A_0B_0$ before the measurements is
	\begin{align*}
	\br{\rF\otimes \rF^\dagger} \ket{\phi^{0,-\sum_{i=1}^mS_ik_i}}_{A_0B_0}
	&=\br{\rF \rZ^{-\sum_{i=1}^mS_ik_i}\otimes \rF^\dagger}\ket{\phi^{0,0}}
	=\br{\rF \rZ^{-\sum_{i=1}^mS_ik_i}\rF^\dagger\otimes \id}\ket{\phi^{0,0}}\\
	&=\br{\rX^{-\sum_{i=1}^mS_ik_i}\otimes \id}\ket{\phi^{0,0}}
	=\ket{\phi^{-\sum_{i=1}^mS_ik_i, 0}}\enspace,
	\end{align*}
	where the first and the last equality follow from Definition~\ref{def:Bellstates}; the second equality holds by Proposition~\ref{fac:uuepr} and the fact that~$\rF=\rF^T$. The third equality follows from Proposition~\ref{fac:paulioperatorcommute}.
	Hence
	\[\prob{\QHA=\QHB}=\prob{\sum S_ik_i=0\mod d}\enspace.\]	
	If $k=0^m$ then the above probability equals $1$. Suppose that $k$ is non-zero in a non-empty subset $J$ of coordinates in $\Br{m}$. Consider the set $Z$ of all $s\in \{0,1\}^m$ such that $\sum s_ik_i=0 \mod d$. Note that the minimum Hamming distance of elements of $Z$ restricted to $J$ is at least $2$, since otherwise there exists $j\in J$ such that $d$ divides $k_j$, contradicting $k_j\in \Br{d-1}$. Fix $j\in J$ and let $e_j\in \{0,1\}^m$ be the string which is $1$ in the $j$-th coordinate and zero everywhere else. For every $s\in Z$, the string $s+e_j$ is not in $Z$. Therefore, $|Z|\leq 2^{m-1}$ and for $S$ uniformly distributed over $\{0,1\}^m$ we have 
	\[ 
	\prob{\sum S_ik_i=0\mod d}\leq\frac{1}{2}.
	\]
	Note that the above bound is tight when $|J|=1$. Finally, by Lemma~\ref{lem:cnotbell} the state in the registers $A_1B_1,\ldots, A_mB_m$ remains unchanged.
\end{proof}

In order to reduce the collision probability to a smaller constant, quantum hashing may be repeated for a constant number of times in every iteration with fresh control MESs and independent random subsets $S$.

\paragraph{Classical seeds needed for quantum hashing.}

Alice and Bob perform quantum hashing and communicate the hash values in every iteration but they only compare their hash values in a subset of iterations. We choose to do so in order to avoid the two parties from getting out-of-sync on which MES register to use in quantum hashing. As the hashing procedure only consumes a constant number of MESs in each iteration, the total number of MESs used in quantum hashing in the entire simulation is $\Theta\br{\Rtotal}=\Theta\br{n\sqrt{\epsilon}}$, and they constitute a constant fraction of the MESs distributed at the outset of the protocol; see Subsection \ref{subsec:description-large-quantum}. On the other hand, generating independent $\Theta\br{rt}$-bit seeds, with~$r\in \Theta\br{1/\sqrt{\epsilon}}$ and $t\in \Theta\br{n\epsilon}$, for each of the $\Rtotal$ iterations would require $\Theta\br{n^2\epsilon}$ bits of shared randomness. The shared randomness is obtained by measuring a fraction of the MESs established at the beginning of the algorithm. Even in the large-alphabet case, sharing $\Theta\br{n^2\epsilon}$ bits of randomness would require too much communication. 

To circumvent this obstacle Alice and Bob start with a smaller number of i.i.d. random bits and extend them to a much longer pseudo-random string. In more detail, they measure $\Theta\br{n\sqrt{\epsilon}}$ MESs in the computational basis and record the binary representation of the outcomes in $R'$. Then they each use the deterministic algorithm of Lemma \ref{lem:stretch} with $\delta=2^{-\Theta\br{n\sqrt{\epsilon}}}$, to obtain a shared $\delta$-biased string $R'$ of length $\Theta\br{rt\Rtotal }=\Theta\br{n^2\epsilon}$. The following lemma bounds the collision probability when instead of a uniformly random seed, a $\delta$-biased seed is used in quantum hashing. Note that in our application of Lemma~\ref{lem:quantum hash-delta biased}, we have $m=O\br{rt}=O\br{n\sqrt{\epsilon}}$.
	
\begin{lemma} \label{lem:quantum hash-delta biased}
	Suppose that the random variable $S$ in Lemma \ref{lem:quantumhash} is $\delta$-biased. Then for all $k\neq 0^m$, we have
	\[
	\prob{\QHA=\QHB}\leq \frac{1}{2}+2^{m/2}\delta \enspace.
	\]
\end{lemma}	 

\begin{proof}
	Let $U$ denote the uniform distribution on $\{0,1\}^m$ and $Z$ be the subset of all $s\in\{0,1\}^m$ such that $\sum s_ik_i=0 \mod d$. By Propositions~\ref{prop:uniform-vs-delta-biased} and~\ref{prop:l1-distance}, we have
	\[
	|U(Z)-S(Z)|\leq \frac{1}{2}\|U-S\|_1 \leq \frac{1}{2} \times 2^{m/2} \|U-S\|_2 \leq 2^{m/2}\delta\enspace.
	\]
	Therefore, by Lemma~\ref{lem:quantumhash}, we have 
	\[
	\prob{\QHA=\QHB} = \prob{\sum S_ik_i=0\mod d} = S(Z) \leq \frac{1}{2}+2^{m/2}\delta \enspace.
	\]
\end{proof}

\subsubsection{Out-of-Sync Quantum Vernam Cipher}\label{subsec:out-of-sync QVC}

Consider the scenario where Alice, based on her view of the simulation so far, implements a $+1$ block, while Bob believes their classical data are not consistent and therefore implements a $\sC$ iteration. Alice simulates a block of the input protocol $\Pi$ while using the next block of MES pairs in the queue $\IndexA$ to encrypt her messages using QVC. At the same time, Bob tries to reconcile the inconsistency through classical communication and does not send his messages using QVC. In this scenario, they also interpret the messages they receive incorrectly. Alice believes Bob is also encrypting his messages using QVC and she applies QVC decoding operations and her Pauli corrections on Bob's messages. Moreover, she potentially applies unitary operations of the input protocol on her local registers. Meanwhile, Bob treats Alice's messages as classical information about the data he believes they need to synchronize. More importantly, since he does not perform the QVC operations (decoding operations in odd rounds and encoding operations in even rounds) on his side, in each round the corresponding MES pair becomes entangled with the message register. So crucial information for continuing the simulation spreads to multiple registers. Moreover, this scenario could continue for several iterations. Nonetheless, we provide a simple way to redirect the quantum information back to the $ABC$ registers, while effectively reducing this type of error to corruptions introduced in the joint state due to transmission errors by the adversary. Once reduced to such errors, Alice and Bob can actively rewind the incorrect part of the simulation and resume from there.

As explained earlier, the first step for Alice and Bob is to ensure they have full knowledge of each other's metadata. Once they both achieve this goal, they discover the discrepancy in the number of MES blocks they have used. Suppose Bob has used fewer blocks of MES pairs than Alice, i.e., $\ellQVCA>\ellQVCB$ and he discovers this at the beginning of the $i$-th iteration. Let $E_1E_2\ldots E_{4r}$ be the registers with Bob containing halves of the $4r$ MESs in the first MES block that Alice has used, say in the $i'$-th iteration, but Bob has not so far. Note that in iteration $i'$, Alice has used the MES pairs corresponding to $E_1E_2,E_5E_6,\ldots,E_{4r-3}E_{4r-2}$ on her side to encrypt quantum information using QVC and she has performed QVC decoding operations on Bob's messages and her marginal of MES pairs corresponding to $E_3E_4,E_7E_8,\ldots,E_{4r-1}E_{4r}$. In the $i$-th iteration, Alice and Bob both send dummy messages to each other. Let $C_1,C_2,\ldots,C_r$ denote the $r$ message registers sent from Alice to Bob after communication over the noisy channel. For every $j\in[r]$, upon receiving $C_j$, Bob applies QVC decoding operations on $C_j$ and $E_{4j-3}E_{4j-2}$, and then applies QVC encoding operations on $C_j$ and $E_{4j-1}E_{4j}$, i.e., he applies
\[
\br{\control{\rZ}}_{E_{4j}C_j}\br{\control{\rX}}_{E_{4j-1}C_j}\br{\control{\rX}^{-1}}_{E_{4j-3}C_j}\br{\control{\rZ}^{-1}}_{E_{4j-2}C_j}, 
\] 
and then he discards the message register $C_j$. The effect of these operations is the same as if Alice and Bob had both used the MES block in sync, i.e., in the $i'$-th iteration, \emph{except the following also happened independently of channel error\/}:

\begin{enumerate}
	\item Alice's messages in the $i'$-th iteration were replaced by $C_1,\ldots,C_r$ and Bob applied his QVC decoding operations on these dummy messages rather than the messages Alice intended to communicate,
	\item the unitary operations used by Bob on the registers $BC$ were all identity, and 
	\item Bob's messages were replaced by his messages of the $i'$-th iteration and  Alice's QVC decoding operations were applied on these (classical) messages. 
\end{enumerate}

The above procedure redirects the quantum information leaked to the MES registers back to the $ABC$ registers, while introducing errors which act exactly the same as transmission errors introduced by the adversary. As in the case of corruptions by the adversary, once Alice and Bob measure the MES block, the error collapses to a Pauli error which can be determined by comparing the measurement outcomes by Alice and Bob. We choose to perform the measurements at the end of the $i$-th iteration, rather than leaving the algorithm to detect the error through quantum hashing (as in the case of transmission errors).


\subsubsection{Out-of-Sync Quantum Hashing} \label{sec:out-of-sync QH}

Consider the scenario in which Alice and Bob have used the same number of MES blocks for communication using QVC but have measured different subsets of MES blocks. Suppose that when they perform quantum hashing, the random subset of MESs they choose contains an MES in registers $A_1B_1$, which has been measured by only one party, say Alice. Let $V_\mathrm{A}V_\mathrm{B}$ be the registers used as the control registers by Alice and Bob in quantum hashing. Alice and Bob compare their quantum hash values only if they believe they have measured the same subset of MES blocks. Therefore, if they compare their hash values in this iteration, it is due to a transmission error or a metadata hash collision. Note that in this scenario Bob applies a controlled-$\rX$ operation on the partially measured MES, while Alice who has measured her marginal does not. Since $A_1$ is already measured by Alice, after Bob's controlled-$\rX$ operation the registers $A_1B_1$ do not get entangled with $V_\mathrm{A}V_\mathrm{B}$. However, the state in the $V_\mathrm{A}V_\mathrm{B}$ registers will be mapped to $\ket{\phi^{0,a}}$, for a random $a\in\{0,1,\ldots,d-1\}$ corresponding to Alice's measurement outcome. This (quite probably) results in Alice and Bob taking incorrect actions from which they can recover once they realize the inconsistency in their classical data. The algorithm is designed to ensure that the register $B_1$ is measured by Bob and $A_1B_1$ is not reused by the two parties in future iterations. Moreover, Bob's controlled-$\rX$ operation does not change the outcome of his measurement on $B_1$. This ensures that Alice and Bob can correctly learn any potential error on the message register in the corresponding communication round once they learn each other's Pauli data.

A subtler scenario occurs when Alice and Bob perform quantum hashing when they have used different numbers of MES blocks. Suppose that $\ellQVCA>\ellQVCB$, i.e., Alice has used more MES blocks and the random subset of MESs they choose for quantum hashing contains an MES which has been only used on Alice's side. As explained in the previous section, when QVC operations are performed only on one side, the MES pair used as the key becomes entangled with the message register and the quantum information in the message register leaks to these MESs. In the above scenario, once quantum hashing is done the information leaks into the additional MES used in hashing. Since the same MES register is used as the control system when applying the $\control{\rX}$ operations on all the MESs in the random subset, the information may leak even further into those registers as well. Surprisingly, the simple solution we provided to recover from out-of-sync QVC resolves this issue as well. The first measure we need to take is to ensure quantum hashing is performed in a sequential way on the MESs in the random subset, starting from the MES used earliest to the latest one. This ensures that the states in the MES registers which have been used both by Alice and Bob do not get disturbed. However, the remaining MESs in the random subset become entangled with the additional MES used in hashing and potentially each other. We need to ensure that these MES registers are not reused in future iterations. Once the two parties synchronize their metadata and realize that they are out of sync on the number of MESs they have used, Bob completes QVC on his side as described in the previous section and immediately measures his marginal of the MES block. This ensures that he will never reuse this block in future iterations. The algorithm is designed so that by the time Alice needs to decide whether to recycle this MES block or not, she will have measured her marginal of the MES registers. We prove that except with probability $2^{-\Theta\br{n\epsilon}}$ recycling is successful in all iterations and such a block of MES registers is never recycled. 

Despite the fact that quantum hashing is performed before Bob completes the QVC operations on his side, this procedure has the same effect as if Bob had completed QVC \emph{before\/} the quantum hashing was performed. To understand this phenomenon, consider the following simpler scenario. Suppose that Alice and Bob share $3$ copies of the MES $\ket{\phi^{0,0}}$ in registers $A_1B_1$, $A_2B_2$ and $V_\mathrm{A}V_\mathrm{B}$. Alice uses the MES pair in registers $A_1B_1$ and $A_2B_2$ as the key to encrypt a message in register $C$ using QVC and sends the message register to Bob. Suppose that the adversary applies the Pauli error $\rX^a \rZ^b$ on the message register for some $a,b\in\{0,1,\ldots,d-1\}$. Now suppose that before Bob applies his QVC decoding operations, Alice applies $\control{\rX}$ on $V_\mathrm{A}A_1$ with $V_\mathrm{A}$ being the control system. Then their joint state is  
\begin{align*}
\br{\control{\rX^{-1}}}_{B_1C}\br{\control{\rZ^{-1}}}_{B_2C}
\br{\control{\rX}}_{V_\mathrm{A}A_1}
\br{\rX^a\rZ^b}_C 
&\br{\control{\rZ}}_{A_2C}\br{\control{\rX}}_{A_1C} \\
&\ket{\phi^{0,0}}_{V_\mathrm{A}V_\mathrm{B}} \ket{\phi^{0,0}}_{A_1B_1} \ket{\phi^{0,0}}_{A_2B_2} \ket{\psi}_C.
\end{align*}
Note that $\br{\control{\rX}}_{V_\mathrm{A}A_1}$ commutes with Bob's QVC decoding operation $\br{\control{\rZ}}_{A_2C}\br{\control{\rX}}_{A_1C}$.  Therefore, by Equation~\eqref{QVC with error} their joint state is given by 
\[
\br{\control{\rX}}_{V_\mathrm{A}A_1} \br{\rX^a \rZ^b}_C 
\ket{\phi^{0,0}}_{V_\mathrm{A}V_\mathrm{B}} \ket{\phi^{0,b}}_{A_1B_1} \ket{\phi^{0,-a}}_{A_2B_2} \ket{\psi}_C.
\]
Note that $V_\mathrm{A}V_\mathrm{B}$ and $A_1B_1$ are entangled as a result of the controlled-$\rX$ operation. Nevertheless, Alice and Bob still extract the correct error syndrome when they measure $A_1,A_2$ and $B_1,B_2$, respectively, and compare their measurement outcomes, as if Bob's QVC decoding operations were performed before Alice's $\br{\control{\rX}}_{V_\mathrm{A}A_1}$ operation. This is due to the fact that the error on the message register is reflected in the MES pair as phase errors and the phase error in each MES can still be detected correctly by local measurements in the Fourier basis even after the controlled-$\rX$ operation is applied. In the out-of-sync quantum hashing scenario described above a similar effect occurs. Finally, note that when Alice and Bob do not agree on the number of MES blocks they have used, they do not compare their quantum hash values unless a transmission error or a meta data hash collision occurs.  

\subsubsection{First representation of the joint quantum state} \label{subsec:Q-1st-rep-sync}

As in Section \ref{sec:general-descrpition-largeclasscial}, we start by introducing a first representation of the joint state, denoted $\JSone$, which in turn is simplified into a more informative representation. This latter representation, denoted $\JStwo$, is the representation which Alice and Bob need to compute correctly in order to make progress in simulation of the input protocol $\Pi$ and decide their next action in $\Pi'$. Recall that due to the recycling of MESs, each block of MES pairs may be used multiple times to encrypt messages using QVC. The representations $\JSone$ and $\JStwo$ defined below are valid only if the recycling has been successful so far in $\Pi'$, namely that Alice and Bob have recycled the same blocks of MES registers and that these registers were indeed in the $\ket{\phi^{0,0}}$ state when recycled. We prove that except with probability $2^{-\Theta\br{n\epsilon}}$ the recycling is successful throughout the algorithm.

\suppress{ In addition to the registers $ABCER$, the representations $\JSone$ and $\JStwo$ contain registers corresponding to MESs recycled for communication using QVC. To simplify the representations, in every iteration, the MES block being recycled is represented by a new block of MES registers all in the ideal $\ket{\phi^{0,0}}$ state, as if no recycling were done and a fresh block of MES pairs were available each time.}  

Recall that in the adversarial noise model, the adversary Eve can introduce arbitrary errors on the quantum communication register $C^\prime$ that passes through her hand subject to the constraints given by Eq.~\eqref{eqn:noise-model-1} and Eq.~\eqref{eqn:noise-model-2}.
\suppress{the global action of the adversary Eve on the $n'$ communication registers in $\Pi'$ has a representation with Kraus operators of weight at most $\epsilon n'$, i.e., each Kraus operator can be written as a linear combination of Pauli errors of weight at most $\epsilon n'$.} Furthermore, as explained in Section~\ref{subsec:out-of-sync QVC}, the algorithm is designed so that once the two parties agree on the number of blocks of MES pairs they have used, the error in the joint state due to out-of-sync QVC in any iteration is translated to a transmission error on the message registers, as if the MES blocks were used in sync and transmission errors were introduced by the adversary. We emphasize that in both cases, the error on the message register is a mixture of linear combinations of Pauli errors and once Alice and Bob measure a block of MES pairs to extract (part of) the syndrome, the error on the corresponding message register collapses to a Pauli error. Then the joint state can be written in terms of a new mixture of linear combinations of Pauli errors conditioned on the measurement outcomes, which are recorded in the Pauli data by the two parties. To simplify the joint state representation and the analysis of the algorithm, without loss of generality, we focus on a fixed but arbitrary error syndrome in any such linear combination of Pauli errors arising in the simulation protocol $\Pi'$. We prove the correctness of the algorithm for any such error syndrome which by linearity implies the correctness of the algorithm against any adversary defined in Section~\ref{sec:noisy_comm_model}. Let~$E \in \P_{d,n^\prime}$ be a Pauli error with~$\mathrm{wt}\br{E}\leq \epsilon n^\prime$. In the remainder of this section, we assume $E$ is the error introduced by the adversary into the~$n^\prime$ communicated qudits in $\Pi'$.

We first define the representations $\JSone$ and $\JStwo$ after $i$ iterations of the algorithm in the case when the two parties have used the same number of blocks of MES pairs, i.e., $\ellQVCA=\ellQVCB$. In Subsection~\ref{subsec:Q-2nd-rep-out-of-sync}, we explain how these representations are modified when Alice and Bob are out of sync in the number of MES blocks they have used. 

We sketch how to obtain $\JSone$ from $\FullMA$, $\FullMB$, $\FullPA$ and $\FullPB$, when the error syndrome is given by $W\in \br{\Sigma^2}^*$. Recall that when~$\ellQVCA=\ellQVCB$, there is a one-to-one correspondence between the state of each MES register and the Pauli error on the message register in the corresponding communication round. Therefore, in order to simplify the representations, without introducing any ambiguity, we omit the MES registers from the representations $\JSone$ and $\JStwo$. The first representation $\JSone$ of the joint state after $i$ iterations (when $\ellQVCA=\ellQVCB$) is given by 
\begin{align} \label{eqn:Q-JS1_1}
\JSone = [*\ellQVCA] \cdots [*2][*1] \ket{\psi_{\mathrm{init}}}^{A B C E R},
\end{align}
where $\ket{\psi_{\mathrm{init}}}^{A B C E R}$ is the initial state of the original input protocol $\Pi$ and the content of each bracket is described below. The $j$-th bracket corresponds to the $j$-th block of MES pairs which have been used by both Alice and Bob and contains from right to left $r$ iterations of the following:
\begin{quote}
	Alice's unitary operation - \\
	Pauli error on Alice's message - \\
	Bob's Pauli correction - Bob's unitary operation - \\
	Pauli error on Bob's message - \\
	Alice's Pauli correction.
\end{quote}
Similar to the teleporation-based protocol, we allow for an additional unitary operation by Alice on the far left when she implements a block of type $-1$. Using the same rules described in Section~\ref{sec:large-classical-second-rep}, in each bracket, the block of unitary operations of the input protocol $\Pi$ applied by Alice (if any) and her block type ($\pm1$ or $0$) can be computed from $\FullMA$. Moreover, her Pauli corrections are recorded in $\FullPA$ and the block of $\FullPA$ containing these Pauli corrections can be located using $\FullMA$. Recall that each block of $\FullPA$ may correspond to two different types of iterations: when Alice measures a block of MES pairs to extract the error syndrome she concatenates $\FullPA$ with a new block containing her measurement outcomes (with no Pauli corrections), whereas in iterations in which she communicates using QVC, she may apply Pauli corrections in between and records the Pauli corrections in $\FullPA$. Therefore, $\FullMA$ may be used to distinguish these two different types of blocks and locate the corresponding Pauli corrections in $\FullPA$. Similarly, in each bracket, the block of unitary operations of $\Pi$ applied by Bob, his block type and his Pauli corrections are obtained from $\FullMB$ and $\FullPB$. 
Finally, in $\JSone$, the Pauli errors on the messages in each bracket are specified in terms of the error syndrome $W=W_1W_2\ldots W_{\ellQVCA} \in {\br{\Sigma^2}}^{2r\times \ellQVCA}$ defined below.  
The communication in each iteration of the algorithm has two parts. In the first part, the parties use a constant number of rounds to communicate the pointers and hash values. Any transmission error introduced by the adversary on these messages only affects the actions of the two parties in the current and future iterations, which will be recorded in the metadata and Pauli data and reflected in the joint state representation. The second part involves $2r$ rounds of communication, in which either classical information is communicated (e.g., to reconcile inconsistencies in the data structures) or QVC is used to communicate quantum information (on one side or both). Transmission errors on these messages can directly modify the joint state and need to be separately taken into account in the joint state representation. Let $W'=W'_1W'_2\ldots W'_{\Rtotal}$ denote the error syndrome corresponding to the restriction of~$E$ to these messages over the $\Rtotal$ iterations of the algorithm, where each $W'_j$ is a string in $\br{\Sigma^2}^{2r}$ representing a Pauli error on $2r$ qudits. For every $j\in \Br{\ellQVCA}$, if the $j$-th block of MES pairs has been used in sync on both sides, say in iteration $j'$, we let $W_j=W'_{j'}$. Otherwise, the $j$-th block of MES pairs has been used out of sync and we define $W_j$ to be the error syndrome arising on the message registers in the corresponding communication rounds due to the remedial actions the parties take to recover from the out-of-sync QVC; see Section~\ref{subsec:out-of-sync QVC} for more details.\suppress{We set $W_j$ to be $\br{0^2}^{2r}$, for all $j\in \Br{\ellQVCA+1,i}$.} Each $W_j\in \br{\Sigma^2}^{2r}$ specifies the $2r$ Pauli errors in the $j$-th bracket from the right. 


Note that in order to compute the representation $\JSone$, one needs to know $\FullMA$, $\FullMB$, $\FullPA$, $\FullPB$ and $W$. This information is not necessarily available to Alice and Bob at any point during the simulation. In fact, We use the representation in order to analyze the progress in the simulation. Alice and Bob compute their best guess for $\JSone$ based on their estimate of each other's classical data. They only compute their estimates of $\JSone$ when they believe that they have full knowledge of each other's metadata and Pauli data, fully agree on the recycling data, have used the same number of MES blocks and have measured all blocks of MES pairs which were corrupted in communication using QVC or due to out-of-sync QVC. 

Alice's estimate $\JSoneA$ of $\JSone$ is of the same form as in Equation \eqref{eqn:Q-JS1_1}, except she uses her best guess of $\FullMB$ and $\FullPB$ in the above procedure. Moreover, the string $W$ in $\JSone$ is replaced by $W^\mathrm{A}=W^\mathrm{A}_1 \ldots W^\mathrm{A}_{\ellQVCA} \in {\br{\Sigma^2}}^{2r\times \ellQVCA}$ computed by Alice as follows. For every $j\in\Br{\ellQVCA}$, 
\begin{itemize}
	\item if $\RA\Br{j}=\sS$, then she sets $W^\mathrm{A}_j=\br{0^2}^{2r}$. Recall that when Alice computes $\JSoneA$, she believes that they have used the same number of MES blocks and have both already measured all blocks of MES pairs which were corrupted in communication using QVC. Therefore, in Alice's view the remaining rounds of communication using QVC have not been corrupted.
	\item Otherwise, $\RA\Br{j}=\RBtilde\Br{j}=\sM$, i.e., Alice has measured the corresponding block of MES pairs and believes Bob has measured them as well. Using $\FullMA$ and $\MBtilde$, Alice locates the corresponding measurement outcomes in $\FullPA$ and $\PBtilde$ and sets $W^\mathrm{A}_j$ to be the error syndrome obtained from the measurement outcomes.
\end{itemize} 

Note that if Alice computes $\JSoneA$ in an iteration with no transmission errors or hash collisions, then the computed representation $\JSoneA$ is indeed equal to $\JSone$. Bob computes $\JSoneB$ similarly based on his best estimate of $\FullMA$ and $\FullPA$.

\subsubsection{Second representation of the joint quantum state} \label{subsec:Q-2nd-rep-sync}

The representation $\JStwo$ is obtained from $\JSone$ as follows. In $\JSone$, starting from the rightmost bracket, we recursively try to cancel consecutive brackets if their contents correspond to inverse of one another. Once no further such cancellation is possible, what we are left with is the $\JStwo$ representation, which is of the following form (when $\ellQVCA=\ellQVCB$):
\begin{align} \label{eqn:Q-JS2_1}
\JStwo = [\#b] \cdots [\#1] [U_{gr} \cdots U_{(g-1)r + 2} U_{(g-1)r + 1} ] 
\cdots [U_r \cdots U_2 U_1]
\ket{\psi_{\mathrm{init}}}^{A B C E R},
\end{align}
where $g$ is the largest integer such that the concatenation of the first $g$ brackets starting from the right equals the sequence $U_{gr},\ldots,U_2,U_1$ of unitary operations of $\Pi$. As in Section~\ref{sec:general-descrpition-largeclasscial}, we refer to these brackets as the ``good'' blocks, and the remaining $b$ brackets which need to be actively rewound are called the ``bad'' blocks. 

Once the parties have synchronized their classical data (to the best of their knowledge) as described earlier, they have all the information they need to compute their best guesses $\JSoneA$ and $\JSoneB$ of $\JSone$. Using the same procedure described above, from $\JSoneA$ Alice computes her estimate $\JStwoA$ of $\JStwo$. Similarly, Bob computes $\JStwoB$ from $\JSoneB$. Note that the two parties may have different views of their joint state, based on which they decide how to further evolve the state in $\Pi'$. The rules by which Alice and Bob decide their respective types ($\pm1$ or $0$) for the next block in $\Pi'$, and which blocks of unitary operations of $\Pi$ (if any) are involved, are the same as the teleportation-based protocol; see Section~\ref{sec:large-classical-second-rep}.

\subsubsection{Representation of the joint state while out-of-sync} \label{subsec:Q-2nd-rep-out-of-sync}

We now define the representations $\JSone$ and $\JStwo$ in the case when $\ellQVCA \neq \ellQVCB$. Note that in this case similar to the teleportation-based protocol, conditioned on the classical data with Alice and Bob and a fixed error syndrome $W'$ by the adversary, $\JSone$ and $\JStwo$ represent a pure state. However, in addition to the $ABCER$ registers we need to include the MES blocks which have been used by only one party. Let $u \defeq |\ellQVCA -\ellQVCB |$. For concreteness suppose that $\ellQVCA > \ellQVCB$. Then the $\JSone$ representation is of the following form: 
\begin{align}\label{eqn:Q-JS1-OoS}
\JSone = [*\ellQVCA] \cdots [*\ellQVCB] \cdots [*2] [*1]
\ket{\psi_{\mathrm{init}}}^{A B C  E R}\enspace.
\end{align}
The content of the first $\ellQVCB$ brackets from the right corresponding to the MES blocks which have been used by both parties are obtained as described in Subsection~\ref{subsec:Q-1st-rep-sync}. The leftmost $u$ brackets, correspond to the MES blocks which have been used only by Alice. We refer to these blocks as the \emph{ugly\/} blocks. These brackets contain Alice's unitary operations from the input protocol $\Pi$, her Pauli correction operations and QVC encoding and decoding operations, as well as all the MES registers involved in these iterations which remain untouched on Bob's side.

The representation $\JStwo$ is obtained from $\JSone$ as follows: We denote by $[@u]\cdots [@1]$ the leftmost $u$ brackets corresponding to the ugly blocks. We use the procedure described in Subsection~\ref{subsec:Q-2nd-rep-sync} on the rightmost $\ellQVCB$ brackets in $\JSone$ to obtain $\JStwo$ of the following form:
\begin{align} \label{eqn:Q-JS2_OoS}
\JStwo = [@u] \cdots [@1] [\#b] \cdots [\#1] [U_{gr} \cdots U_{(g-1)r + 2} U_{(g-1)r + 1} ] 
\cdots [U_r \cdots U_2 U_1]
\ket{\psi_{\mathrm{init}}}^{A B C E R}\enspace,
\end{align}
for some non-negative integers $g$ and $b$, which we refer to as the number of \emph{good\/} blocks and the number of \emph{bad\/} blocks in $\JStwo$ representation, respectively. We point out that Alice and Bob do not compute their estimates of $\JSone$ and $\JStwo$ unless, based on their view of the simulation so far, they believe that they have used the same number of MES blocks. Therefore, whenever computed, $\JSoneA,\JSoneB$ and $\JStwoA,\JStwoB$ are always of the forms described in Subsections~\ref{subsec:Q-1st-rep-sync} and~\ref{subsec:Q-2nd-rep-sync}, respectively. Note that Alice and Bob can realize that $\ellQVCA \neq \ellQVCB$ by learning each other's meta data. Then if they do as described in Subsection~\ref{subsec:out-of-sync QVC}, if no error or collision occurs, they will reduce the number of ugly blocks in $\JStwo$ by one. 

\subsubsection{Constant collision probability for classical hashing suffices}

As in the teleportation-based algorithm, in our algorithm in the plain model the hash function of Lemma~\ref{lem:hashes} is used to check for inconsistencies in the classical data maintained by Alice and Bob. Recall that in Section~\ref{sec:BKlarge}, the collision probability $p$ for the hash function $h$ of Lemma~\ref{lem:hashes} is chosen to be $1/\mathrm{poly}(n)$. The output of the hash function is of length $o=\Theta\br{\log \frac{1}{p}}=\Theta(\log n)$ bits. Therefore, in the large-alphabet case the hash values corresponding to the classical data can be communicated using only a constant number of rounds. However, in the small-alphabet case, using a logarithmic number of rounds to communicate the hash values leads to a vanishing simulation rate. In our algorithm in this section, we use the hash family of Lemma~\ref{lem:hashes} with a constant collision probability and show that $p=\Theta\br{1}$ suffices to keep the number of hash collisions low. We address this issue in the simpler large-alphabet setting to simplify the proof in the small-alphabet case at the conceptual level.

Following Haeupler \cite{Haeupler:2014}, we circumvent the barrier explained above using the observation that hashing only makes one-sided errors. In other words, collision only occurs when the data to be compared are not equal, which in turn is a result of corruptions by the adversary. As the error budget of the adversary is bounded by $2n\epsilon$, one would expect the total number of rounds in which the classical data being compared are not equal to be bounded by $\Theta\br{n\epsilon}$. In fact, this allows us to have a constant collision probability while keeping the number of hash collisions in the same order as the number of transmission errors. 

\subsubsection{Summary of main steps}

In Algorithm~\ref{algo:Main steps-large-alphabet-Yao}, we summarize the outline of the steps which are followed by Alice and Bob in the simulation. Note that since synchronizing recycling data creates new Pauli data, we choose to do this step before synchronizing the Pauli data.  Similar to the teleportation-based case, the algorithm is designed so that unless there is a transmission error or a hash collision in comparing a given type of data, Alice and Bob will simultaneously go down these steps while never returning to a previous step. This in fact is a crucial property used in the analysis of the algorithm. 

\RestyleAlgo{boxruled}
\begin{algorithm}
	 Agree on the history of the simulation contained in metadata, i.e., ensure $\FullMA= \MAtilde$ and $\FullMB=\MBtilde$. This involves Algorithm~\ref{algo:rewindMD}---\textbf{\textsf{rewindMD}} and Algorithm~\ref{algo:extendMD}---\textbf{\textsf{extendMD}}.\\
	 
	 Synchronize the number of MES blocks used in QVC, in particular, ensure $\ellQVCA=\ellQVCBtilde$ and $\ellQVCB=\ellQVCAtilde$. This is done via Algorithm~\ref{algo:Q-syncMES}---\textbf{\textsf{Q-syncMES}}.\\
	 
	 Agree on the measurement pointers and the recycling data up to the pointers, in particular, ensure $\br{\ellRA,\RA\Br{1:\ellRA}}=\br{\ellRBtilde,\RBtilde\Br{1:\ellRBtilde}}$ and $\br{\ellRB,\RB\Br{1:\ellRB}}=\br{\ellRAtilde,\RAtilde\Br{1:\ellRAtilde}}$. This involves Algorithm~\ref{algo:rewindRD}---\textbf{\textsf{rewindRD}}.\\
	 
	 Ensure no undetected quantum error from earlier rounds exists. This is done by ensuring $\QHA=\QHB$ and involves Algorithm~\ref{algo:QmeasureEPR}---\textbf{\textsf{measuresyndrome}}.\\
	 
	 Ensure $\ellRA=\ellQVCA$ and $\ellRB=\ellQVCB$. This is achieved via Algorithm~\ref{algo:extendRD}---\textbf{\textsf{extendRD}}.\\
	 
	 Agree on Pauli data, in particular, ensure $\FullPA=\PAtilde$ and $\FullPB=\PBtilde$. This is done via Algorithm~\ref{algo:Q-rewindPD}---\textbf{\textsf{Q-rewindPD}} and Algorithm~\ref{algo:Q-extendPD}---\textbf{\textsf{Q-extendPD}}.\\
	 
	 Compute the best guess for $\JSone$ and $\JStwo$. If there are any
	 ``bad'' blocks in the guess for $\JStwo$, reverse the last
	 bad block of unitary operations. I.e.,
	 implement quantum rewinding so that~$b = 0$ in $\JStwo$. This
	 is done in Algorithm~\ref{algo:Q-simulate}---\textbf{\textsf{Q-simulate}}.\\
	 
	 If no ``bad'' blocks remain, implement the next block of
	 rounds of the original protocol. This results in an increase in $g$ in
	 $\JStwo$, and is also done through Algorithm~\ref{algo:Q-simulate}---\textbf{\textsf{Q-simulate}}.\\

	\caption{Main steps in one iteration of the simulation for the large alphabet recycling-based model}
	\label{algo:Main steps-large-alphabet-Yao}
\end{algorithm}
\RestyleAlgo{ruled}

Note that although in step $1$ we use the same algorithms for synchronizing the metadata as in the teleportation-based case  (\textsf{rewindMD} and \textsf{extendMD}), the alphabet over which the metadata strings are defined are now different (see Subsection~\ref{sec:data-structure-large-alphabet-Yao}). The algorithms mentioned in the remaining step are presented in the next section. Figure~\ref{fig:flow-qvc} summarizes the main steps in flowchart form.

\begin{figure}[!t]
\centering
\includegraphics[width=475pt]{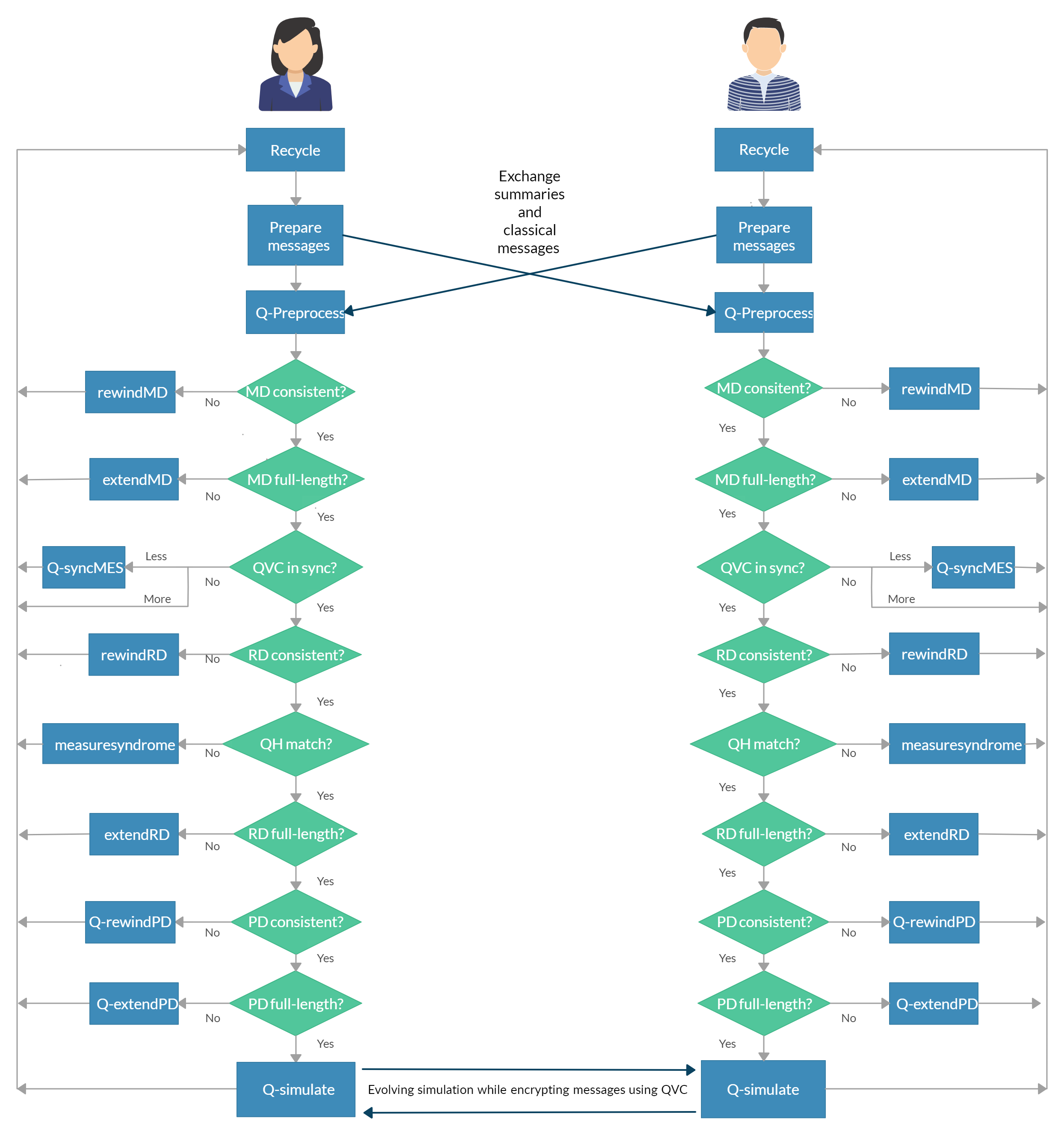}
\caption{Flowchart of the recycling-based scheme for high rate noisy interactive quantum communication.}
\label{fig:flow-qvc}
\end{figure}

\subsection{Algorithm}

For each subroutine, we list all the global variables accessed by the subroutine as the \textbf{Input} at the beginning of the subroutine. Whenever applicable, the relation between the variables when the subroutine is called is stated as the \textbf{Promise} and the global variables which are modified by the subroutine are listed as the \textbf{Output}.

\subsubsection{Data structure}\label{sec:data-structure-large-alphabet-Yao}
Alice and Bob maintain a data structure obtained by modifying the one introduced in Section~\ref{sec:datastructurelargealphabetcleveburhman}.
\begin{itemize}
	\item \textbf{Metadata:} The metadata now contain new symbols corresponding to recycling operations and error detection measurements. In every iteration $\NewMetaA \in \{\mathsf{\pm1},\mathsf{0},\sC,\OED,\sM,\sC'\}$ specifies Alice's action in the current iteration. Similar to the teleportation-based algorithm, in an iteration with $\NewMetaA=\sC$, Alice does not access the quantum registers and if $\NewMetaA\in \{\mathsf{\pm1},\mathsf{0}\}$, then  the unitary operators of the original protocol applied locally by Alice in the current iteration have exponent $\NewMetaA$. If $\NewMetaA=\sM$ Alice measures the block of MES registers specified by her measurement pointer and moves the pointer back by $1$. $\NewMetaA=\sC'$ corresponds to a classical iteration for Alice in which she just moves her measurement pointer forward by $1$. Finally, in an iteration with $\NewMetaA=\OED$, Alice completes QVC on her side as explained in Section~\ref{subsec:out-of-sync QVC}. The notation $\OED$ is used to emphasize that similar to an iteration with $\NewMetaA=\mathsf{0}$ no unitary operator of the original protocol is applied by Alice, but she measures the block of MES registers at the end of the iteration for error detection. 
	
	Alice records her metadata in $\FullMA$ which is concatenated with $\NewMetaA$ in each iteration and has length $i$ at the end of the $i$-th iteration. Similar to the teleportation-based algorithm, Alice maintains a string $\MBtilde$ as an estimate of Bob's metadata, which is not necessarily full-length. The length of $\MBtilde$ is denoted by $\ellMBtilde$. Alice also maintains $\ellMA$, her estimate for the length of $\MAtilde$, which is with Bob. $\MA$ is defined as the prefix of $\FullMA$ of length $\ellMA$, i.e., $\MA\defeq \FullMA\Br{1:\ellMA}$. When $\MA$ appears in any of the algorithms in this section, it is implicitly computed by Alice from $\FullMA$ and $\ellMA$. We denote by $\ellQVCA$ the number of iterations in which Alice has performed QVC so far. Note that $\ellQVCA$ can be computed from $\FullMA$ and is equal to the number of $\mathsf{\pm1}$, $\mathsf{0}$, $\OED$ symbols in $\FullMA$. Bob's local metadata variables are defined similarly.
	
	\item \textbf{Recycling data:} Quantum recycling data are used to decide whether a block of MES pairs can be reused for communication using QVC. The recycling data alphabet consists of the symbol $\sM$ corresponding to a measured block of MES pairs and the symbol $\sS$, used for a block of MES pairs still in superposition. Alice records her recycling data in a string $\RA$, which is also computable from $\FullMA$. Alice maintains a queue, $\IndexA$, of indices corresponding to recycled MES blocks to be reused in QVC. In every iteration, $\NextMESIndexA$ holds the index of the recycled MES block by Alice and the string $\IndexA$ is concatenated with the index. If no such index exits in an iteration, it is assigned the value $\perp$ and the protocol aborts. $\ellNextMESA$ is a pointer on $\RA$ and $\IndexA$ used by Alice in the recycling process. In every iteration, Alice moves this pointer to the next $\sS$ symbol in $\RA$ (if it exists) and records the value of $\IndexA$ in this coordinate into $\NextMESIndexA$. Alice also maintains a measurement pointer, $\ellRA$, which specifies the block of MES pairs to be measured in iterations with $\NewMetaA=\sM$. The measurement pointers serve as reference points up to which Alice and Bob compare their recycling data in each iteration. The pointer $\ellRA$ can stay the same or move backward or forward by $1$ and can be computed from $\FullMA$. In every iteration, if the protocol does not abort, $\RA$ is concatenated with an $\sS$ symbol corresponding to the recycled MES block and if an MES block is measured the corresponding element of $\RA$ is changed from $\sS$ to $\sM$. Bob's recycling data variables are defined similarly. Alice computes $\RBtilde$ and $\ellRBtilde$ as her estimate of Bob's $\RB$ and $\ellRB$, respectively, based on her full-length estimate $\MBtilde$ of $\FullMB$. Note that if $\MBtilde=\FullMB$ then Alice's estimates of Bob's recycling data are correct. 
	
	\item \textbf{Pauli data:} In any iteration, new Pauli data is generated on Alice's side if and only if she measures a block of MES pairs or performs QVC locally. $\NewPauliA$ has three parts: If a block of $r$ MES pairs is measured locally by Alice in the current iteration then the measurement outcome, $(m_1,m_2)\in \Sigma^{4r}$, is recorded in the first two parts of $\NewPauliA$, and the third part contains $\perp^{2r}$, corresponding to no Pauli corrections. Otherwise, if Alice performs QVC then $\perp^{2r}$ is recorded in each of the first two parts and the third part similar to the teleportation-based protocol specifies the Pauli corrections. 
	
	Alice records her Pauli data in $\FullPA$. Starting from the empty string, $\FullPA$ is concatenated with $\NewPauliA$ whenever Alice measures an MES block or performs QVC. She maintains a string $\PBtilde$ as an estimate of Bob's Pauli data. The length of $\PBtilde$ is denoted by $\ellPBtilde$. Alice also maintains $\ellPA$, her estimate for the length of $\PAtilde$, which is with Bob. $\PA$ denotes the prefix of $\FullPA$ of length $\ellPA$, i.e., $\PA\defeq \FullPA\Br{1:\ellPA}$. When $\PA$ appears in any of the algorithms in this section, it is implicitly computed by Alice from $\FullPA$ and $\ellPA$. We define $q_{\MA}\defeq\left| \FullPA\right| / 6r$. Note that $q_{\MA}$ can be computed from $\FullMA$. Alice computes her estimate $q_{\MBtilde}$ of $q_{\MB}$ using her full-length estimate $\MBtilde$ of $\FullMB$. Bob's Pauli data variables are defined similarly.
	
	\item As in Section~\ref{sec:BKlarge}, we use $H$ with different variables as subscript to represent hash values, e.g., $\HMA$ represents a hash value corresponding to $\MA$. We use $\QHA$ and $\QHB$ to represent quantum hash values. The data variables with a superscript $'$ denote the received data after transmission over the noisy channel, e.g., $\ellMB'$ denotes what Alice receives when Bob sends $\ellMB$. 
	
	\item The variable $\Itertype$ takes two new values: $\RD$ corresponding to recycling data synchronization and $\QH$ corresponding to an iteration in which the received quantum hash value does not match the locally computed quantum hash value.
\end{itemize}

Finally, $\LQVC$ denotes the total number of MES blocks to be used as the keys in QVC.

\begin{remark} We point out an additional subtlety in interpreting the Pauli data by the two parties. Recall that when Alice and Bob compute their estimates of the joint state, they use the metadata to locate in the Pauli data, the measurement outcomes for each block of measured MES pairs. However, due to transmission errors or collisions, it is possible to have an inconsistency between the metadata and the Pauli data. For instance, $\MAtilde$ may indicate that a specific block of $\PAtilde$ contains the measurement outcomes on an MES block, whereas it actually has $\perp^{2r}$ in the first two parts and corresponds to an iteration in which Alice has performed QVC. In any such scenario, Alice and Bob interpret the $\perp$ symbols as $0$ and compute the joint state. This most likely introduces new errors on the joint state from which Alice and Bob can recover once they obtain a correct view of the simulation.
\end{remark}

\subsubsection{Pseudo-codes}

This section contains the pseudo-codes for the main algorithm and the subroutines that each party runs locally in the simulation protocol.


\begin{algorithm}
	\Input{$n$ round protocol $\Pi$ in plain quantum model over polynomial-size alphabet $\Sigma$}
	\BlankLine
	\textbf{\textsf{Q-Initialization}}\;
	\BlankLine	
	\SetKwProg{ForLoop}{For}{}{}
	\SetAlgoLined
	\ForLoop{$i = 1 \to \Rtotal$}	
	{   
		\SetAlgoVlined
		\textbf{\textsf{Recycle}}\;
		\If {$\NextMESIndexA=\perp$} 
			{\text{Abort}\;}
		$\IndexA \leftarrow \br{\IndexA,\NextMESIndexA}$\;
		$\RA \leftarrow \br{\RA,\sS}$\;
		\Comment*[f] {\textbf{computing hash values}}\\
		$\HMA \leftarrow h_{S_{4i-3}}\br{\MA}$\;
		$\HMBtilde \leftarrow h_{S_{4i-2}}\br{\MBtilde}$\;
		$\HPA \leftarrow h_{S_{4i-1}}\br{\PA}$\; 
		$\HPBtilde \leftarrow h_{S_{4i}}\br{\PBtilde}$\;
		\textbf{\textsf{Quantum-hash}}\;  
		\BlankLine
		Send 
		$$\br{\HMA,\ellMA,\HMBtilde,\ellMBtilde,\HPA,\ellPA,\HPBtilde,\ellPBtilde,\QHA};$$\\
		Receive
		$$\br{\HMAtilde',\ellMAtilde',\HMB',\ellMB',\HPAtilde',\ellPAtilde',\HPB',\ellPB',\QHB'};$$\\
		
	    \BlankLine
		\textbf{\textsf{Q-Preprocess}};
		\Comment*[f] {\textbf{Determining iteration type}}\\
		\BlankLine
		
		\If   {$\Itertype \neq \SIM$}
			{   Send $\msg$\;
				Receive $\msg^\prime$\;
			}
		\tcp*[f]{\textbf{messages are communicated alternately}}\\    
		\BlankLine
		\Comment*[f] {\textbf{Case i.A}}\\
		\If   {$\Itertype = \MD \;\mathrm{and}\; \RewindExtend = \sR$}
		{   \textbf{\textsf{rewindMD}}\;
		}	
		\Comment*[f] {\textbf{Case i.B}}\\
		\ElseIf {$\Itertype = \MD \;\mathrm{and}\; \RewindExtend = \sE$}
		{   \textbf{\textsf{extendMD}}\;
		}
		\Comment*[f] {\textbf{Case ii.A}}\\
		\ElseIf{$\Itertype = \MES \;\mathrm{and}\; \NewMetaA=\sC$}
		{   return\;
		}
		\Comment*[f] {\textbf{Case ii.B}}\\
		\ElseIf{$\Itertype = \MES \;\mathrm{and}\; \NewMetaA=\OED$}
		{   \textbf{\textsf{Q-syncMES}}\;
		}		
			
	}
	\caption{{\textbf{\textsf{Q-Main}} (Alice's side) }}
\end{algorithm}

\setcounter{algocf}{12}

\begin{algorithm}
	\setcounter{AlgoLine}{26}
	\SetKwBlock{Begin}{}{}
	\Begin{
	    \Comment*[f] {\textbf{Case iii}}\\
		\ElseIf{$\Itertype=\RD \;\mathrm{and}\; \RewindExtend=\sR$}
		{
			\textbf{\textsf{rewindRD}}\;
		}
		\Comment*[f] {\textbf{Case iv}}\\
		\ElseIf{$\Itertype=\mathrm{QH}$}
		{
			\textbf{\textsf{measuresyndrome}}\;
		}
		\Comment*[f] {\textbf{Case v}}\\
		\ElseIf{$\Itertype =\RD \;\mathrm{and}\; \RewindExtend=\sE$}
		{
			\textbf{\textsf{extendRD}}\;
		}
		\Comment*[f] {\textbf{Case vi.A}}\\
		\ElseIf   {$\Itertype = \PD \;\mathrm{and}\; \RewindExtend = \sR$}
		{   \textbf{\textsf{rewindPD}}\;
		}
		\Comment*[f] {\textbf{Case vi.B}}\\
		\ElseIf{$\Itertype = \PD \;\mathrm{and}\; \RewindExtend = \sE$}
		{   \textbf{\textsf{extendPD}}\;
		}
		\Comment*[f] {\textbf{Case vii}}\\	
		\Else{\textbf{\textsf{Q-Simulate}}\;}
	}		
	
	\Return{\textup{\textbf{\textsf{Q-Main}}}}\;
	\caption{\textbf{\textsf{Q-Main}} (Alice's side, cont. from previous page) }
	\label{algo:MainalgorithmQMessage}
\end{algorithm}


\begin{algorithm}
	
	\Input{ 
		    $$\br{ 
		    	  \begin{array}{c}
			        \HMA,\ellMA,\HMBtilde,\ellMBtilde,\HPA,\ellPA,\HPBtilde,\ellPBtilde,\QHA \\
			        \HMAtilde',\ellMAtilde',\HMB',\ellMB',\HPAtilde',\ellPAtilde',\HPB',\ellPB',\QHB' \\
		            \FullMA,\ellQVCA,\MBtilde,\RA,\ellRA,\FullPA,\PBtilde
		          \end{array}
		         }$$
	      }
	
	\Output{$\br{\Itertype, \RewindExtend, \NewMetaA, \FullMA, \ellMA, \NewMetaBtilde, \MBtilde, \ellMBtilde, \msg}$}
	
	\BlankLine
	\If   
	    {  
	        $\br{\HMA,\HMBtilde}=\br{\HMAtilde',\HMB'}  \;\mathrm{and}\; \ellMA=\ellMAtilde'=\ellMBtilde=\ellMB'=i-1$
	    }
	    {   Compute $\ellQVCBtilde$, $\RBtilde$, $\ellRBtilde$, $q_{\MBtilde}$\;
	    }
	
	\BlankLine 
	\Comment*[f] {\textbf{Processing Metadata}}\\
	\Comment*[f] {\textbf{Case i.A}}\\
	\If
	   {
	   	   $\br{\HMA,\HMBtilde,\ellMA,\ellMBtilde}\neq 			           \br{\HMAtilde',\HMB',\ellMAtilde',\ellMB'}$
       }
	   {   $\Itertype \leftarrow \MD$\;
	   	   $\RewindExtend \leftarrow \sR$\;
	   	   $\NewMetaA \leftarrow \sC$\;
	   	   $\FullMA \leftarrow \br{\FullMA,\NewMetaA}$\;	
	   	   $\msg \leftarrow \text{dummy message of length } r$\;
	   }
   
   	\caption{\textbf{\textsf{Q-Preprocess}} (Alice's side) }
   \label{algo:QPreprocess}
\end{algorithm}

\setcounter{algocf}{13} 

\begin{algorithm}
\setcounter{AlgoLine}{8}

    \Comment*[f] {\textbf{Case i.B}}\\
	\ElseIf
	   {   $\br{\ellMA < i-1} \;\mathrm{or}\; \br{\ellMBtilde < i-1}$
	   }
	   {   $\Itertype \leftarrow \MD$\;
	       $\RewindExtend \leftarrow \sE$\;
	   	   $\NewMetaA \leftarrow \sC$\;
	   	   $\FullMA \leftarrow \br{\FullMA,\NewMetaA}$\;
	   	   \If   {$\ellMA < i-1$}
	             {   $\msg \leftarrow \mathrm{encodeMD}\br{\FullMA\Br{\ellMA+1,\ellMA+2}}\!;$ 
	             	 \tcp*[f]{\textbf{Encode MD in $\Sigma^r$}}\\
	             }
	        \Else
	             {   $\msg \leftarrow \text{dummy message of length } r$\;
	             }      	
	   }

    \Comment*[f] {\textbf{Comparing number of used MES blocks}}\\
   
	\Comment*[f] {\textbf{Case ii.A}}\\
	\ElseIf   {$\ellQVCA > \ellQVCBtilde$}
	   {   $\Itertype \leftarrow \MES$\;
	   	   $\NewMetaA \leftarrow \sC$\;
	   	   $\FullMA \leftarrow \br{\FullMA,\NewMetaA}$\;
           $\ellMA \leftarrow \ellMA+1$\;
	   	   $\NewMetaBtilde \leftarrow \OED$\;
           $\MBtilde \leftarrow \br{\MBtilde,\NewMetaBtilde}$\;
           $\ellMBtilde \leftarrow \ellMBtilde+1$\;
		   $\msg \leftarrow \text{dummy message of length } r$\;
	   }

  	\Comment*[f] {\textbf{Case ii.B}}\\
   	\ElseIf   {$\ellQVCA < \ellQVCBtilde$}	
	   {   $\Itertype \leftarrow \MES$\;
	   	   $\NewMetaA \leftarrow \OED$\;
		   $\FullMA \leftarrow \br{\FullMA,\NewMetaA}$\;
		   $\ellMA \leftarrow \ellMA+1$\;
		   $\NewMetaBtilde \leftarrow \sC$\;
		   $\MBtilde \leftarrow \br{\MBtilde,\NewMetaBtilde}$\;
		   $\ellMBtilde \leftarrow \ellMBtilde+1$\;
		   $\msg \leftarrow \text{dummy message of length }r$\;
	   }

   \caption{\textbf{\textsf{Q-Preprocess}} (Alice's side, cont. from previous page)}
\end{algorithm}

\setcounter{algocf}{13}

\begin{algorithm}
\setcounter{AlgoLine}{35}

    \Comment*[f] {\textbf{Processing recycling data}}\\
    
	\Comment*[f] {\textbf{Case iii}}\\
	\ElseIf   { $\br{\ellRA, \RA\Br{1:\ellRA}} \neq \br{\ellRBtilde, \RBtilde\Br{1:\ellRBtilde}}$}
		{	$\Itertype \leftarrow \RD$\;
			$\RewindExtend\leftarrow \sR$\;
		  	\If   {$\ellRA > \ellRBtilde$}
		    	{   $\NewMetaA \leftarrow \sM$\;
		        	$\NewMetaBtilde \leftarrow \sC$\;
		   	    }
		  	\ElseIf   {$\ellRA < \ellRBtilde$} 
		    	{   $\NewMetaA \leftarrow \sC$\;
		     	    $\NewMetaBtilde \leftarrow \sM$\;
		   	    }
		  	\Else
		   		{   $\NewMetaA \leftarrow \sM$\;
		        	$\NewMetaBtilde \leftarrow \sM$\;
		    	}	       
					$\FullMA \leftarrow \br{\FullMA,\NewMetaA}$\;
					$\ellMA \leftarrow \ellMA+1$\;
					$\MBtilde \leftarrow \br{\MBtilde,\NewMetaBtilde}$\;
					$\ellMBtilde \leftarrow \ellMBtilde+1$\;
					$\msg \leftarrow \text{dummy message of length }r$\;
		}

    \Comment*[f] {\textbf{Case iv}}\\
    \ElseIf{$\QHA\neq \QHB'$}
    {   $\Itertype \leftarrow \QH$\;
    	$\NewMetaA \leftarrow \sM$\;
    	$\FullMA \leftarrow \br{\FullMA,\NewMetaA}$\;
    	$\ellMA \leftarrow \ellMA+1$\;
    	$\NewMetaBtilde \leftarrow \sM$\;
    	$\MBtilde \leftarrow \br{\MBtilde,\NewMetaBtilde}$\;
    	$\ellMBtilde \leftarrow \ellMBtilde+1$\;
    	$\msg \leftarrow \text{dummy message of length }r$\;
    }

    \Comment*[f] {\textbf{Case v}}\\
\ElseIf{$\ellRA < \ellQVCA $}
{	$\Itertype \leftarrow \RD$\;
	$\RewindExtend\leftarrow \sE$\;
	$\NewMetaA \leftarrow \sC'$\;
	$\FullMA \leftarrow \br{\FullMA,\NewMetaA}$\;
	$\ellMA \leftarrow \ellMA+1$\;
	$\NewMetaBtilde \leftarrow \sC'$\;
	$\MBtilde \leftarrow \br{\MBtilde,\NewMetaBtilde}$\;
	$\ellMBtilde \leftarrow \ellMBtilde+1$\;
	$\msg \leftarrow$ dummy message of length $r$\; 	
}

\caption{\textbf{\textsf{Q-Preprocess}} (Alice's side, cont. from previous page)}
\end{algorithm}

\setcounter{algocf}{13}
\begin{algorithm}
\setcounter{AlgoLine}{71}

    \Comment*[f] {\textbf{Processing Pauli data}}\\
    \Comment*[f] {\textbf{Case vi.A}}\\
    \ElseIf   {$\br{\HPA,\HPBtilde,\ellPA,\ellPBtilde}\neq \br{\HPAtilde',\HPB',\ellPAtilde',\ellPB'}$
    }
    {   $\Itertype \leftarrow \PD$\;
    	$\RewindExtend \leftarrow \sR$\;
    	$\NewMetaA \leftarrow \sC$\;
    	$\FullMA \leftarrow \br{\FullMA,\NewMetaA}$\;
    	$\ellMA \leftarrow \ellMA+1$\;
    	$\NewMetaBtilde \leftarrow \sC$\;
    	$\MBtilde \leftarrow \br{\MBtilde,\NewMetaBtilde}$\;
    	$\ellMBtilde \leftarrow \ellMBtilde+1$\;
    	$\msg \leftarrow \text{dummy message of length }r$\;
    }

    \Comment*[f] {\textbf{Case vi.B}}\\
    \ElseIf   {$\br{\ellPA < 6q_{\MA} \cdot r} \;\mathrm{or}\; \br{\ellPBtilde <
    		6q_{\MBtilde} \cdot r}$}
    {   $\Itertype \leftarrow \PD$\;
    	$\RewindExtend \leftarrow \sE$\;
    	$\NewMetaA \leftarrow \sC$\;
    	$\FullMA \leftarrow \br{\FullMA,\NewMetaA}$\;
    	$\ellMA \leftarrow \ellMA+1$\;
    	$\NewMetaBtilde \leftarrow \sC$\;
    	$\MBtilde \leftarrow \br{\MBtilde,\NewMetaBtilde}$\;
    	$\ellMBtilde \leftarrow \ellMBtilde+1$\;
    	\If   {$\ellPA < 6q_{\MA} \cdot r$}
    	{   $\msg \leftarrow {\FullPA}\Br{\ellPA+1,\ellPA+r}$
    	}	
    }

    \Comment*[f] {\textbf{Case vii}}\\
    \Else
    {   \textbf{\textsf{Q-Computejointstate}}\;
        $\Itertype \leftarrow \SIM$\;
    	$\FullMA=\br{\FullMA,\NewMetaA}$\;
    	$\ellMA \leftarrow \ellMA+1$\;
    	$\MBtilde \leftarrow \br{\MBtilde,\NewMetaBtilde}$\;
    	$\ellMBtilde \leftarrow \ellMBtilde+1$\;
    }	
	
	\Return{\textup{\textbf{\textsf{Q-Preprocess}}}}\;
	\caption{\textbf{\textsf{Q-Preprocess}} (Alice's side, cont. from previous page)}
\end{algorithm}


\begin{algorithm}
	
	Initialize \\
	\nonl          $\qquad \LQVC \leftarrow\Theta\br{n\epsilon}$ \;
	\nonl          $\qquad r \leftarrow\Theta\br{1/\sqrt{\epsilon}}$ \; 
	\nonl          $\qquad \Rtotal \leftarrow \lceil n/{2r}+\Theta(n\eps)\rceil$ \;
	\nonl          $\qquad t \leftarrow \Theta\br{n\epsilon}$ \;
	\nonl		   $\qquad \ellMA,\ellMBtilde,\ellPA,\ellPBtilde,\ellQVCA,\ellRA,\ellNextMESA \leftarrow 0$ \; 
	\nonl		   $\qquad \FullMA,\MBtilde,\FullPA,\PBtilde\leftarrow \emptyset$ \; 
	\nonl		   $\qquad \RA \leftarrow (\sS,\ldots,\sS) \textrm{ of length } t$ \;
	\nonl		   $\qquad \IndexA \leftarrow \br{\LQVC-t+1,\ldots,\LQVC}$	\;
	\tcp*[f]{\textbf{The indexing of the strings $\RA$ and $\IndexA$ starts from $-t+1$}}\\
	
	\BlankLine		
	$h\leftarrow$ hash function of Lemma~\ref{lem:hashes} with $p=\Theta(1), o=\Theta(1), s=\Theta\br{\log n}$ \;
	
	\textbf{\textsf{Robust Entanglement Distribution}($\Theta\br{n\sqrt{\epsilon}}$)} \;
	
	Reserve $\LQVC\cdot 4r$ MES pairs to be used as the keys in QVC \;
	
	Reserve $10\Rtotal$ MESs to be used in quantum hashing \;
	
	Measure $\Theta\br{\Rtotal}$ MESs in the computational basis and record the binary representation of the outcomes in $S_1,\ldots,S_{4\Rtotal}$ \;
	\tcp*[f]{\textbf{$4\Rtotal$ seeds of length $s$ for the hash function $h$}}\\
	
	Measure the remaining $\Theta\br{n\sqrt{\epsilon}}$ MESs in the computational basis and record the binary representation of the outcomes in $R'$ \; 
	
	Extend $R'$ to a $\delta$-biased pseudo-random string $R=R_1,\ldots,R_{10\Rtotal}$ using the deterministic algorithm of Lemma~\ref{lem:stretch} where $\delta= 2^{-\Theta\br{\frac{n}{r}}}$ \;  
	\tcp*[f]{\textbf{$10\Rtotal$ seeds of length $4rt$ used in quantum hashing}}\\

	\Return {\textup{\textbf{\textsf{Q-Initialization}}}}\;
	\caption{\textbf{\textsf{Q-Initialization}} (Alice's side) }
\end{algorithm}\label{Qalgo:Initialization}


\begin{algorithm}	
	\Input{$\br{\RA,\IndexA,\ellNextMESA,i}$}
	
	\Output{$\br{\ellNextMESA,\NextMESIndexA}$}
	\BlankLine
	
	\If (\tcp*[f]{\textbf{No recycling in first $\LQVC$ iterations}})  {$i \leq \LQVC$}
		{  $\NextMESIndexA \leftarrow i$ \;
		}
	\Else 
		{   $\ellNextMESA \leftarrow \ellNextMESA+1$ \;
			\While { $\br{\ellNextMESA < i} \; \mathrm{and} \; \br{\RA\Br{\ellNextMESA} = \sM}$ }
				{$\ellNextMESA \leftarrow \ellNextMESA+1$ \;}
		   \If  {$\ellNextMESA=i$}
				{  $\NextMESIndexA \leftarrow \perp$ \;
				}
		   \Else 
				{  $\NextMESIndexA \leftarrow \IndexA\Br{\ellNextMESA}$ \;
				}
		}
	
	\Return{\textup{\textbf{\textsf{Recycle}}}}\;
	\caption{\textbf{\textsf{Recycle}} (Alice's side) }
	\label{algo:Recycle}
\end{algorithm}


\begin{algorithm}
	\Input{$\br{\RA,\IndexA,i,R}$}

	\Output{$\QHA$}
	\BlankLine
	
	$\QHA \leftarrow \emptyset$ \;
	\For {$k=1 \to 10$} 
		{	
			Choose a fresh MES from "Quantum Hash" category, and let $F$ denote the register containing Alice's half of the state\;
			\For (\tcp*[f]{\textbf{Hashing between blocks $i-t+1$ and $\ellQVCA$}}) {$j=1 \to 4r \br{t+\ellQVCA-i}$}
				{   \If    {$R_{10i+k}\Br{j}=1 \;\mathrm{and}\; 
						    \RA\Br{i-t+\lceil\frac{j}{4r}\rceil}\neq \sM$}
						   {   
						   		
							    Apply $\br{\control{\rX}}_{FA_b}$ , where 
							    $b=4r \cdot \IndexA \Br{ i-t+\lfloor \frac{j}{4r} \rfloor }+\br{ j\,\mod 4r }$\;
			   		   	   }
				}
    		Apply the Fourier transform operator on $F$, measure it in the computational basis and record the outcome in $qh$\; 
    		$\QHA \leftarrow \br{\QHA,qh}$\;
    	}
    
    \Return{\textup{\textbf{\textsf{Quantum-hash}}}}\;
	\caption{\textbf{\textsf{Quantum-hash}} (Alice's side) }
	\label{algo:Quantum-hash}
\end{algorithm}


\begin{algorithm}
	
	\Input{$\br{\IndexA,\ellQVCA,\msg',\FullPA}$}
	
	\Promise{$\br{\HMA,\HMBtilde,\ellMA,\ellMBtilde}=\br{\HMAtilde',\HMB',\ellMAtilde'+1,\ellMB'+1}$, $\ellMA=\ellMBtilde=i \;,\; \ellQVCA < \ellQVCBtilde$}
	
	\Output{$\br{\ellQVCA,\NewPauliA,\FullPA}$}
	\BlankLine
	
	Let~$C_0$ be the communication register at the beginning of the current iteration (which is in Alice's possession) and for every~$j\in[r]$, let~$C_j$ denote the communication register containing $\msg^\prime(j)$, Bob's $j$-th message in this iteration\;
	
	$\ellQVCA \leftarrow \ellQVCA+1$\;
	
	Let~$E_1 E_2 \dotsb E_{4r}$ be the registers with Alice containing halves of the~$4r$ MESs in the block indexed by $\IndexA \Br{\ellQVCA}$ \;

	Apply $$\br{\control{\rZ}}_{E_2C_0}\br{\control{\rX}}_{E_1C_0}$$ 
	
	For every~$j\in[r-1]$, upon receiving $C_j$ apply $$\br{\control{\rZ}}_{E_{4j+2}C_j}\br{\control{\rX}}_{E_{4j+1}C_j}\br{\control{\rX}^{-1}}_{E_{4j-1}C_j}\br{\control{\rZ}^{-1}}_{E_{4j}C_j}$$ 
	
	Upon receiving $C_r$ apply $$\br{\control{\rX}^{-1}}_{E_{4r-1}C_{r}}\br{\control{\rZ}^{-1}}_{E_{4r}C_{r}}$$

	\tcp*[f]{\textbf{See Section~\ref{subsec:out-of-sync QVC} for the rationale and Bob's analogue of above steps}} \\	
	
	Apply the Fourier transform operator to~$E_1,E_2, \dotsb, E_{4r}$ and measure them in the computational basis. Store the measurement outcomes in $\br{m_1,m_2}\in{\Sigma}^{4r}$\;
	$\RA\Br{\ellQVCA}\leftarrow \sM$\;
	$\NewPauliA \leftarrow \br{m_1,m_2,\perp^{2r}}$\;
	$\FullPA \leftarrow \br{\FullPA,\NewPauliA}$\;
	
	\Return{\textup{\textbf{\textsf{Q-syncMES}}}}\;
	\caption{\textbf{\textsf{Q-syncMES}} (Alice's side)}
	\label{algo:Q-syncMES}
\end{algorithm}


\begin{algorithm}
	
	\Input{$\br{\NewMetaA,\RA,\ellRA,\IndexA,\FullPA}$}
	
	\Promise{$\br{\HMA,\HMBtilde,\ellMA,\ellMBtilde,\ellQVCA}=
		     \br{\HMAtilde',\HMB',\ellMAtilde'+1,\ellMB'+1,\ellQVCBtilde}$ , 
		     $\ellMA=\ellMBtilde=i$ , 
		     $\br{\ellRA,\RA\Br{1:\ellRA}} \neq \br{\ellRBtilde,\RBtilde\Br{1:\ellRBtilde}}$
	        }
	
	\Output{$\br{\NewPauliA,\FullPA,\RA,\ellRA}$}
	\BlankLine
	
	\If   {$\NewMetaA=\sM \; \mathrm{and} \; \RA\Br{\ellRA} = \sS$}
	    {
	      Sequentially apply the Fourier transform operator to all the MESs in the block indexed by $\IndexA \Br{\ellRA}$ and measure them in the computational basis \;
	      Store the measurement outcomes in $\br{m_1,m_2}\in{\Sigma}^{4r}$\;
	      $\NewPauliA \leftarrow \br{m_1,m_2,\perp^{2r}}$\;
	      $\FullPA \leftarrow \br{\FullPA,\NewPauliA}$\;
	      $\RA \Br{\ellRA} \leftarrow \sM$\;
	    }
	$\ellRA \leftarrow \ellRA-1$\;

	\Return{\textup{\textbf{\textsf{rewindRD}}}}\;
	\caption{\textbf{\textsf{rewindRD}} (Alice's side) }
	\label{algo:rewindRD}
\end{algorithm}


\begin{algorithm}
	
	\Input{$\br{\RA,\ellRA,\IndexA,\FullPA}$}
	
	\Promise{$\br{\HMA,\HMBtilde,\ellMA,\ellMBtilde,\ellQVCA}
		     =\br{\HMAtilde',\HMB',\ellMAtilde'+1,\ellMB'+1,\ellQVCBtilde}$ , 
		     $\ellMA=\ellMBtilde=i$ , $\br{\ellRA,\RA\Br{1:\ellRA}}=\br{\ellRBtilde,\RBtilde\Br{1:\ellRBtilde}}$ , 
		     $\QHA \neq \QHB'$
        	}
	
	\Output{$\br{\NewPauliA,\FullPA,\RA,\ellRA}$}
	\BlankLine
	
	\If {$\RA\Br{\ellRA} = \sS$}
	{
		Sequentially apply the Fourier transform operator to all the MESs in the block indexed by $\IndexA \Br{\ellRA}$ and measure them in the computational basis \;
		Store the measurement outcomes in $\br{m_1,m_2}\in{\Sigma}^{4r}$\;
		$\NewPauliA \leftarrow \br{m_1,m_2,\perp^{2r}}$\;
		$\FullPA \leftarrow \br{\FullPA,\NewPauliA}$\;
		$\RA \Br{\ellRA} \leftarrow \sM$\;
	}	
		$\ellRA \leftarrow \ellRA-1$
	
	\Return{\textup{\textbf{\textsf{measuresyndrome}}}}\;
	\caption{\textbf{\textsf{measuresyndrome}} (Alice's side) }
	\label{algo:QmeasureEPR}
\end{algorithm}


\begin{algorithm}
	
	\Input{$\ellRA$}
	
	\Promise{$\br{\HMA,\HMBtilde,\ellMA,\ellMBtilde,\ellQVCA}
	     	 =\br{\HMAtilde',\HMB',\ellMAtilde'+1,\ellMB'+1,\ellQVCBtilde}$ , 
		     $\ellMA=\ellMBtilde=i$ , $\br{\ellRA,\RA\Br{1:\ellRA}}=\br{\ellRBtilde,\RBtilde\Br{1:\ellRBtilde}}$ , 
		     $\QHA=\QHB'$ , $\ellRA < \ellQVCA$ 
	        }
        
	\Output{$\ellRA$}
	\BlankLine
	
	$\ellRA \leftarrow \ellRA+1$\;
	
	\Return{\textup{\textbf{\textsf{extendRD}}}}\;
	\caption{\textbf{\textsf{extendRD}} (Alice's side)}
	\label{algo:extendRD}
\end{algorithm}


\begin{algorithm}
		
	\Input{$\br{\HPA,\ellPA,\HPBtilde,\ellPBtilde,\HPAtilde',\ellPAtilde',\HPB',\ellPB'}$}
	
	\Promise{$\br{\HMA,\HMBtilde,\ellMA,\ellMBtilde,\ellQVCA}
	     	 =\br{\HMAtilde',\HMB',\ellMAtilde'+1,\ellMB'+1,\ellQVCBtilde}$ , 
		     $\ellMA=\ellMBtilde=i$ , $\br{\ellRA,\RA\Br{1:\ellRA}}=\br{\ellRBtilde,\RBtilde\Br{1:\ellRBtilde}}$ , 
		     $\QHA=\QHB'$ , $\ellRA = \ellQVCA$ , 
             $\br{\HPA,\HPBtilde,\ellPA,\ellPBtilde}\neq \br{\HPAtilde',\HPB',\ellPAtilde',\ellPB'}$.
             }
	
	\Output{$\br{\ellPA,\ellPBtilde}$}
	
	\If   {$\ellPA \neq \ellPAtilde' \;\mathrm{or}\; \ellPBtilde \neq \ellPB'$}
	{   \If   {$\ellPA > \ellPAtilde'$}
		{$\ellPA \leftarrow \ellPA-r$\;}		
		\If   {$\ellPBtilde > \ellPB'$}
		{$\ellPBtilde \leftarrow \ellPBtilde-r$\;}		
	}
	\Else   {
		\If   {$\HPA \neq \HPAtilde'$}
		{$\ellPA \leftarrow \ellPA-r$\;}
		\If   {$\HPBtilde \neq \HPB'$}
		{$\ellPBtilde \leftarrow \ellPBtilde-r$\;}		
	}	
	
	\Return{\textup{\textbf{\textsf{Q-rewindPD}}}}\;
	\caption{\textbf{\textsf{Q-rewindPD}} (Alice's side)}
	\label{algo:Q-rewindPD}
\end{algorithm}


\begin{algorithm}
	
	\Input{$\br{\ellPA,\ellPBtilde,\PBtilde,q_{\MA},q_{\MBtilde},\msg'}$}
	
	\Promise{$\br{\HMA,\HMBtilde,\ellMA,\ellMBtilde,\ellQVCA}
	     	 =\br{\HMAtilde',\HMB',\ellMAtilde'+1,\ellMB'+1,\ellQVCBtilde}$ , 
		     $\ellMA=\ellMBtilde=i$ , $\br{\ellRA,\RA\Br{1:\ellRA}}=\br{\ellRBtilde,\RBtilde\Br{1:\ellRBtilde}}$ , 
		     $\QHA=\QHB'$ , $\ellRA = \ellQVCA$ ,  $\br{\HPA,\HPBtilde,\ellPA,\ellPBtilde}= \br{\HPAtilde',\HPB',\ellPAtilde',\ellPB'}$ , 
		     $\ellPA < 6q_{\MA} \cdot r \quad \mathrm{or} \quad \ellPBtilde < 6q_{\MBtilde} \cdot r$.}
	
	\Output{$\br{\ellPA,\PBtilde,\ellPBtilde}$}
	
	\If   {$\ellPA < 6q_{\MA} \cdot r$}
  	    {$\ellPA \leftarrow \ellPA+r$\;}
	\If   {$\ellPBtilde < 6q_{\MBtilde} \cdot r$}
      	{   $\PBtilde\Br{\ellPBtilde+1:\ellPBtilde+r} \leftarrow \msg'$\;
		$\ellPBtilde \leftarrow \ellPBtilde+r$\;
	    }
	
	\Return{\textup{\textbf{\textsf{Q-extendPD}}}}\;
	\caption{\textbf{\textsf{Q-extendPD}} (Alice's side) }
	\label{algo:Q-extendPD}
\end{algorithm}


\begin{algorithm}
	
	\Input{$\br{\FullMA,\MBtilde,\RA,\FullPA,\PBtilde}$}
	
	\Promise{$\br{\HMA,\HMBtilde,\ellMA,\ellMBtilde,\ellQVCA}
	     	 =\br{\HMAtilde',\HMB',\ellMAtilde',\ellMB',\ellQVCBtilde}$ , 
		     $\ellMA=\ellMBtilde=i-1$ , $\br{\ellRA,\RA\Br{1:\ellRA}}=\br{\ellRBtilde,\RBtilde\Br{1:\ellRBtilde}}$ , 
		     $\QHA=\QHB'$ , $\ellRA=\ellQVCA$ ,
		     $\br{\HPA,\HPBtilde,\ellPA,\ellPBtilde}= \br{\HPAtilde',\HPB',\ellPAtilde',\ellPB'}$,
		     $\ellPA=6q_{\MA} \cdot r\;, \ellPBtilde =6q_{\MBtilde} \cdot r$,
		    }
	
	\Output{$\br{\JSoneA, \JStwoA,\NewMetaA,\NewMetaBtilde,
		    \mathit{Block},\RewindExtend,\PCorr,\PCorrtilde}$}
	\BlankLine
	
	Compute $\JSoneA$\;
	Compute $\JStwoA$\;
	Compute $\NewMetaA$\;
	Compute $\NewMetaBtilde$\;
    Compute	$\RewindExtend$\;
	Compute $\mathit{Block}$\;
	Compute $\PCorr$\;
	Compute $\PCorrtilde$\;
	\tcp*[f]{\textbf{Refer to Sections~\ref{subsec:Q-1st-rep-sync},~\ref{subsec:Q-2nd-rep-sync} to see how these variables are computed}}
	
	\Return{\textup{\textbf{\textsf{Q-computejointstate}}}}\;
	\caption{\textbf{\textsf{Q-computejointstate}} (Alice's side)}
	\label{algo:Q-computejointstate}
\end{algorithm}


\begin{algorithm}
	
	\Input{$\br{\ellQVCA,\IndexA,\FullPA,\ellPA,\PBtilde,\ellPBtilde,
		   \NewMetaA,\RewindExtend,\mathit{Block},\PCorr,\PCorrtilde}$
	      }
	
	\Promise{$\br{\HMA,\HMBtilde,\ellMA,\ellMBtilde,\ellQVCA}
		     =\br{\HMAtilde',\HMB',\ellMAtilde'+1,\ellMB'+1,\ellQVCBtilde}$ , 
		     $\ellMA=\ellMBtilde=i$ , $\br{\ellRA,\RA\Br{1:\ellRA}}=\br{\ellRBtilde,\RBtilde\Br{1:\ellRBtilde}}$ , 
		     $\QHA=\QHB'$ , $\ellRA=\ellQVCA$ ,
		     $\br{\HPA,\HPBtilde,\ellPA,\ellPBtilde}= \br{\HPAtilde',\HPB',\ellPAtilde',\ellPB'}$,
		     $\ellPA=6q_{\MA} \cdot r\;, \ellPBtilde =6q_{\MBtilde} \cdot r$,
	        }
	
	\Output{$\br{\ellQVCA,\NewPauliA,\FullPA,\ellPA,
			\NewPauliBtilde,\PBtilde,\ellPBtilde,\ellRA}$}
    \BlankLine
	
	$\ellQVCA \leftarrow \ellQVCA+1$\;
	Continue the simulation of the noiseless protocol according to the output of \textbf{\textsf{Q-computejointstate}} using the block of MESs indexed by $\IndexA\Br{\ellQVCA}$\;
	$\NewPauliA \leftarrow \br{\bot^{2r},\bot^{2r},\PCorr}$\;
	$\FullPA \leftarrow \br{\FullPA,\NewPauliA}$\;
	$\ellPA \leftarrow \ellPA+6r$\;
	$\NewPauliBtilde\leftarrow \br{\bot^{2r},\bot^{2r},\PCorrtilde}$\;
	$\PBtilde \leftarrow \br{\PBtilde,\NewPauliBtilde}$\;
	$\ellPBtilde \leftarrow \ellPBtilde+6r$\;
	$\ellRA \leftarrow \ellRA+1$\;	
	
	\Return{\textup{\textbf{\textsf{Q-simulate}}}}\;
	\caption{\textbf{\textsf{Q-simulate}} (Alice's side)}
	\label{algo:Q-simulate}
\end{algorithm}

\newpage
\subsection{Analysis}

To simplify the analysis of the algorithm, without loss of generality, we assume that the error introduced by the adversary on the $n^\prime$ message registers in~$\Pi'$ is a Pauli error of weight at most~$\epsilon n^\prime$. We prove the correctness of the algorithm for any such error syndrome, which by linearity implies the correctness of the algorithm against any adversary defined in Section~\ref{sec:noisy_comm_model}. In order to track the simulation progress and show the correctness of the algorithm, we condition on some view of the local classical data recorded by Alice and Bob.

Similar to Section~\ref{subsec:polysizeclassicalanalysis}, the analysis of Algorithm~\ref{algo:MainalgorithmQMessage} is in terms of potential functions which measure the correctness of the two players' views of what has happened so far in the simulation and quantify the progress in reproducing the joint state of the input protocol. We recall the following definitions from Section~\ref{subsec:polysizeclassicalanalysis}:
\begin{align}
	&\mdAplus \defeq~\text{the length of the longest prefix where $\MA$ and $\MAtilde$ agree;}\label{eqn:Q-mda+}\\
	&\mdBplus \defeq~\text{the length of the longest prefix where $\MB$ and $\MBtilde$ agree;}\label{eqn:Q-mdb+}\\
	&\mdAminus\defeq \max\{\ellMA,\ellMAtilde\}-\mdAplus;\label{eqn:Q-mda-}\\
	&\mdBminus\defeq \max\{\ellMB,\ellMBtilde\}-\mdBplus;\label{eqn:Q-mdb-}\\
	&\pdAplus\defeq \lfloor\frac{1}{r} \times~\text{the length of the longest prefix where $\PA$ and $\PAtilde$ agree}\rfloor;\label{eqn:Q-pda+}\\
	&\pdBplus\defeq \lfloor\frac{1}{r} \times~\text{the length of the longest prefix where $\PB$ and $\PBtilde$ agree}\rfloor;\label{eqn:Q-pdb+}\\
	&\pdAminus\defeq \frac{1}{r} \max\{\ellPA,\ellPAtilde\}-\pdAplus;\label{eqn:Q-pda-}\\
	&\pdBminus\defeq \frac{1}{r} \max\{\ellPB,\ellPBtilde\}-\pdBplus.\label{eqn:Q-pdb-}
\end{align}
We recall the definition of $g,b,u$ from Subsection~\ref{subsec:Q-2nd-rep-out-of-sync}:
\begin{align}
	&g \defeq
	\text{the number of good blocks in $\JStwo$,}\label{eqn:Q-g}\\
	&b \defeq
	\text{the number of bad blocks in $\JStwo$, and}\label{eqn:Q-b}\\
	&u\defeq |\ellQVCA-\ellQVCB|,\label{eqn:Q-u}
\end{align}
We define
\begin{align}
	rd^+ \defeq \max &\{j:\; j\leq \min\{\ellRA,\ellRB\} \;,\; \RA\Br{1:j}=\RB\Br{1:j} \nonumber \\
	&\;, W_{k}=0^{4r}\text{ for all $k\leq j$ with $\RA\Br{k}=\sS$}\}  ;\label{eqn:rd+}\\
	rd^- \defeq \max &\{\ellRA,\ellRB\}-rd^+,\label{eqn:rd-}
\end{align}
where $W$ in Equation~\ref{eqn:rd+} is the string corresponding to the error syndrome defined in Subsection~\ref{subsec:Q-1st-rep-sync}. At the end of the $i$-th iteration, we let
\begin{align}
&\Phi_{\MD}\defeq 2i-\mdAplus+3\mdAminus-\mdBplus+3\mdBminus,\label{eqn:Q-phimd}\\
&\Phi_{\RD}\defeq \ellQVCA+\ellQVCB+ 13rd^- - 2rd^+, \label{eqn:Q-phird}\\
&\Phi_{\PD}\defeq 6q_{\MA}+6q_{\MB}-\pdAplus+\pdAminus-\pdBplus+\pdBminus,\label{eqn:Q-phipd}\\
&\Phi_{\mathrm{Q}}\defeq g-b-9u, \label{eqn:Q-phiQ}\\
&\Phi\defeq\Phi_{\mathrm{Q}}-\Phi_{\MD}-\Phi_{\RD}-\Phi_{\PD}.\label{eqn:Q-phi}
\end{align}

The following lemma states an important property of potential functions $\Phi_{\MD}$, $\Phi_{\RD}$ and $\Phi_{\PD}$ defined above which we use in the analysis of the algorithm.

\begin{lemma} \label{lem:potential-value-vs-data}
	Throughout the algorithm, it holds that 
	\begin{itemize}
		\item $\Phi_{\MD}\geq 0$ with equality if and only if Alice and Bob have full knowledge of each other's metadata, i.e., $\mdAplus=\mdBplus=i$ and $\mdAminus=\mdBminus=0$.
		\item $\Phi_{\RD}\geq 0$ with equality if and only if Alice and Bob have used the same number of MES blocks $(\ellQVCA=\ellQVCB)$, their measurement pointers $\ellRA$ and $\ellRB$ agree and are equal to $\ellQVCA$, they fully agree on the recycling data $(\RA=\RB)$ and $W_{k}=0^{4r}$ for all $k\leq \ellQVCA$ with $\RA\Br{k}=\sS$, i.e., $\ellQVCA=\ellQVCB=rd^+$ and $rd^-=0$.
		\item $\Phi_{\PD}\geq 0$ with equality if and only if Alice and Bob have full knowledge of each other's Pauli data, i.e., $\pdAplus=6q_\MA$, $\pdBplus=6q_\MB$ and $\pdAminus=\pdBminus=0$. 
	\end{itemize}
\end{lemma}

\begin{proof}
	The first statement follows from the property that $\mdAminus,\mdBminus \geq 0$ and $\mdAplus,\mdBplus\leq i$. The second statement holds since $rd^- \geq 0$, $rd^+ \leq \min\{\ellRA,\ellRB\}$ and the property that $\ellRA\leq\ellQVCA$ and $\ellRB\leq\ellQVCB$. The third statement follows since $\pdAminus,\pdBminus \geq 0$, $\pdAplus\leq 6q_\MA$ , and $\pdBplus\leq 6q_\MB$.  
\end{proof}
 
In order to avoid ambiguity, whenever necessary we use a superscript $i$ to indicate the value of the variables of the algorithm at the end of the $i$-th iteration. For instance, we denote Alice's recycling data at the end of the $i$-th iteration by $\RA^i$. Before presenting the analysis of Algorithm~\ref{algo:MainalgorithmQMessage}, we formally define successful recycling.

\begin{definition} \label{def:successful recycling}
	We say recycling is successful in the $i$-th iteration of Algorithm~\ref{algo:MainalgorithmQMessage} if the following hold: 
	\begin{itemize}
		\item The algorithm does not abort in the $i$-th iteration, i.e., $\NextMESIndexA^i,\NextMESIndexB^i \neq \perp$,
		\item $\NextMESIndexA^i=\NextMESIndexB^i$,
		\item The block of MES registers indexed by $\NextMESIndexA^i$ are in the $\ket{\phi^{0,0}}^{\otimes 4r}$ state at the beginning of the $i$-th iteration.
	\end{itemize}
\end{definition}

Note that the conditions of Definition~\ref{def:successful recycling} are all satisfied in the first $\LQVC$ iterations of the algorithm, since in fact no recycling is done in those iterations. Moreover, we have $\IndexA\Br{1:\LQVC}=\IndexB\Br{1:\LQVC}=1:\LQVC$.

 {\bf Proof Outline of Theorem~\ref{thm:Qmessagelargealphabet}.} In order to prove successful simulation of an $n$-round protocol, it suffices to show that $\Phi \geq n/{2r}$, at the end of the simulation. In Section~\ref{subsec:polysizeclassicalanalysis} we showed that except with exponentially small probability, the total number of hash collisions is~$O\br{n\epsilon}$. Then, for sufficiently large number of iterations, to prove the correctness it was sufficient to show that in any iteration with no error or hash collision the potential function increases by at least one, while any iteration with errors or hash collisions decreases the potential by at most some fixed constant. However, this statement is not necessarily true for Algorithm~\ref{algo:MainalgorithmQMessage} if the recycling of MESs has not been successful in an earlier iteration. In fact, the potential function is defined in terms of~$\JStwo$ which is a valid representation of the joint state at any stage in the simulation only if recycling has been successful so far. Therefore, to use such an argument, one needs to prove successful recycling first.
 On the other hand, to prove successful recycling in an iteration, we need to bound the number of iterations with a hash collision, as well as the number of iterations dedicated to ``recovery'' from hash collisions and transmission errors. Therefore, the analysis of the recycling-based protocol involves an inductive argument. 
 
 The analysis in this section involves constants 
 \begin{align*}
     c_1 < c_2 < c_3 < c_4 < c_5 < c_6 < c_7 < c_8 < c_9 \enspace,
 \end{align*}
 chosen such that $c_i$ is sufficiently large depending only on $c_j$ with $j<i$.
 
 \begin{definition}
    We say an iteration of Algorithm~\ref{algo:MainalgorithmQMessage} suffers from a \emph{metadata hash collision\/} when $\HMA=\HMAtilde$ despite the fact that $\MA \neq \MAtilde$, or $\HMB=\HMBtilde$ despite the fact that $\MB \neq \MBtilde$.
    Note that we distinguish between the above scenario and when, for instance,~$\HMA=\HMAtilde'$ due to a transmission error on $\HMAtilde$, despite the two might have similar effects.
 \end{definition} 
 
The following lemma bounds the number of iterations with a metadata hash collision up to any point in the simulation assuming successful recycling in the earlier iterations.

\begin{lemma}\label{lem:MD collisions}
	Suppose that recycling is successful in the first $i$ iterations of Algorithm~\ref{algo:MainalgorithmQMessage}. Then the number of iterations of the algorithm suffering from a metadata hash collision in the first $i$ iterations is at most $c_1n\epsilon$ with probability at least $1-2^{-\Theta(n\epsilon)}$.
\end{lemma}

\begin{proof}	
	Note that a metadata hash collision occurs in an iteration only if $\mdAminus + \mdBminus \neq 0$ at the beginning of the iteration. Let $\alpha_\MD$ denote the number of such iterations in the first $i$ iterations of Algorithm~\ref{algo:MainalgorithmQMessage}. It suffices to prove that 
	\[\Pr\br{\alpha_\MD > c_1n\epsilon} \leq 2^{-\Theta(n\epsilon)} \enspace.
	\]
	Note that in any iteration $\mdAminus + \mdBminus$ increases by at most $6$. Moreover, in an iteration with $\mdAminus + \mdBminus \neq 0$, if $\mdAminus + \mdBminus$ decreases, it decreases by at least $1$. Therefore, in at least $\alpha_\MD/7$ iterations, $\mdAminus + \mdBminus$ increases or remains unchanged at a nonzero value. Note that $\mdAminus + \mdBminus > 0$ increases or remains unchanged only if a transmission error or a metadata hash collision occurs. Moreover, when $\mdAminus + \mdBminus$ increases from zero in an iteration, it is due to a transmission error. The number of iterations is less than $2n$. So the total number of iterations with transmission errors, is at most $2n\epsilon$. This implies that in all the remaining iterations, i.e., at least $\alpha_\MD/7-2n\epsilon$ iterations a metadata hash collision occurs. Since the algorithm uses independent seeds in each iteration and the probability of collision is chosen to be $0.1$, the expected number of collisions is at most $\alpha_\MD/10$. If $\alpha_\MD > c_1n\epsilon$ for a sufficiently large $c_1$, then the Chernoff bound implies that the probability of having so many collisions is at most $2^{-\Theta\br{n\epsilon}}$.	
\end{proof}

\begin{definition}
    We refer to an iteration of Algorithm~\ref{algo:MainalgorithmQMessage} as a \emph{recovery iteration of type I\/} if at least one of Alice or Bob conducts one of the cases i.A, i.B, ii.A, or ii.B.
\end{definition}

We use the following lemma to bound the number of type I recovery iterations.

\begin{lemma} \label{lem:type I recovery}
	Suppose that in the first $i$ iterations of Algorithm~\ref{algo:MainalgorithmQMessage}, recycling is successful  and the number of iterations suffering from a metadata hash collision is at most $c_1n\epsilon$. Then the number of type I recovery iterations in the first $i$ iterations is at most $c_2n\epsilon$.
\end{lemma}

\begin{proof}
	Let 
	\[ \Phi_\mathrm{I} \defeq u+\Phi_{\MD}\enspace. \]
	By Lemma~\ref{lem:potential-value-vs-data} and the definition of $u$ in Eq.~\eqref{eqn:Q-u}, $\Phi_\mathrm{I}$ is always non-negative and is equal to zero if and only if Alice and Bob know each other's full metadata and have used the same number of MES blocks for QVC. 
	
	Note that if $\Phi_\mathrm{I}=0$ at the beginning of an iteration, then the iteration is a type I recovery iteration only if a transmission error in communication of metadata messages occurs. The total number of such iterations is at most $2n\epsilon$. 
	
	Let $\beta_\mathrm{I}$ denote the number of iterations in the first $i$ iterations starting with $\Phi_\mathrm{I}> 0$. Note that in any iteration, $\Phi_\mathrm{I}$ increases or remains unchanged at a nonzero value only if a metadata hash collision or a transmission error occurs. In each iteration, regardless of the number of errors and collisions, $\Phi_\mathrm{I}$ increases by at most $23$. Moreover, if $\Phi_\mathrm{I}$ decreases, it decreases by at least $1$. Assuming the number of metadata hash collisions is at most $c_1n\epsilon$, this implies that the number of iterations in which $\Phi_\mathrm{I}$ decreases is at most $23\br{c_1+2} n\epsilon$. So we have $\beta_\mathrm{I} \leq 24\br{c_1+2} n\epsilon$.
	 
	Therefore, the total number of type I recovery iterations is at most $c_2n\epsilon$, where~$c_2 \defeq 24\br{c_1+2}+2$.  	
\end{proof}

\begin{definition}
    We say an iteration of Algorithm~\ref{algo:MainalgorithmQMessage} suffers from a \emph{quantum hash collision\/} when recycling has been successful so far, Alice and Bob know each other's metadata, have used the same  number of MES blocks ($\ellQVCA=\ellQVCB$) and agree on their measurement pointers and their recycling data up to the measurement pointers ($\ellRA=\ellRB$ and ~$\RA\Br{1:\ellRA}=\RB\Br{1:\ellRB}$) but despite the fact that there is an undetected quantum error from earlier iterations, their quantum hash values match, i.e.,~$\QHA=\QHB$. Note that we distinguish between the above scenario and when~$\QHA=\QHB'$ due to a transmission error on $\QHB$.
\end{definition}   
  
We use the following lemma to bound the number of iterations suffering from a quantum hash collision.
  
\begin{lemma} \label{lem:QH collisions}
	Suppose that in the first $i$ iterations of Algorithm~\ref{algo:MainalgorithmQMessage} recycling is successful and the number of iterations suffering from a metadata hash collision is at most $c_1n\epsilon$. Then the number of iterations suffering from a quantum hash collision in the first $i$ iterations is at most $c_3n\epsilon$ with probability at least $1-2^{-\Theta(n\epsilon)}$.
\end{lemma}

\begin{proof}	
	Note that a quantum hash collision occurs in an iteration only if $rd^- \neq 0$ at the beginning of the iteration. Let $\alpha_\RD$ denote the number of such iterations in the first $i$ iterations of Algorithm~\ref{algo:MainalgorithmQMessage}. It suffices to prove that 
	\[\Pr\br{\alpha_\RD > c_3n\epsilon} \leq 2^{-\Theta(n\epsilon)} \enspace.
	\]
	Note that in any iteration $rd^-$ increases by at most $2$. Moreover, in an iteration with $rd^- \neq 0$, if $rd^-$ decreases, it decreases by at least $1$. Therefore, in at least $\alpha_\RD/3$ iterations, $rd^-$ increases or remains unchanged at a nonzero value. Note that $rd^->0$ increases or remains unchanged only if
	\begin{itemize}
		\item A metadata hash collision or a transmission error on metadata messages (i.e., $\HMA$, $\ellMA$, $\HMB$, $\ellMB$, $\HMBtilde$, $\ellMBtilde$, $\HMAtilde$, $\ellMAtilde$)  occurs, or else,
 		\item The iteration is a type I recovery iteration. Alice and Bob are still reconciling an earlier inconsistency in their metadata and they are both in case i.A or case i.B, or one of them is in case ii.A and the other one in case ii.B. Else,
		\item A transmission error on quantum hash values or a quantum hash collision occurs. At least one party does not realize that $rd^->0$ and conducts one of the cases v, vi.A, vi.B, or vii.
	\end{itemize} 
	Moreover, the value of $rd^-$ increases from zero in an iteration only if
	\begin{itemize}
	    \item A metadata hash collision occurs and Alice and Bob act based on incorrect estimates of each other's recycling data, or else,
	    \item A transmission error on metadata messages occurs, or else,
	    \item A transmission error on quantum hash values occurs and only one party conducts case iv, or else,
	    \item A transmission error on the Pauli data messages (i.e., $\HPA$, $\ellPA$, $\HPB$, $\ellPB$, $\HPBtilde$, $\ellPBtilde$, $\HPAtilde$, $\ellPAtilde$) occurs and one party conducts case vi.A or vi.B while the other is in case vii. Else,
	    \item A transmission error occurs on the communicated QVC messages when both parties conduct case vii.
	\end{itemize}
	Assuming the number of metadata hash collisions is at most $c_1n\epsilon$, by Lemma~\ref{lem:type I recovery}, the total number of type I recovery iterations is at most $c_2n\epsilon$. The total number of transmission errors is at most $2n\epsilon$. Therefore, in at least $\alpha_\RD/3-\br{c_1+c_2+2}n\epsilon$ iterations a quantum hash collision occurs.  
	
	The shared random string used as the classical seed for quantum hashing is $\delta$-biased with $\delta=2^{-\Theta\br{n\sqrt{\epsilon}}}$. By Lemma~\ref{lem:stretch}, the seeds are also $\delta^{\Theta\br{1}}$-statistically close to being $\Theta\br{\Rtotal}$-wise independent. Therefore, all hashing steps are statistically close to being fully independent. Combined with Lemma~\ref{lem:quantum hash-delta biased}, this implies that the expected number of quantum hash collisions is at most $10^{-3}\alpha_\RD$. 
	For sufficiently large $c_3$, if $\alpha_\RD > c_3n\epsilon$, the Chernoff bound implies that the probability of having at least $\alpha_\RD/3-\br{c_1+c_2+2}n\epsilon$ quantum hash collisions is at most $2^{-\Theta(n\epsilon)}$. 
\end{proof}

\begin{definition}
    We refer to an iteration of Algorithm~\ref{algo:MainalgorithmQMessage} as a recovery iteration of type II if it is not a type I recovery iteration and at least one of Alice or Bob conducts one of the cases iii, iv, or v.
\end{definition}

We use the following lemma to bound the number of type II recovery iterations.

\begin{lemma} \label{lem:type II recovery}
	Suppose that in the first $i$ iterations of Algorithm~\ref{algo:MainalgorithmQMessage}, recycling is successful, the number of iterations suffering from a metadata hash collision is at most $c_1n\epsilon$ and the number of iterations suffering from a quantum hash collision is at most $c_3n\epsilon$. Then the total number of type II recovery iterations in the first $i$ iterations is at most $c_4n\epsilon$.
\end{lemma}

\begin{proof}	  
	Note that by Lemma~\ref{lem:potential-value-vs-data}, $\Phi_{\RD}$ is always non-negative. If at the beginning of an iteration $\Phi_\RD=0$, then the iteration is a type II recovery iteration only if
	\begin{itemize}
		\item $\Phi_\MD>0$ but due to a metadata hash collision or a transmission error on metadata messages both Alice and Bob do not realize that. In this case they compute their estimates of each other's recycling data based on incorrect estimates of each other's metadata.
		\item $\Phi_\MD=0$ but a transmission error in communication of quantum hashes occurs. 
	\end{itemize}  
	
	Therefore, in the first $i$ iterations the total number of type II recovery iterations starting with $\Phi_\RD=0$ is at most~$\br{c_1+2}n\epsilon$. \\
	Let $\beta_\RD$ denote the number of iterations starting with $\Phi_\RD > 0$ in the first $i$ iterations. Assuming successful recycling in the preceding iterations, $\Phi_\RD>0$ increases or remains unchanged in an iteration only if
	\begin{itemize}
		\item The iteration is a type I recovery iteration, or else,
		\item $\Phi_\MD>0$ but due to a metadata hash collision or transmission errors on metadata messages, Alice and Bob don't realize that and act based on their incorrect estimates of each other's recycling data.
		\item $\Phi_\MD=0$, i.e., Alice and Bob have correct estimates of each other's recycling data but a quantum hash collision or a transmission error on quantum hash values occurs. 
	\end{itemize} 
	
	Moreover, $\Phi_\RD$ increases from zero in an iteration only if
	\begin{itemize}
	    \item A transmission error occurs, or else, 
	    \item A metadata hash collision occurs and the two parties act based on incorrect estimates of each other's recycling data.
	\end{itemize}
	
	Therefore, the number of iterations in the first $i$ iterations with $\Phi_\RD$ increasing or remaining unchanged at a nonzero value is at most~$\br{c_1+c_2+c_3+2}n\epsilon$. Note that in each iteration, regardless of the number of errors and collisions, $\Phi_\RD$ increases by at most $30$. Moreover, if $\Phi_\RD$ decreases, it decreases by at least $1$. This implies that the number of iterations in which $\Phi_\RD$ decreases is at most $30\br{c_1+c_2+c_3+2}n\epsilon$. So, we have $\beta_\RD \leq 31\br{c_1+c_2+c_3+2}n\epsilon$. \\
	Therefore, the total number of recovery iterations of type II in the first $i$ iterations is at most $c_4n\epsilon$, where $c_4 \defeq 31\br{c_1+c_2+c_3+2}+\br{c_1+2}$.	
\end{proof}

\begin{definition}
    We say an iteration of Algorithm~\ref{algo:MainalgorithmQMessage} suffers from a \emph{Pauli data hash collision\/} when recycling has been successful so far, Alice and Bob know each other's metadata, agree on the number of MES blocks they have used ($\ellQVCA=\ellQVCB$), agree on their recycling data, their measurement pointers satisfy $\ellRA=\ellRB=\ellQVCA$, all the non-measured MES blocks are in the $\ket{\phi^{0,0}}^{\otimes 4r}$ state and $\HPA=\HPAtilde$ despite the fact that $\PA \neq \PAtilde$ or $\HPB=\HPBtilde$ despite the fact that $\PB \neq \PBtilde$. Note that we distinguish between the above scenario and when for instance~$\HPA=\HPAtilde'$ due to a transmission error on $\HPAtilde$.
\end{definition}

We use the following lemma to bound the number of iterations suffering from a Pauli data hash collision.

\begin{lemma} \label{lem:PD collisions}
	Suppose that in the first $i$ iterations of Algorithm~\ref{algo:MainalgorithmQMessage}, recycling is successful, the number of iterations suffering from a metadata hash collision is at most $c_1n\epsilon$ and the number of iterations suffering from a quantum hash collision is at most $c_3n\epsilon$. Then the number of iterations suffering from a Pauli data hash collision in the first $i$ iterations is at most $c_5n\epsilon$ with probability at least $1-2^{-\Theta(n\epsilon)}$.
\end{lemma}

\begin{proof}	
	Note that a Pauli data hash collision occurs in an iteration only if $\pdAminus + \pdBminus \neq 0$ at the beginning of the iteration. Let $\alpha_\PD$ denote the number of such iterations in the first $i$ iterations of Algorithm~\ref{algo:MainalgorithmQMessage}. We prove that 
	\[\Pr\br{\alpha_\PD > c_5n\epsilon} \leq 2^{-\Theta(n\epsilon)} \enspace.
	\]
	Note that in any iteration $\pdAminus + \pdBminus$ increases by at most $8$. Moreover, in an iteration with $\pdAminus + \pdBminus \neq 0$, if $\pdAminus + \pdBminus$ decreases, it decreases by at least $1$. Therefore, in at least $\alpha_\PD/9$ iterations, $\pdAminus + \pdBminus$ increases or remains unchanged at a nonzero value. Note that when $\pdAminus + \pdBminus >0$ increases or remains unchanged in an iteration, it is due to one of the following reasons:
	\begin{itemize}
		\item The iteration is a type I recovery iteration, or else,
		\item The iteration is a type II recovery iteration, or else,
		\item A Pauli data hash collision or a transmission error on Pauli data messages (i.e., $\HPA$, $\ellPA$, $\HPB$, $\ellPB$, $\HPBtilde$, $\ellPBtilde$, $\HPAtilde$, $\ellPAtilde$) occurs.
	\end{itemize} 
	
	Moreover, $\pdAminus + \pdBminus$ increases from zero in an iteration only if
	\begin{itemize}
	    \item A metadata hash collision or transmission error on the metadata messages occurs, or else, 
	    \item A transmission error on the quantum hash values occurs, or else,
	    \item A transmission error on the Pauli data messages occurs, or else,
	    \item Both parties conduct case vi.B and due to a transmission error, at least one of Alice or Bob extends her/his estimate of the other party's Pauli data incorrectly.
	\end{itemize}
	
	Assuming the number of iterations of Algorithm~\ref{algo:MainalgorithmQMessage} suffering from a metadata hash collision is at most $c_1n\epsilon$ and the number of iterations suffering from a quantum hash collision is at most $c_3n\epsilon$, the total number of type I and type II iterations is at most~$\br{c_2+c_4}n\epsilon$.
	The number of transmission errors is at most $2n\epsilon$. Therefore, in at least $\alpha_\PD/9-\br{c_1+c_2+c_4+2}n\epsilon$ iterations a Pauli data hash collision occurs.  

	Since the algorithm uses independent seeds in each iteration and the probability of a collision is chosen to be $0.1$, the expected number of Pauli data hash collisions is at most $\alpha_\PD/10$. For sufficiently large $c_5$, if $\alpha_\PD > c_5n\epsilon$, the Chernoff bound implies that the probability of having so many Pauli data hash collisions in the first $i$ iterations is at most $2^{-\Theta\br{n\epsilon}}$.
\end{proof}

\begin{definition}
    We refer to an iteration of Algorithm~\ref{algo:MainalgorithmQMessage} as a recovery iteration of type III if it is not a type I or type II recovery iteration and at least one of Alice or Bob conducts one of the cases vi.A or vi.B.
\end{definition}

We use the following lemma to bound the number of type III recovery iterations.

\begin{lemma} \label{lem:type III recovery}
	Suppose that in the first $i$ iterations of Algorithm~\ref{algo:MainalgorithmQMessage}, recycling is successful and the number of iterations suffering from metadata , quantum and Pauli data hash collisions is at most $c_1n\epsilon$, $c_3n\epsilon$ and $c_5n\epsilon$, respectively. Then the total number of type III recovery iterations in the first $i$ iterations is at most $c_6n\epsilon$.
\end{lemma}

\begin{proof}
	Note that by Lemma~\ref{lem:potential-value-vs-data}, $\Phi_{\PD}$ is always non-negative and it is equal to zero if and only if Alice and Bob have full knowledge of each other's Pauli data. If at the beginning of an iteration $\Phi_\PD=0$, then the iteration is a type III recovery iteration only if 
	\begin{itemize}
	    \item A transmission error occurs on the Pauli data messages, or else,
	    \item A metadata hash collision or a transmission error on metadata messages occurs. In this case at least one party incorrectly believes that his/her estimate of the other party's Pauli data is not full-length.
	\end{itemize} 
	
	Therefore, in the first $i$ iterations the total number of type II recovery iterations starting with $\Phi_\PD=0$ is at most~$\br{c_1+2}n\epsilon$.
	
	Let $\beta_\PD$ denote the number of iterations starting with $\Phi_\PD > 0$ in the first $i$ iterations. Note that in any iterations $\Phi_\PD>0$ increases or remains unchanged only if
	\begin{itemize}
		\item The iteration is a type I recovery iteration, or else,
		\item The iteration is a type II recovery iteration, or else,
		\item A Pauli data hash collision or a transmission error on Pauli data messages occurs. 
	\end{itemize}
	
	Moreover, $\Phi_\PD$ increases from zero in an iteration only if
		\begin{itemize}
	    \item A metadata hash collision or transmission error on the metadata messages occurs, or else,
	    \item The iteration is a type I recovery iteration in which one party conducts case ii.A and the other case ii.B, or else,
	    \item A transmission error on the quantum hash values occurs, or else,
	    \item The iteration is a type II recovery iteration in which both parties conduct case iii or both conduct case iv, or else,
	    \item A transmission error on the Pauli data messages occurs.
	\end{itemize}
	
	Therefore, the number of iterations in the first $i$ iterations with $\Phi_\PD$ increasing or remaining unchanged at a nonzero value is at most~$\br{c_1+c_2+c_4+c_5+2}n\epsilon$. Note that in each iteration, regardless of the number of errors and collisions, $\Phi_\PD$ increases by at most $22$. Moreover, if $\Phi_\PD$ decreases, it decreases by at least $1$. This implies that the number of iterations in which $\Phi_\PD$ decreases is at most $22\br{c_1+c_2+c_4+c_5+2}n\epsilon$. So, we have $\beta_\PD \leq 23\br{c_1+c_2+c_4+c_5+2}n\epsilon$.
	
	Therefore, the total number of recovery iterations of type III in the first $i$ iterations is at most $c_6n\epsilon$, where $c_6 \defeq 23\br{c_1+c_2+c_4+c_5+2}+\br{c_1+2}$.	
\end{proof}

 Let $\omega_i$ denote the number of iterations suffering from a transmission error or a hash collision in the first $i$ iterations, plus the number of recovery iterations of type I, II, or III, in the first $i$ iterations. Note that the bounds and probabilities in Lemmas~\ref{lem:MD collisions}--\ref{lem:type III recovery} are all independent of the iteration number $i$. As a corollary we have:
 
 \begin{cor} \label{cor:errors-collisions}
 	There exist ~$q=2^{-\Theta(n\epsilon)}$ and a constant $c_7$ such that, for every $i\in \Br{\Rtotal}$, assuming successful recycling in the first $i$ iterations of Algorithm~\ref{algo:MainalgorithmQMessage}, except with probability at most~$q$, we have $\omega_i \leq c_7 n\epsilon$.
 \end{cor}
  
The following lemma is the last ingredient we need to prove successful recycling in every iteration of the simulation. Recall that $\RA^i$ and $\RB^i$ denote the recycling data of Alice and Bob, respectively, at the end of the $i$-th iteration.

\begin{lemma} \label{lem:recycling_requirements}
	Let $t=c_8n\epsilon$ where $c_8>3c_7$. Then for every $i\in \Br{\Rtotal}$ where $i\geq t$, if recycling is successful in the first $i-1$ iterations and $\omega_{i-1} \leq c_7 n\epsilon$, then:
	\begin{enumerate}
		\item $\RA^i\Br{1:i-t}=\RB^i\Br{1:i-t}$, i.e., at the end of iteration $i$, the recycling data of Alice and Bob agree in a prefix of length at least $i-t$.
		\item $\RA^i\Br{1:i-t}=\RA^{i+1}\Br{1:i-t}$, i.e., the prefix of $\RA$ of length $i-t$ does not get modified in the next iteration. The same statement holds for $\RB$.
		\item For every $k\in \Br{i-t}$ such that $\RA^i\Br{k}=\RB^i\Br{k}=\sS$, we have $W_{k}=0^{4r}$.
		\suppress{
		... in $\JStwo$ representation at the end of the $i$-th iteration the $k$-th block of MES registers is in the $\ket{\phi^{0,0}}^{\otimes 4r}$ state.
		} 
	\end{enumerate}
\end{lemma}

\begin{proof}
	
	Part $1$: 
	
	Toward contradiction, suppose that there exists $t'\in\Br{t,i-1}$ such that $\RA^i\Br{i-t'}\neq \RB^i\Br{i-t'}$. Without loss of generality, assume that $\RA^i\Br{i-t'}=\sM$ and $\RB^i\Br{i-t'}=\sS$. Suppose that the last time Alice's measurement pointer $\ellRA$ was equal to $i-t'$ was $t_2$ iterations earlier, i.e., iteration $i-t_2$. In that iteration, $\ellRA$ has distance $t_1\defeq t'-t_2$ from  $i-t_2$, the iteration number. Note that the distance between the iteration number and $\ellRA$ increases only in (some) recovery iterations and it increases by at most $2$: the distance remains the same if Alice is in case v or case vii. Otherwise, it increases by $1$ if $\ellRA$ does not move and increases by $2$ when it moves back. This implies that in the first $i-t_2$ iterations, there have been at least $t_1/2$ recovery iterations. In the $t_2$ iterations after that, there is an inconsistency in the recycling data which does not get resolved. In any of these iterations one of the following holds:
	\begin{itemize}
		\item A metadata hash collision or a transmission error on metadata messages occurs, or else,
		\item The iteration is a type I recovery iteration and Alice and Bob are still trying to reconcile an inconsistency in their metadata, or else, 
		\item The iteration is a type II recovery iteration. In this case, no metadata transmission error or collision occurs and the iteration is not a type I recovery iteration. So Alice and Bob know each other's recycling data and aware of the inconsistency, they are trying to resolve it.  
	\end{itemize}
	Therefore, in the first $i-1$ iterations, the number of recovery iterations plus the number of iterations suffering from a transmission error or a collision is at least
	\[ t_1/2+t_2-1 \geq t'/2-1 \geq t/2-1 \geq \frac{c_8}{2}n\epsilon-1\enspace,
	\] 
	contradicting~$\omega_{i-1} \leq c_7 n\epsilon$. Note that here we implicitly use the reasonable assumption that $n\epsilon$ is at least a constant.
	
	Part $2$:
	
	Suppose that recycling is successful in the first $i-1$ iterations and $\omega_{i-1}\leq c_7n\epsilon$. Note that by the same argument as in part $1$, at the end of iteration $i+1$,  the difference between the measurement pointers and the iteration number is at most $2\omega_{i-1}+4 \leq 2c_7n\epsilon+4 \leq t$. Therefore, the prefixes of both $\RA$ and $\RB$ of length $i-t$ do not get modified in the next iteration. 
	
	Part $3$:
	
	Note that the difference between the iteration number and $\ellQVCA$ increases only in (some) recovery iterations and it increases by at most $1$: The difference remains the same if Alice is in case ii.B or case vii. Otherwise, $\ellQVCA$ remains unchanged and the distance increases by $1$. The number of recovery iterations in the first $i$ iterations is at most $\omega_{i-1}+1 < t$. Therefore, at the end of the $i$-th iteration we have~$\min\{\ellQVCA,\ellQVCB\} > i-t$. 
	Toward contradiction, suppose that there exists $t'\in\Br{t,i-1}$ such that~$\RA^i\Br{i-t'}=\RB^i\Br{i-t'}=\sS$ and $W_{i-t'} \neq 0^{4r}$. \suppress{...in $\JStwo$ representation at the end of the $i$-th iteration the block $i-t'$ of MES registers is not in the $\ket{\phi^{0,0}}^{\otimes 4r}$ state.} This is due to one of the following scenarios:
	\begin{itemize}
		\suppress{
		\item The corresponding block of MES registers has been used only by one party for communication using QVC.
		}
		\item The corresponding block of MES registers has been used by both parties for communication using QVC, but in different iterations (out-of-sync QVC).
		\item It has been used in the same iteration by both parties for communication using QVC but transmission errors have occurred on the messages.
	\end{itemize}
	In any case, suppose that the last time one of the pointers $\ellQVCA$ or $\ellQVCB$ was equal to $i-t'$ was $t_2$ iterations earlier, i.e., iteration $i-t_2$ and without loss of generality, suppose that $\ellQVCA$ is that pointer. In iteration $i-t_2$, the pointer $\ellQVCA$ has distance $t_1\defeq t'-t_2$ from  $i-t_2$, the iteration number.  This implies that in the first $i-t_2$ iterations, there have been at least $t_1$ recovery iterations. In the $t_2$ iterations after that, the block of MES registers indexed by~$\IndexA\Br{i-t'}$ are not measured by any of the parties. In any of these iterations one of the following holds:
	\begin{itemize}
		\item A metadata hash collision or a transmission error on metadata messages occurs, or else,
		\item The iteration is a type I recovery iteration and Alice and Bob are still trying to reconcile an inconsistency in their metadata, or else, 
		\item A quantum hash collision or a transmission error on quantum hash values occurs, or else,
		\item The iteration is a type II recovery iteration. In this case, no metadata transmission error or collision occurs and the iteration is not a type I recovery iteration. So Alice and Bob know each other's recycling data. So Alice and Bob are both be in case iii or both in case iv.
	\end{itemize}
	\suppress{Note that, if at the $i$-th iteration only one party has used the block of MES registers, then in any of the previous $t_2$ iterations exactly one of the first two scenarios occurs.}
	The above argument implies that in the first $i-1$ iterations, the number of recovery iterations plus the number of iterations suffering from a transmission error or a collision is at least
	\[ t_1+t_2-1 = t'-1 \geq t-1 \geq c_8n\epsilon-1\enspace,
	\] 
	contradicting~$\omega_{i-1} \leq c_7 n\epsilon$.
\end{proof}

We are now ready to prove that except with exponentially small probability recycling is successful in every iteration of the algorithm. Recall that we denote Alice's recycling pointer at the end of iteration $i$ by $\ellNextMESA^{i}$. We use $m_i^\mathrm{A}$ to denote the number of $\sM$ symbols in~$\RA^i\Br{1:\ellNextMESA^{i}}$. Similarly, $\ellNextMESB^{i}$ denotes Bob's recycling pointer at the end of iteration $i$ and the number of $\sM$ symbols in~$\RB^i\Br{1:\ellNextMESB^{i}}$ is denoted by $m_i^\mathrm{B}$.

\begin{lemma} \label{lem:successful-recycling}
	Let $\LQVC=c_9n\epsilon$, where $c_9 > c_7+c_8$. Then with probability at least $1-2^{-\Theta(n\epsilon)}$, recycling is successful throughout the execution of Algorithm~\ref{algo:MainalgorithmQMessage}.
\end{lemma}

\begin{proof}
	The proof is based on induction on the iteration number. Note that recycling starts from iteration $\LQVC+1$.
	
	\textbf{Base case} ($i=\LQVC+1$):  Note that the conditions of Definition~\ref{def:successful recycling} are satisfied in the first $\LQVC$ iterations of the algorithm and we have $\IndexA\Br{1:\LQVC}=\IndexB\Br{1:\LQVC}=1:\LQVC$. Therefore, by Corollary~\ref{cor:errors-collisions}, except with probability at most~$q=2^{-\Theta(n\epsilon)}$, we have~$\omega_{\LQVC} \leq c_7 n\epsilon$. Assuming~$\omega_{\LQVC} \leq c_7 n\epsilon$, by Lemma~\ref{lem:recycling_requirements},~$\RA^{\LQVC+1}\Br{1:\LQVC+1-t}=\RB^{\LQVC+1}\Br{1:\LQVC+1-t}$ and for every $k\in \Br{\LQVC+1-t}$ such that $\RA^{\LQVC+1}\Br{k}=\RB^{\LQVC+1}\Br{k}=\sS$, we have $W_{k}=0^{4r}$. \suppress{in $\JStwo$ representation at the end of iteration $\LQVC+1$, the $k$-th block of MES registers is in the $\ket{\phi^{0,0}}^{\otimes 4r}$ state.} Note that the number of $\sM$ symbols in $\RA^{\LQVC+1}\Br{1:\LQVC+1-t}$ and $\RB^{\LQVC+1}\Br{1:\LQVC+1-t}$ is at most the number of type I and type II recovery iterations so far, hence at most $\omega_{\LQVC+1} \leq \omega_{\LQVC}+1 \leq c_7 n\epsilon+1 < \LQVC+1-t$. Therefore, after running the \textsf{Recycle} subroutine, the algorithm does not abort and at the end of iteration $\LQVC+1$, the recycling pointers are equal, i.e., $\ellNextMESA^{\LQVC+1}=\ellNextMESB^{\LQVC+1}$. Together with the fact that $\IndexA\Br{1:\LQVC}=\IndexB\Br{1:\LQVC}$, this implies that the conditions of Definition~\ref{def:successful recycling} are satisfied. Therefore, there exists an event $\mathcal{E}$ of probability at most~$q=2^{-\Theta(n\epsilon)}$ such that if $\lnot\mathcal{E}$ then,
	\begin{itemize}
		\item recycling is successful in iteration $\LQVC+1$,
		\item $\ellNextMESA^{\LQVC+1}=\ellNextMESB^{\LQVC+1}$, and
		\item $\ellNextMESA^{\LQVC+1}=m_{\LQVC+1}^\mathrm{A}+1$ and $\ellNextMESB^{\LQVC+1}=m_{\LQVC+1}^\mathrm{B}+1$.
	\end{itemize}
	
	For $\LQVC < i \leq \Rtotal$, let $\mathcal{T}_i$ be the following statement in terms of the iteration number $i$:
	\begin{itemize}
		\item Recycling is successful in the first $i$ iterations of the algorithm,
		\item $\ellNextMESA^{i} = \ellNextMESB^{i}$, i.e., the recycling pointers are equal at the end of iteration $i$, and
		\item $\ellNextMESA^{i}=m_{i}^\mathrm{A}+i-\LQVC$ and $\ellNextMESB^{i}=m_{i}^\mathrm{B}+i-\LQVC$.
	\end{itemize}

	\textbf{Induction hypothesis:} For $\LQVC < i \leq \Rtotal$, there exists an event $\mathcal{E}_i$ of probability at most~$\br{i-\LQVC}\cdot q$ such that if $\lnot\mathcal{E}_i$ then $\mathcal{T}_i$ holds.
	 
	\textbf{Inductive step:} Assuming $\lnot\mathcal{E}_i$, by Corollary~\ref{cor:errors-collisions}, except with probability at most~$q=2^{-\Theta(n\epsilon)}$, we have~$\omega_{i} \leq c_7 n\epsilon$. Let $\mathcal{E}'$ be the event that~$\omega_{i} > c_7 n\epsilon$. Note that 
	$\Pr \left( \mathcal{E}'| \lnot \mathcal{E}_i\right) \leq q$.
	Suppose further that $\lnot \mathcal{E}'$. Since $\ellNextMESA^{i}=m_{i}^\mathrm{A}+i-\LQVC$, we have
	\[ i-t-\ellNextMESA^{i} = \LQVC-t-m_{i}^\mathrm{A} \geq \LQVC-t-\omega_{i} \geq \Omega(n\epsilon). 
	\]
	Therefore, by induction hypothesis, we also have $i-t-\ellNextMESB^{i} \geq \Omega(n\epsilon)$. By part $1$ of Lemma~\ref{lem:recycling_requirements}, we have $\RA^{i+1}\Br{1:i+1-t}=\RB^{i+1}\Br{1:i+1-t}$.  Since $\omega_{i-1} \leq \omega_{i} \leq c_7 n\epsilon$, by part $2$ of Lemma~\ref{lem:recycling_requirements}, we have $\RA^{i+1}\Br{1:i-t}=\RA^i\Br{1:i-t}$ and $\RB^{i+1}\Br{1:i-t}=\RB^i\Br{1:i-t}$. Therefore, the algorithm does not abort in iteration $i+1$ and we have $\ellNextMESA^{i+1} = \ellNextMESB^{i+1}$. Moreover, $\ellNextMESA^{i+1}=m_{i+1}^\mathrm{A}+(i+1)-\LQVC$ and $\ellNextMESB^{i+1}=m_{i+1}^\mathrm{B}+(i+1)-\LQVC$. Note that since recycling is successful in the first $i$ iterations, at the beginning of iteration $i+1$, we have $\IndexA=\IndexB$. So $\NextMESIndexA^{i+1}=\NextMESIndexB^{i+1}\neq \perp$, i.e., the first and second conditions of Definition~\ref{def:successful recycling} are satisfied for iteration $i+1$.
	
	By part $3$ of Lemma~\ref{lem:recycling_requirements}, for every $k\in \Br{i+1-t}$ such that $\RA^{i+1}\Br{k}=\RB^{i+1}\Br{k}=\sS$, we have $W_{k}=0^{4r}$.\suppress{in $\JStwo$ representation at the end of iteration $i+1$, the $k$-th block of MES registers is in the $\ket{\phi^{0,0}}^{\otimes 4r}$ state.}
	Note that in the strings $\IndexA$ and $\IndexB$, while each index in $\Br{\LQVC}$ may appear several times before the recycling pointers, it can only appear at most once after these pointers. Therefore, the block of MES registers indexed by $\NextMESIndexA^{i+1}$ is indeed in the $\ket{\phi^{0,0}}^{\otimes 4r}$ state when it is recycled in iteration $i+1$ and the third condition of Definition~\ref{def:successful recycling} is also satisfied.
	
	For $\mathcal{E}_{i+1} \defeq \mathcal{E}_i \lor \mathcal{E}'$, we have 
	\[
	\Pr \br{\mathcal{E}_{i+1}} \leq \Pr \br{\mathcal{E}_{i}}+ \Pr \br{\mathcal{E}'|\lnot \mathcal{E}_i} \leq \br{i-\LQVC}\cdot q+q=\br{i+1-\LQVC}\cdot q \enspace.
	\] 
	By the above argument, if $\lnot\mathcal{E}_{i+1}$ then $\mathcal{T}_{i+1}$ holds. Note that for $\LQVC < i \leq \Rtotal$, we have
	\[ \br{i-\LQVC}\cdot q = 2^{-\Theta(n\epsilon)} \enspace.
	\]
\end{proof}

\begin{lemma} \label{lem:Q-potential-increase}
	Assuming successful recycling throughout the execution of Algorithm~\ref{algo:MainalgorithmQMessage}, each iteration with no transmission error or hash collision increases the potential function $\Phi$ defined in Equation~\eqref{eqn:Q-phi} by at least $1$.
\end{lemma}
 \begin{proof}
 	Note that in an iteration with no error or hash collision Alice and Bob agree on the iteration type. Moreover, if $\Itertype=\MD, \RD$ or $\PD$, they also agree on whether they extend or rewind the data and if $\Itertype=\MES$ (Case ii), then exactly one of them is in sub-case A and the other one in sub-case B. We analyze the potential function in each case keeping in mind the hierarchy of the cases; e.g., Case ii or later cases are encountered only if Alice and Bob have full knowledge of each other's metadata. Lemma~\ref{lem:potential-value-vs-data} guarantees that $\Phi_{\MD}=0$ on entering Case ii, $\Phi_{\MD}=\Phi_{\RD}=0$ on entering Case vi and $\Phi_{\MD}=\Phi_{\RD}=\Phi_{\PD}=0$ on entering Case vii.
 	\begin{itemize}
 		\item Alice and Bob are in Case i.A:
 				\begin{itemize}
 					\item $\Phi_{\RD}$, $\Phi_{\PD}$ and $\Phi_{\mathrm{Q}}$ stay the same.
 					\item $i$ increases by $1$.
 					\item $\mdAplus$ and $\mdBplus$ stay the same.
 					\item None of $\mdAminus$ and $\mdBminus$ increases and at least one decreases by $1$.
 				\end{itemize}
 			Therefore, $\Phi_{\MD}$ decreases by at least $3-2=1$ and $\Phi$ increases by at least $1$.
 		\item Alice and Bob are in Case i.B:
 				\begin{itemize}
 					\item $\Phi_{\RD}$, $\Phi_{\PD}$ and $\Phi_{\mathrm{Q}}$ stay the same.
 					\item $i$ increases by $1$.
 					\item $\mdAminus$ and $\mdBminus$ stay at $0$.
 					\item  At least one of $\ellMA$ or $\ellMB$ is smaller than $i-1$; If only $\ellMA < i-1$, then $\mdAplus$ increases by $2$, and $\mdBplus$ by $1$. The case where only $\ellMB < i-1$ is similar. If both are smaller than $i-1$, then $\mdAplus$ and $\mdBplus$ both increase by $2$. 
 				\end{itemize}
 			Therefore, $\Phi_{\MD}$ decreases by at least $3-2=1$ and $\Phi$ increases by at least $1$.
 		\item Alice is in Case ii.A, Bob is in Case ii.B:
 				\begin{itemize}
 					\item $\Phi_{\MD}$ stays at $0$.
 					\item $rd^+$ and $rd^-$ stay the same.
 					\item $\ellQVCA$ stays the same and $\ellQVCB$ increases by $1$.
 					\item $q_\MA$ stays the same and $q_\MB$ increases by $1$.
 					\item $\pdAplus,\pdAminus,\pdBplus,\pdBminus$ stay the same.
 					\item $g$ stays the same, $b$ increases by at most $1$ and $u$ decreases by $1$.
 				\end{itemize}
 			Therefore, $\Phi_{\RD}$, $\Phi_{\PD}$ and $\Phi_\mathrm{Q}$ increase by $1$, $6$ and at least $8$, respectively. So $\Phi$ increases by at least $1$. 
 		\item Alice is in Case ii.B, Bob is in Case ii.A: This case is similar to the one above.
 		\item Alice and Bob are in Case iii:
 				\begin{itemize}
 					\item $\Phi_{\MD}$ stays at $0$.
 					\item $rd^+$, $\ellQVCA$ and $\ellQVCB$ stay the same.
 					\item $rd^-$ decreases by $1$.
 					\item $\pdAplus,\pdAminus,\pdBplus,\pdBminus$ stay the same.
 					\item None of $q_\MA$ and $q_\MB$ decreases. $q_\MA$ increases by $1$ if $\RA\Br{\ellRA}=\sS$. Similarly, $q_\MB$ increases by $1$ if $\RB\Br{\ellRB}=\sS$.
 					\item $\Phi_{\mathrm{Q}}$ stays the same.
 				\end{itemize}
 			Therefore, $\Phi_{\RD}$ decreases by $13$ and $\Phi_{\PD}$ increases by at most $12$. So $\Phi$ increases by at least $1$.
 		\item Alice and Bob are in Case iv:
 				\begin{itemize}
 					\item $\Phi_{\MD}$ stays at $0$.
 					\item $rd^+$, $\ellQVCA$ and $\ellQVCB$ stay the same.
 					\item $rd^-$ decreases by $1$.
 					\item $\pdAplus,\pdAminus,\pdBplus,\pdBminus$ stay the same.
 					\item None of $q_\MA$ and $q_\MB$ decreases. $q_\MA$ increases by $1$ if $\RA\Br{\ellRA}=\sS$. Similarly, $q_\MB$ increases by $1$ if $\RB\Br{\ellRB}=\sS$.
 					\item $\Phi_{\mathrm{Q}}$ stays the same.		
 				\end{itemize}
 		\item Alice and Bob are in Case v:
 				\begin{itemize}
 					\item $\Phi_{\MD}$ stays at $0$.
 					\item $rd^-$ stays at $0$ and $rd^+$ increases by $1$.
 					\item $\ellQVCA$ and $\ellQVCB$ stay the same.
 					\item $\Phi_{\PD}$ and $\Phi_{\mathrm{Q}}$ stay the same.
 				\end{itemize}
 			Therefore, $\Phi_{\RD}$ decreases by $2$ and $\Phi$ increases by $2$.
 		\item Alice and Bob are in Case vi.A:
 				\begin{itemize}
 					\item $\Phi_{\MD}$ and $\Phi_{\RD}$ stay at $0$.
 					\item $\pdAplus,\pdBplus,q_\MA,q_\MB$ stay the same.
 					\item None of $\pdAminus$ and $\pdBminus$ increases and at least one decreases by $1$.
 					\item $\Phi_{\mathrm{Q}}$ stays the same. 
 				\end{itemize}
 			Therefore, $\Phi_{\PD}$ decreases by at least $1$. So $\Phi$ increases by at least $1$.
 		\item Alice and Bob are in Case vi.B:
 				\begin{itemize}
 					\item $\Phi_{\MD}$ and $\Phi_{\RD}$ stay at $0$.
 					\item $q_\MA$, $q_\MB$ stay the same and $\pdAminus$, $\pdBminus$ stay at $0$.
 					\item At least one of the following holds:
 					$\ellPA < 6q_{\MA}\cdot r$, in which case $\pdAplus$ increases
 					by $1$ (otherwise it remains unchanged), or $\ellPB < 6q_{\MB}\cdot r$, and then $\pdBplus$ increases by $1$ (otherwise it remains unchanged).
 					\item $\Phi_{\mathrm{Q}}$ stays the same.
 				\end{itemize}
 			Therefore, $\Phi_{\PD}$ decreases by at least $1$. So $\Phi$ increases by at least $1$.
 		\item Alice and Bob are in Case vii:
 				\begin{itemize}
 					\item $\Phi_{\MD}$, $\Phi_{\RD}$ and $\Phi_{\PD}$ stay at $0$.
 					\item $u$ stays at $0$.
 					\item If $b\neq0$ then $g$ stays the same and $b$ decreases by $1$, otherwise, $b$ stays at $0$ and $g$ increases by $1$.
 				\end{itemize}
 			Therefore, $\Phi_{\mathrm{Q}}$ increases by $1$ and so does $\Phi$.
 	\end{itemize} 
 So assuming successful recycling throughout the execution of the algorithm, the potential function $\Phi$ increases by at least $1$ in every iteration with no transmission error or hash collision.
 \end{proof}

\begin{lemma} \label{lem:Q-potential-decrease}
	Assuming successful recycling throughout the execution of Algorithm~\ref{algo:MainalgorithmQMessage}, each iteration of the algorithm, regardless of the number of hash collisions and transmission errors, decreases the potential function $\Phi$ by at most $85$.
\end{lemma}

\begin{proof}
	In any iteration, $i$ increases by $1$, while $g$, $\mdAplus$, $\mdBplus$, $rd^+$, $\pdAplus$ and $\pdBplus$ decrease by at most $1$; $b$, $u$, $\ellQVCA$, $\ellQVCB$, $q_\MA$ and $q_\MB$ increase by at most $1$; $\mdAminus$ and $\mdBminus$ increase by at most $3$; $rd^-$ increases by at most $2$; and $\pdAminus$ and $\pdBminus$ increase by at most $4$. Hence, $\Phi_{\MD}$, $\Phi_{\RD}$, $\Phi_{\PD}$ increase by at most $22$, $30$ and $22$, respectively, and $\Phi_{\mathrm{Q}}$ decreases by at most $11$. So in total, $\Phi$ decreases by at most $85$.   
\end{proof}

Finally, we are ready to prove the main result of this section.

\setcounter{theorem}{0}
\begin{theorem}[\textbf{Restated}]
	Consider any $n$-round alternating communication protocol $\Pi$ in the plain quantum model, communicating messages over a noiseless channel with an alphabet $\Sigma$ of bit-size $\Theta\br{\log n}$. Algorithm~\ref{algo:MainalgorithmQMessage} is a quantum coding scheme which given $\Pi$, simulates it with probability at least $1-2^{-\Theta\br{n\epsilon}}$, over any fully adversarial error channel with alphabet $\Sigma$ and error rate $\epsilon$. The simulation uses $n\br{1+\Theta\br{\sqrt{\epsilon}}}$ rounds of communication, and therefore achieves a communication rate of $1-\Theta\br{\sqrt{\epsilon}}$.
\end{theorem}

\begin{proof}
	Let $\Rtotal=\ceil{\frac{n}{2r}}+86\br{c_1+c_3+c_5+2}n\epsilon$. By Lemma~\ref{lem:successful-recycling}, recycling is successful throughout the execution of the algorithm with probability at least $1-2^{-\Theta\br{n\epsilon}}$.  Assuming successful recycling, by Lemmas~\ref{lem:MD collisions},~\ref{lem:QH collisions} and~\ref{lem:PD collisions}, the total number of iterations with a hash collision is at most $c_1+c_3+c_5$ except with probability $2^{-\Theta\br{n\epsilon}}$. Since the number of iterations is less than $2n$, the total number of iterations with a transmission error is at most $2n\epsilon$. Therefore, by Lemma~\ref{lem:Q-potential-increase}, in the remaining $\Rtotal-\br{c_1+c_3+c_5+2}n\epsilon$ iterations the potential function $\Phi$ increases by at least $1$. The potential function decreases in an iteration only if a hash collision or a transmission error occurs and by Lemma~\ref{lem:Q-potential-decrease}, it decreases by at most $85$. So at the end of the simulation, we have 
	\[
	g-b-u \geq \Phi_{\mathrm{Q}} \geq \Phi \geq \Rtotal-\br{c_1+c_3+c_5+2}n\epsilon-85\br{c_1+c_3+c_5+2}n\epsilon \geq \frac{n}{2r}\enspace.
	\]
	Therefore the simulation is successful. The cost of entanglement distribution is~$\Theta\br{n\sqrt{\epsilon}}$. Moreover, the amount of communication in each iteration is independent of the iteration type and is always $\br{2r+\Theta(1)}$: in every iteration each party sends $\Theta(1)$ symbols to communicate the hash values and the value of the pointers in line $13$ of Algorithm~\ref{algo:MainalgorithmQMessage}; each party sends another $r$ symbols either in line $17$ of Algorithm~\ref{algo:MainalgorithmQMessage}, if $\Itertype\neq\SIM$ or in Algorithm~\ref{algo:Q-simulate}. Hence, we have
	\begin{align*}
	\text{Total number of communicated qudits} &=
	\Theta\br{n\sqrt{\epsilon}} + \Rtotal\cdot \br{2r+\Theta(1)} \\
	&= \Theta\br{n\sqrt{\epsilon}} + \br{\ceil{\frac{n}{2r} + \Theta\br{n\epsilon}}}\br{2r+\Theta(1)} \\
	&= n\br{1+\Theta\br{\sqrt{\epsilon}}}\enspace.    
	\end{align*}
	
\end{proof}

\section{Conclusion}
In this paper, we study efficient simulation of noiseless two-party interactive quantum communication via low noise channels.
For noise parameter $\epsilon$, a lower bound of $1-\Theta(\sqrt{\epsilon})$ on the communication rate is proved in the plain quantum model with large communication alphabets. To achieve this goal, we first study the teleportation-based model in which the parties have access to free entanglement and the communication is over a noisy classical channel. In this model, we show the same lower bound of $1-\Theta(\sqrt{\epsilon})$ in the large alphabet case. We adapt the framework developed for the teleportation-based model to the plain quantum model in which the parties do not have access to pre-shared entanglement and communicate over a noisy quantum channel. We show how quantum Vernam cipher can be used in the interactive communication setting to efficiently recycle and reuse entanglement, allowing us to simulate any input protocol with an overhead of only $1+\Theta(\sqrt{\epsilon})$.
In an upcoming paper, we will show how the same communication rate can be achieved when the communication alphabet is of constant size.

\section*{Acknowledgments}

D.~Leung's research supported in part by an NSERC Discovery grant and a CIFAR
research grant via the Quantum Information Science program; A.~Nayak's research supported in part by NSERC Canada; A.~Shayeghi's research supported in part by NSERC Canada and OGS; D.~Touchette's research supported in part by NSERC, partly via PDF program, CIFAR and by Industry Canada; P.~Yao's research is supported by the National Key R\&D Program of China 2018YFB1003202, National Natural Science Foundation of China (Grant No. 61972191), a China Youth 1000-Talent grant and Anhui Initiative in Quantum Information Technologies Grant No. AHY150100; N.~Yu's research supported in part by the Australian Research Council (Grant
No: DE180100156). Most of this project was done while D.~Touchette was a postdoctoral fellow at Institute for Quantum Computing (IQC) and Perimeter Institute for Theoretical Physics (PI). Part of the work was done while P.~Yao visited PI, and P.~Yao thanks PI for its hospitality. Part of the work was done while N. Yu visited IQC, and N.~Yu thanks IQC for its hospitality.
IQC and PI are supported in part by the Government of Canada and the Province of Ontario.

\bibliographystyle{alpha}
\bibliography{references}

\newcommand{\etalchar}[1]{$^{#1}$}
\begin{thebibliography}{BDSW96}

\bibitem[AA03]{aaronson:2003}
Scott Aaronson and Andris Ambainis.
\newblock Quantum search of spatial regions.
\newblock In {\em Foundations of Computer Science, 2003. Proceedings. 44th
  Annual IEEE Symposium on}, pages 200--209. IEEE, 2003.

\bibitem[ABY17]{BenYishaiSK:2017}
Young-Han~Kim Assaf Ben-Yishai, Ofer~Shayevitz.
\newblock Interactive coding for markovian protocols.
\newblock In {\em Proceedings of the 55th Annual Allerton Conference on
  Communication, Control, and Computing}, Allerton '17, page to appear, 2017.

\bibitem[Ari09]{Arikan:2009}
Erdal. Arikan.
\newblock Channel polarization: A method for constructing capacity-achieving
  codes for symmetric binary-input memoryless channels.
\newblock {\em IEEE Transactions on Information Theory}, 55(7):3051--3073, July
  2009.

\bibitem[BBJ{\etalchar{+}}94]{BBJMPSW94}
Charles~H Bennett, Gilles Brassard, Richard Jozsa, Dominic Mayers, Asher Peres,
  Benjamin Schumacher, and William~K Wootters.
\newblock Reduction of quantum entropy by reversible extraction of classical
  information.
\newblock {\em Journal of Modern Optics}, 41(12):2307--2314, 1994.

\bibitem[BCW98]{BCW98}
Harry Buhrman, Richard Cleve, and Avi Wigderson.
\newblock Quantum vs. classical communication and computation.
\newblock In {\em Proceedings of the 30th Annual ACM Symposium on Theory of
  Computing}, pages 63--68. ACM, 1998.

\bibitem[BDSW96]{BDSW96}
Charles~H Bennett, David~P DiVincenzo, John~A Smolin, and William~K Wootters.
\newblock Mixed-state entanglement and quantum error correction.
\newblock {\em Phys. Rev. A}, 54(5):3824--3851, 1996.

\bibitem[BE17]{BravermanK:2017}
Mark Braverman and Klim Efremenko.
\newblock List and unique coding for interactive communication in the presence
  of adversarial noise.
\newblock {\em SIAM Journal on Computing}, 46(1):388--428, 2017.

\bibitem[BGK{\etalchar{+}}15]{BGK+15}
Mark Braverman, Ankit Garg, Young~Kun Ko, Jieming Mao, and Dave Touchette.
\newblock Near-optimal bounds on bounded-round quantum communication complexity
  of disjointness.
\newblock In {\em Proceedings of the 2015 IEEE 56th Annual Symposium on
  Foundations of Computer Science (FOCS)}, FOCS '15, pages 773--791,
  Washington, DC, USA, 2015. IEEE Computer Society.

\bibitem[BK12]{BK12}
Zvika Brakerski and Yael~Tauman Kalai.
\newblock Efficient interactive coding against adversarial noise.
\newblock In {\em Proceedings of the 53rd Annual IEEE Symposium on Foundations
  of Computer Science}, pages 160--166. IEEE, 2012.

\bibitem[BKN14]{BrakerskiKN:2014}
Zvika Brakerski, Yael~Tauman Kalai, and Moni Naor.
\newblock Fast interactive coding against adversarial noise.
\newblock {\em J. ACM}, 61(6):35:1--35:30, December 2014.

\bibitem[BN13]{BN13}
Zvika Brakerski and Moni Naor.
\newblock Fast algorithms for interactive coding.
\newblock In {\em Proceedings of the 24th Annual ACM-SIAM Symposium on Discrete
  Algorithms}, pages 443--456. Society for Industrial and Applied Mathematics,
  2013.

\bibitem[BNT{\etalchar{+}}19]{BNTTU14}
Gilles Brassard, Ashwin Nayak, Alain Tapp, Dave Touchette, and Falk Unger.
\newblock Noisy interactive quantum communication.
\newblock {\em SIAM Journal on Computing}, 48(4):1147--1195, 2019.

\bibitem[Bom15]{Bombin15}
Hector Bombin.
\newblock Gauge color codes: optimal transversal gates and gauge fixing in
  topological stabilizer codes.
\newblock {\em New Journal of Physics}, 17(8):083002, 2015.

\bibitem[BR14]{BR11}
Mark Braverman and Anup Rao.
\newblock Toward coding for maximum errors in interactive communication.
\newblock {\em IEEE Trans. Inform. Theory}, 60(11):7248--7255, 2014.

\bibitem[CRSS97]{CRSS97}
A.~R. Calderbank, E.~M. Rains, P.~W. Shor, and N.~J.~A. Sloane.
\newblock Quantum error correction via codes over gf(4).
\newblock In {\em Proceedings of IEEE International Symposium on Information
  Theory}, pages 292--, Jun 1997.

\bibitem[CS96]{CS96-goodcode}
A.~R. Calderbank and Peter~W. Shor.
\newblock Good quantum error-correcting codes exist.
\newblock {\em Phys. Rev. A}, 54:1098--1105, Aug 1996.

\bibitem[Dev05]{Dev05}
Igor Devetak.
\newblock The private classical capacity and quantum capacity of a quantum
  channel.
\newblock {\em IEEE Trans. Inform. Theory}, 51(1):44--55, 2005.

\bibitem[Die82]{Dieks82}
DGBJ Dieks.
\newblock Communication by {EPR} devices.
\newblock {\em Phys. Lett. A}, 92(6):271--272, 1982.

\bibitem[DSS98]{DSS98}
David DiVincenzo, Peter Shor, and John Smolin.
\newblock Quantum-channel capacity of very noisy channels.
\newblock {\em Physical Review A}, 57(2):830--839, 1998.

\bibitem[EGH15]{EfremenkoGH:2015}
Klim Efremenko, Ran Gelles, and Bernhard Haeupler.
\newblock Maximal noise in interactive communication over erasure channels and
  channels with feedback.
\newblock In {\em Proceedings of the 2015 Conference on Innovations in
  Theoretical Computer Science}, ITCS '15, pages 11--20, New York, NY, USA,
  2015. ACM.

\bibitem[FGOS15]{FranklinGOS:2015}
M.~Franklin, R.~Gelles, R.~Ostrovsky, and L.~J. Schulman.
\newblock Optimal coding for streaming authentication and interactive
  communication.
\newblock {\em IEEE Transactions on Information Theory}, 61(1):133--145, Jan
  2015.

\bibitem[FM04]{FMa:2004}
Keqin Feng and Zhi Ma.
\newblock A finite gilbert-varshamov bound for pure stabilizer quantum codes.
\newblock {\em IEEE Transactions on Information Theory}, 50(12):3323--3325, Dec
  2004.

\bibitem[G{\etalchar{+}}17]{Gelles17}
Ran Gelles et~al.
\newblock Coding for interactive communication: A survey.
\newblock {\em Foundations and Trends{\textregistered} in Theoretical Computer
  Science}, 13(1--2):1--157, 2017.

\bibitem[GH14]{GH14}
Mohsen Ghaffari and Bernhard Haeupler.
\newblock Optimal error rates for interactive coding {ii}: Efficiency and list
  decoding.
\newblock In {\em Proceedings of the 55th Annual IEEE Symposium on Foundations
  of Computer Science}, pages 394--403. IEEE, 2014.

\bibitem[GHS14]{GhaffariHS:2014}
Mohsen Ghaffari, Bernhard Haeupler, and Madhu Sudan.
\newblock Optimal error rates for interactive coding i: Adaptivity and other
  settings.
\newblock In {\em Proceedings of the Forty-sixth Annual ACM Symposium on Theory
  of Computing}, STOC '14, pages 794--803, New York, NY, USA, 2014. ACM.

\bibitem[GMS11]{GMS12}
Ran Gelles, Ankur Moitra, and Amit Sahai.
\newblock Efficient and explicit coding for interactive communication.
\newblock In {\em Proceedings of the 52nd Annual IEEE Symposium on Foundations
  of Computer Science}, pages 768--777. IEEE, 2011.

\bibitem[GMS14]{GellesMS:2014}
R.~Gelles, A.~Moitra, and A.~Sahai.
\newblock Efficient coding for interactive communication.
\newblock {\em IEEE Transactions on Information Theory}, 60(3):1899--1913,
  March 2014.

\bibitem[Hae14]{Haeupler:2014}
Bernhard Haeupler.
\newblock Interactive channel capacity revisited.
\newblock In {\em Proceedings of the 2014 IEEE 55th Annual Symposium on
  Foundations of Computer Science}, FOCS '14, pages 226--235, Washington, DC,
  USA, 2014. IEEE Computer Society.

\bibitem[Has09]{Hastings2009}
Matthew~B Hastings.
\newblock Superadditivity of communication capacity using entangled inputs.
\newblock {\em Nature Physics}, 5(4):255, 2009.

\bibitem[HDW02]{HdW:2002}
Peter H{\o}yer and Ronald De~Wolf.
\newblock Improved quantum communication complexity bounds for disjointness and
  equality.
\newblock In {\em STACS}, pages 299--310. Springer, 2002.

\bibitem[Hol98]{Hol98}
Alexander~S Holevo.
\newblock The capacity of the quantum channel with general signal states.
\newblock {\em IEEE Trans. Inform. Theory}, 44(1):269--273, 1998.

\bibitem[HV17]{HaeuplerV:2017}
Bernhard Haeupler and Ameya Velingker.
\newblock Bridging the capacity gap between interactive and one-way
  communication.
\newblock In {\em Proceedings of the Twenty-Eighth Annual ACM-SIAM Symposium on
  Discrete Algorithms}, SODA '17, pages 2123--2142, Philadelphia, PA, USA,
  2017. Society for Industrial and Applied Mathematics.

\bibitem[JRS03]{JainRS:09}
Rahul Jain, Jaikumar Radhakrishnan, and Pranab Sen.
\newblock A lower bound for the bounded round quantum communication complexity
  of set disjointness.
\newblock In {\em Foundations of Computer Science, 2003. Proceedings. 44th
  Annual IEEE Symposium on}, pages 220--229. IEEE, 2003.

\bibitem[KNTZ07]{KNTZ07}
Hartmut Klauck, Ashwin Nayak, Amnon {Ta-Shma}, and David Zuckerman.
\newblock Interaction in quantum communication.
\newblock {\em IEEE Trans. Inform. Theory}, 53(6):1970--1982, 2007.

\bibitem[KR13]{KR13}
Gillat Kol and Ran Raz.
\newblock Interactive channel capacity.
\newblock In {\em Proceedings of the 45th Annual ACM Symposium on Theory of
  Computing}, pages 715--724. ACM, 2013.

\bibitem[Leu02]{Leung:2002}
Debbie~W. Leung.
\newblock Quantum vernam cipher.
\newblock {\em Quantum Info. Comput.}, 2(1):14--34, December 2002.

\bibitem[Llo97]{Lloyd97}
Seth Lloyd.
\newblock Capacity of the noisy quantum channel.
\newblock {\em Phys. Rev. A}, 55(3):1613--1622, 1997.

\bibitem[NC00]{NC00}
Michael~A Nielsen and Isaac~L Chuang.
\newblock {\em Quantum computation and quantum information}.
\newblock Cambridge University Press, Cambridge, UK, 2000.

\bibitem[NN93]{NaorNaor}
Joseph Naor and Moni Naor.
\newblock Small-bias probability spaces: Efficient constructions and
  applications.
\newblock {\em SIAM Journal on Computing}, 22(4):838--856, 1993.

\bibitem[Raz99]{Raz99}
Ran Raz.
\newblock Exponential separation of quantum and classical communication
  complexity.
\newblock In {\em Proceedings of the 31st Annual ACM Symposium on Theory of
  Computing}, pages 358--367. ACM, 1999.

\bibitem[RK11]{Regev:2011}
Oded Regev and Bo'az Klartag.
\newblock Quantum one-way communication can be exponentially stronger than
  classical communication.
\newblock In {\em Proceedings of the forty-third annual ACM symposium on Theory
  of computing}, pages 31--40. ACM, 2011.

\bibitem[Sch92]{Sch92}
Leonard~J Schulman.
\newblock Communication on noisy channels: A coding theorem for computation.
\newblock In {\em Proceedings of the 33rd Annual IEEE Symposium on Foundations
  of Computer Science}, pages 724--733. IEEE, 1992.

\bibitem[Sch93]{Sch93}
Leonard~J Schulman.
\newblock Deterministic coding for interactive communication.
\newblock In {\em Proceedings of the 25th Annual ACM Symposium on Theory of
  Computing}, pages 747--756. ACM, 1993.

\bibitem[Sch96]{Sch96}
Leonard~J Schulman.
\newblock Coding for interactive communication.
\newblock {\em IEEE Trans. Inform. Theory}, 42(6):1745--1756, 1996.

\bibitem[Sha48]{Shannon48a}
C.~E. Shannon.
\newblock A mathematical theory of communication.
\newblock {\em Bell System Tech. J.}, 27:379--423, 623--656, 1948.

\bibitem[Sho02]{Shor02}
Peter~W Shor.
\newblock The quantum channel capacity and coherent information.
\newblock Lecture notes, MSRI Workshop on Quantum Computation, 2002.

\bibitem[Sto02]{Stolte:2002}
Norbert Stolte.
\newblock {\em Rekursive Codes mit der Plotkin-Konstruktion und ihre
  Decodierung}.
\newblock PhD thesis, TU Darmstadt, Fachbereich Elektrotechnik und
  Informationstechnik,, 2002.

\bibitem[SW97]{SW97}
Benjamin Schumacher and Michael~D Westmoreland.
\newblock Sending classical information via noisy quantum channels.
\newblock {\em Phys. Rev. A}, 56(1):131--138, 1997.

\bibitem[SY08]{SY09}
Graeme Smith and Jon Yard.
\newblock Quantum communication with zero-capacity channels.
\newblock {\em Science}, 321(5897):1812--1815, 2008.

\bibitem[Wat16]{Wat08}
John Watrous.
\newblock Theory of quantum information.
\newblock Draft of a book, 2016.

\bibitem[Wil13]{Wilde11}
Mark~M Wilde.
\newblock {\em Quantum information theory}.
\newblock Cambridge University Press, Cambridge, UK, 2013.

\bibitem[WZ82]{WZ82}
William~K Wootters and Wojciech~H Zurek.
\newblock A single quantum cannot be cloned.
\newblock {\em Nature}, 299(5886):802--803, 1982.

\end{thebibliography}

\end{document}